\documentclass[ejs]{imsart}

\RequirePackage{amsthm,amsmath,amsfonts,amssymb}
\RequirePackage[authoryear]{natbib}
\RequirePackage[colorlinks,citecolor=blue,urlcolor=blue]{hyperref}
\usepackage{booktabs}
\RequirePackage{graphicx}
\usepackage{xcolor}
\usepackage{enumitem}
\usepackage{multirow}
\definecolor{ForestGreen}{RGB}{34,139,34}
\startlocaldefs

\usepackage{tikz}
\usetikzlibrary{decorations.pathreplacing}
\usepackage{dsfont}
\usepackage{xcolor}
\usepackage{placeins}
\usepackage{subcaption}
\usepackage{placeins}
\theoremstyle{plain}
\newtheorem{theo}{Theorem}
\newtheorem{lem}{Lemma}
\newtheorem{prop}{Proposition}

\newtheorem{cor}{Corollary}
\newtheorem{defi}{Definition}

\newtheorem{assumpCons}{Assumption}
\newtheorem{assumpId}{Assumption}

\newcommand{\Cor}[1]{\textcolor{black}{#1}}
\newcommand{\Coragain}[1]{\textcolor{black}{#1}}
\theoremstyle{remark}

\newcommand{\bY}{\boldsymbol{Y}}
\newcommand{\bpi}{\boldsymbol{\pi}}
\newcommand{\bmu}{\boldsymbol{\mu}}
\newcommand{\bsigma}{\boldsymbol{\sigma}}
\newcommand{\bT}{\boldsymbol{T}}
\newcommand{\btheta}{\boldsymbol{\theta}}
\newcommand{\bz}{\boldsymbol{z}}

\newcommand{\bx}{\boldsymbol{x}}
\newcommand{\bX}{\boldsymbol{X}}
\newcommand{\bxi}{\boldsymbol{\xi}}
\newcommand{\bw}{\boldsymbol{w}}

\newcommand{\zik}{z_{ik}}
\newcommand{\pik}{\pi_{k}}

\newcommand{\sigmah}{\widehat{\sigma}}
\newcommand{\muh}{\widehat{\mu}}
\newcommand{\Yibarr}[2]{\overline{\bY}_{i,#1:#2,r}}
\newcommand{\Yibarrc}[2]{\Yibarr{#1}{#2}^{(c)}}

\newcommand{\Prob}[1]{\mathbb{P}\left(#1\right)}
\newcommand{\Esp}[1]{\mathbb{E}\left[#1\right]}
\newcommand{\Espc}[2]{\mathbb{E}_{#1}\left[#2\right]}

\newcommand{\Scal}[2]{\left<#1,#2\right>} 

\newcommand{\bthetac}{\btheta^{(c)}}
\newcommand{\bTc}{\bT^{(c)}}
\newcommand{\Lk}{L_{k}}
\newcommand{\sik}{s_{ik}}
\newcommand{\sikc}{\sik^{(c)}}
\newcommand{\pikc}{\pi_k^{(c)}}
\newcommand{\muc}{\mu^{(c)}}
\newcommand{\sigmac}{\sigma^{(c)}}
\newcommand{\Tc}{T^{(c)}}
\newcommand{\pic}{\pi^{(c)}}
\newcommand{\Lc}{L^{(c)}}
\newcommand{\Deltakc}{\Delta_k^{(c)}}

\newcommand{\splusk}{s_{+k}}
\newcommand{\spluskc}{\splusk^{(c)}}

\newcommand{\setR}{\mathbb{R}}

\newcommand{\bzs}{\bz^{\star}}
\newcommand{\bTs}{\bT^{\star}}
\newcommand{\bthetas}{\btheta^{\star}}

\newcommand{\zs}{z^{\star}}
\newcommand{\Zs}{Z^{\star}}

\newcommand{\thetas}{\theta^{\star}}
\newcommand{\pis}{\pi^{\star}}

\newcommand{\RK}{\mathds{R}_K\left(\bzs,\bz\right)}
\newcommand{\NL}[2]{\mathds{N}_{(L_{#1}+1)\times (L_{#2}+1)}^{\left({#1},{#2}\right)}\left(\bTs,\bT\right)}
\newcommand{\NLk}{\NL{k}{k'}}
\newcommand{\mus}{\mu^{\star}}

\renewcommand{\muh}{\widehat{\mu}}

\newcommand{\bTheta}{\boldsymbol{\Theta}}
\newcommand{\bpis}{\bpi^\star}

\newcommand{\Fntz}{F_n\left(\btheta,\bT,\bz\right)}
\newcommand{\Gntz}{g_n\left(\btheta,\bT,\bz\right)}
\newcommand{\Lambdat}{\widetilde{\Lambda}_n}
\newcommand{\zsik}{\zs_{ik}}
\newcommand{\KL}[2]{\text{KL}\left(#1,#2\right)}

\newcommand{\bGamma}{\boldsymbol{\Gamma}}
\newcommand{\pih}[1]{\widehat{\bpi}\left(#1\right)}
\newcommand{\pisk}{\pis_k}

\endlocaldefs

\begin{document}

\begin{frontmatter}

\title{Mixture of segmentation for heterogeneous functional data }
\runtitle{Mixture of segmentation}

\begin{aug}

\author[A]{\fnms{Vincent} \snm{Brault}\ead[label=e1]{vincent.brault@univ-grenoble-alpes.fr}},
\author[B]{\fnms{\'{E}milie} \snm{Devijver}\ead[label=e2]{emilie.devivjer@univ-grenoble-alpes.fr}}
\and
\author[C]{\fnms{Charlotte} \snm{Laclau}\ead[label=e3]{charlotte.laclau@telecom-paris.fr}}

\address[A]{Univ. Grenoble Alpes, CNRS, Grenoble INP\footnote[1]{Institute of Engineering Univ. Grenoble Alpes}, LJK, 38000 Grenoble, France, \printead{e1}}
\address[B]{CNRS, Univ. Grenoble Alpes, Grenoble INP\footnotemark[1], LIG, 38000 Grenoble, France, \printead{e2}}
\address[C]{Télécom Paris, Institut Polytechnique de Paris, \printead{e3}}
\end{aug}

\begin{abstract}
In this paper, we consider functional data with heterogeneity in time and population. 
We propose a mixture model with segmentation of time to represent this heterogeneity while keeping the functional structure. 
The maximum likelihood estimator is considered and proved to be identifiable and consistent. 
In practice, an EM algorithm is used, combined with dynamic programming for the maximization step, to approximate the maximum likelihood estimator. 
The method is illustrated on a simulated dataset and used on a real dataset of electricity consumption. 
\end{abstract}


\end{frontmatter}
\tableofcontents

\section{Introduction}

Functional Data Analysis (FDA) deals with the theory and the exploration of data observed over a finite discrete grid and expressed as curves (or mathematical functions) varying over some continuum such as time. This type of data is commonly encountered in many fields, including economy \citep{bugni2009}, computational biology \citep{giacofci2013} or environmental sciences \citep{bouveyron:21}, to name a few.
\Cor{For an in-depth review of techniques and applications, we refer the interested readers to recent surveys \citep{Li2022,wang2016} and more exhaustive books \citep{FerratyView2006,ramsay2005,kokoszka2017}}.  
In many of these applications, such as \Cor{electricity load \citep{Fontana, devijver2020, Maturo2023}}, used for illustration here, we observe multiple curves corresponding to several individuals over a given time interval. As a result, one can expect a high heterogeneity of the data, both at the level of the studied individuals, that may correspond to different behavior or consumer profiles, but also on the time dimension where changes in power consumption regimes are likely to occur over one year for instance. 
To consider a parametric model, homogeneous data is required, both at population and time levels. In this paper, we propose to split the considered heterogeneous data into homogeneous clusters of individual curves, each of them being segmented over time into homogeneous regimes. To this end, we consider a mixture of segmentation over the projection of the curves onto \Cor{a wavelet basis, which retains a temporal aspect that is coherent with the segmentation}. Figure \ref{fig:motivation} serves to illustrate this objective. The top row represents the initial functional data consisting of 100 individuals (curves) observed over 50 days. The following rows allow us to visualize on the one hand the decomposition of the population into clusters (here 3 clusters - red yellow, purple), and on the other hand, within each cluster the segmentation obtained on the time dimension. Note that, in our case we allow each cluster to have a different segmentation, leading to a more flexible model. \Cor{In this example, we visualize the segmentation on the projection on a wavelet basis into $3$ dimensions.}

\begin{figure}[!htpb]
    \centering
    \includegraphics[width=\linewidth]{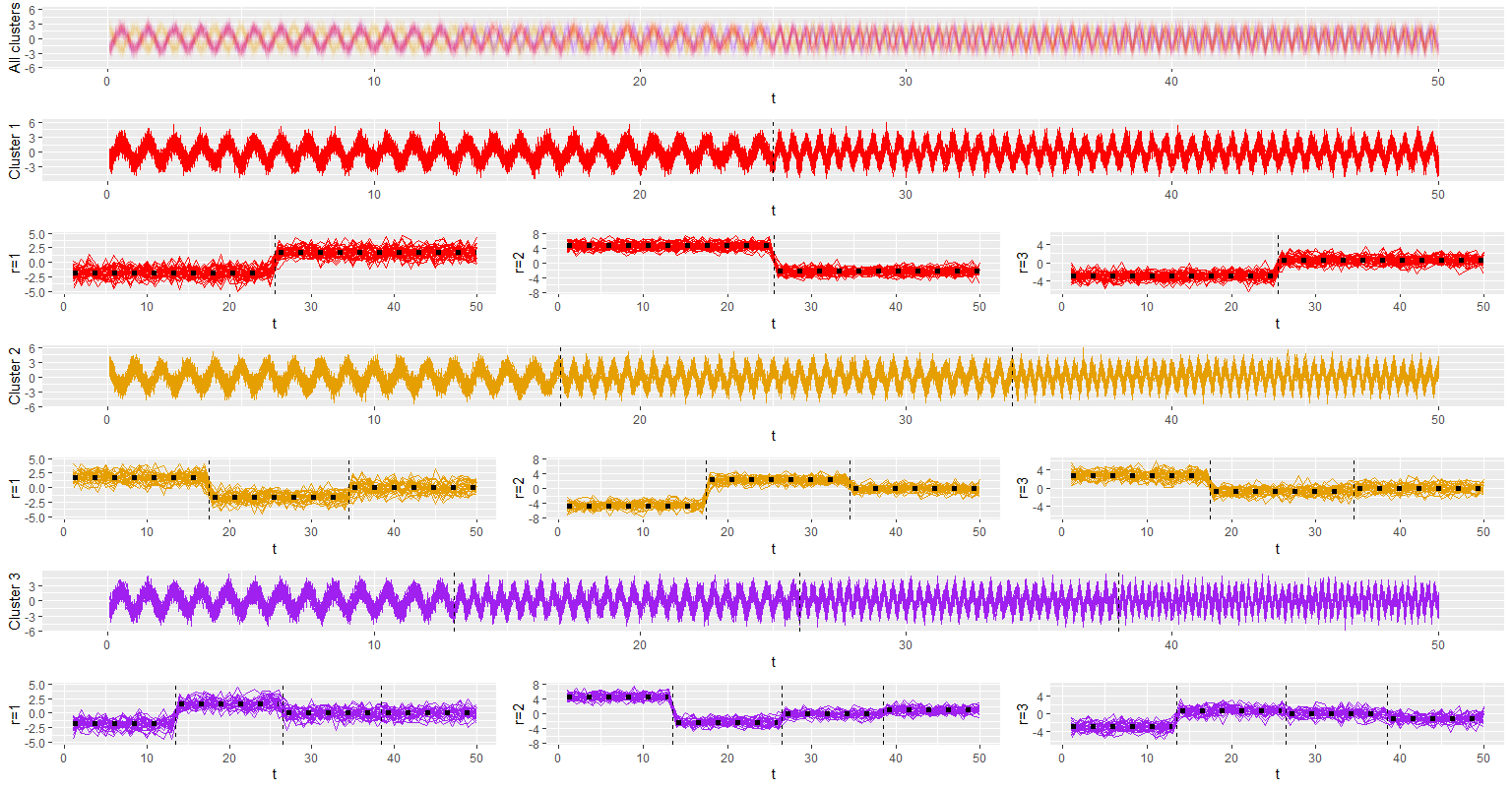}
    \caption{Motivation example. The top row represents the initial functional data consisting of 100 individuals (curves) observed every half-hour for 50 days. The following rows allow to visualize 1/ the decomposition of the population into clusters (here 3 clusters - red yellow, purple), 2/ within each cluster the segmentation obtained on the time dimension, 3/ the projection on a wavelet basis at several levels $r\Cor{\in\{1,2,p=3\}}$.}
    \label{fig:motivation}
\end{figure}

Model-based clustering approaches for functional data have been extensively studied in the literature (\cite{james2003_sparsefunc,Liu2009,bouveyron2011,jjacques2013,jjacques2014, devijver2017}). For the particular case of heterogeneous data that interests us in this article, one can broadly differentiate between methods that perform simultaneously clustering and segmentation (\cite{alon2003,Hebrail_2010,same2011,Same_2012,chamroukhi_2016}) and co-clustering based methods (\cite{bouveyron:2017, bouveyron:21, bouveyron:2021_2, GALVANI2021107219}). We provide further details on these two families of approaches hereafter and position our contributions with respect to the existing state of the art.


\paragraph*{Clustering and segmentation}


\cite{same2011} proposed to deal with heterogeneous time series by integrating the notion of change of regimes within a mixture of hidden logistic process regressions. The model considers two latent variables, one for the mixture component and one for the segmentation. 
Model selection is done through an adapted BIC criterion. However, while attempting to consider changes of regime, this approach fails to account for the ordering of observations, a key feature when dealing with functional data.  
\cite{Same_2012} extended this model for online segmentation of time series.
In an effort to account for these potential changes of regimes, another family of mixture models, namely the mixture of piecewise regression, has been proposed. 
\cite{Hebrail_2010} first define this notion of piecewise regression to analyze temporal data, by proposing a distance-based model that simultaneously performs clustering on the set of functional observations (through a Kmeans-like algorithm) and segmentation (in the form of piecewise constant function summarizing) within each of the obtained cluster. This work was further generalized to a more flexible probabilistic framework by \cite{chamroukhi_2016}, who designed a model based on a mixture of piecewise regression densities. The piecewise regression is modeled by a segmentation of polynomial functions, as a generalization of spline basis where knots have to be fixed. However, this sets a particular form within each segment. 



\paragraph*{Model-based co-clustering for FDA}
\cite{bouveyron:2017} proposed a co-clustering model to analyze multivariate functional data. They applied this model to analyze electricity consumption curves and found that due to the nature of the temporal data, the clustering over time points is close to a segmentation over time.
\cite{bouveyron:21} extend this method to multivariate time series (with several time series for each observation and each timepoint), using a sparse representation over principal components. In \cite{bouveyron:2021_2}, authors extend this co-clustering approach using a shape invariant model, allowing for translation in time, and translation and scaling in mean.
\cite{GALVANI2021107219} propose another bi-clustering algorithm for functional data while considering a potential misalignment through translation. While co-clustering based approach have proven efficient in this context, the clustering obtained on the time dimension do not account for the ordering of the observation.



\paragraph{Contributions and organization of the paper}
Our contribution is threefold, we propose (1) a method to study multivariate functional data, decomposing the population into homogeneous clusters and the time into homogeneous segments, where we ensure coherence on the time order; (2) we then focus on the theoretical study of the model (identifiability) and the estimation of the parameters (consistency), which is completely missing from all aforementioned related articles; \Cor{(3) we perform extension simulation experiments to study the behavior of our approach.}(4) finally, we study a real-world electricity consumption data set to illustrate the benefits of our method.

The paper is organized as follows. 
In Section \ref{sec:framework} the modeling framework is introduced together with the necessary notations. The identifiability of the model is obtained. Details about the estimation procedure are provided in Section \ref{sec:estimation}. The maximum likelihood estimator is proposed, approximated by an EM algorithm. The maximization step is solved by a dynamical programming. The consistency of the estimator is provided. The finite-sample performance of the proposed estimation method is investigated in Section \ref{sec:simu}. The methodology is finally used to analyse electricity consumption in Section \ref{sec:real_data}.  The paper concludes by some discussion in Section \ref{sec:conc}. The code is publicly available at \url{https://github.com/laclauc/MixtSegmentation}. All proofs are given in the Appendix. 

\section{The model and its identifiability}\label{sec:framework}

\paragraph{Notations}
\Cor{Random variables are indicated by uppercase letters and observations by lowercase letters. Then for matrices and vectors we use bold type (such as $\mathbf{A}$, $\mathbf{B}$) while scalars are denoted without the bold formatting (such as $A$). To indicate their dimension, we use notations like $\mathbf{A}\in\mathbb{R}^n$ or $\mathbf{B}\in\mathbb{R}^{n\times d}$. The i(,j)-th element of $\mathbf{A}$ is written as $\mathbf{A}_i$ or $\mathbf{B}_{ij}$}. 
In the following, we suppose one observes multivariate individual  curves ${X}_1(t), \ldots, {X}_n(t)$ on discrete timepoints $t_1,\ldots, t_d$. 
First, we introduce the various elements of the modeling framework, and provide the identifiability of the model. The proof of identifiability can be found in Appendix \ref{App:identifiability}.

\subsection{A multivariate functional model with segments in time and clusters in  population}
We observe multivariate individual curves $({X}_{ih}(t_j))_{1\leq i \leq n, 1\leq j \leq d, 1\leq h \leq H}$ of dimension $H$ over $d$ timepoints  and within a population of size $n$.
The heterogeneous population is studied through a mixture model of $K$ clusters, encoded indifferently in its binary form, $z_{ik}=1$ if and only if the curve $i$ belongs to the cluster $k$, and its vector form, $z_{i}=k$ if and only if the curve $i$ belongs to the cluster $k$, for $1 \leq k \leq K$ and $1\leq i \leq n$. Each observation belongs to the cluster $k\in\{1,\ldots, K\}$ with probability $\pi_k \in [0,1]$.
The heterogeneity in time is represented through $L_k +1$ segments $(I_{k \ell})_{ 0\leq \ell \leq L_K}$: if $z_{ik}=1$ and $j \in I_{k \ell}$, encoded by $w_{j\ell} = 1$,
\begin{align}
    \label{model_fda}
    {X}_{ih}(t_j) = f_{k \ell h}(t_j) + \eta_{ijh}
\end{align}
with $\eta_{ijh}$ corresponds to some random noise, more details being given in Section \ref{sec:proj}.
Usually in segmentation, we assume that the signal is constant. Here, we would like to emphasize some coherence in time \Cor{given by the underlying function $f_{k\ell h}$}, but not necessarily through a strong assumption as constant. Then, we propose to decompose our signal into several time periods that are meaningful in practice (in hours, in days, in weeks depending on the application), and to have the same function $f_{k \ell h}$ within the considered interval, through the same segment. 

The modeling assumption is equivalent to a main function $f_{k \ell h}$ for the $h$th component, for individuals belonging to the cluster $k$, and for a timepoint in the $\ell$th segment. This means that within a segment and a cluster, there is a random variation (seen as a noise) independent and identically distributed over each component of the multivariate curve. 

\subsection{Projection onto a functional basis and matrix-variate model}\label{sec:proj}

We denote $(\boldsymbol{Y}_{ij}) \in \mathbb{R}^{p}$ the  coefficient decomposition vectors of the  component $j \in \{1,\ldots, d\}$ onto the functional basis, and the individual $i \in \{1,\ldots, n\}$, and the orthonormal characterization leads to, for the level $M$,
$$(\mathbf{X}_{i.}(t_j))_{1\leq j \leq d}= \Pi \boldsymbol{Y}_{ij} ;$$
where $\Pi$ is a matrix defined by  the functional basis of size $M$. \Cor{We focus on the wavelet family, among which one can choose between several wavelets, such as Haar, Daubechies to cite a few. \Coragain{Selecting an appropriate basis presents inherent challenges, although our paper doesn't delve extensively into this aspect. Rather, our emphasis lies in investigating the influence of a complex model on coefficients to capture a range of variations. Our approach prioritizes a nuanced understanding over exhaustive testing of various bases.
The localized assumption inherent in wavelets aligns seamlessly with the requirement for constancy in segmentation. The intentional choice of the Haar basis stems from its alignment with our research goals, demonstrated effectiveness in experiments, and its ability to maintain simplicity and interpretability.}
Wavelets are particularly well suited for the segmentation, as they are ordered and keep a temporal aspect. In Figure \ref{fig:wavelet}, the Haar basis employed in the experimental section is depicted, alongside reconstructions of coefficients at different levels on a simulated curve. Every coefficient carries meaningful information, showcasing the capacity of wavelets to broaden the spectrum of available choices.}
\begin{figure}[t]
   \includegraphics[scale=0.25]{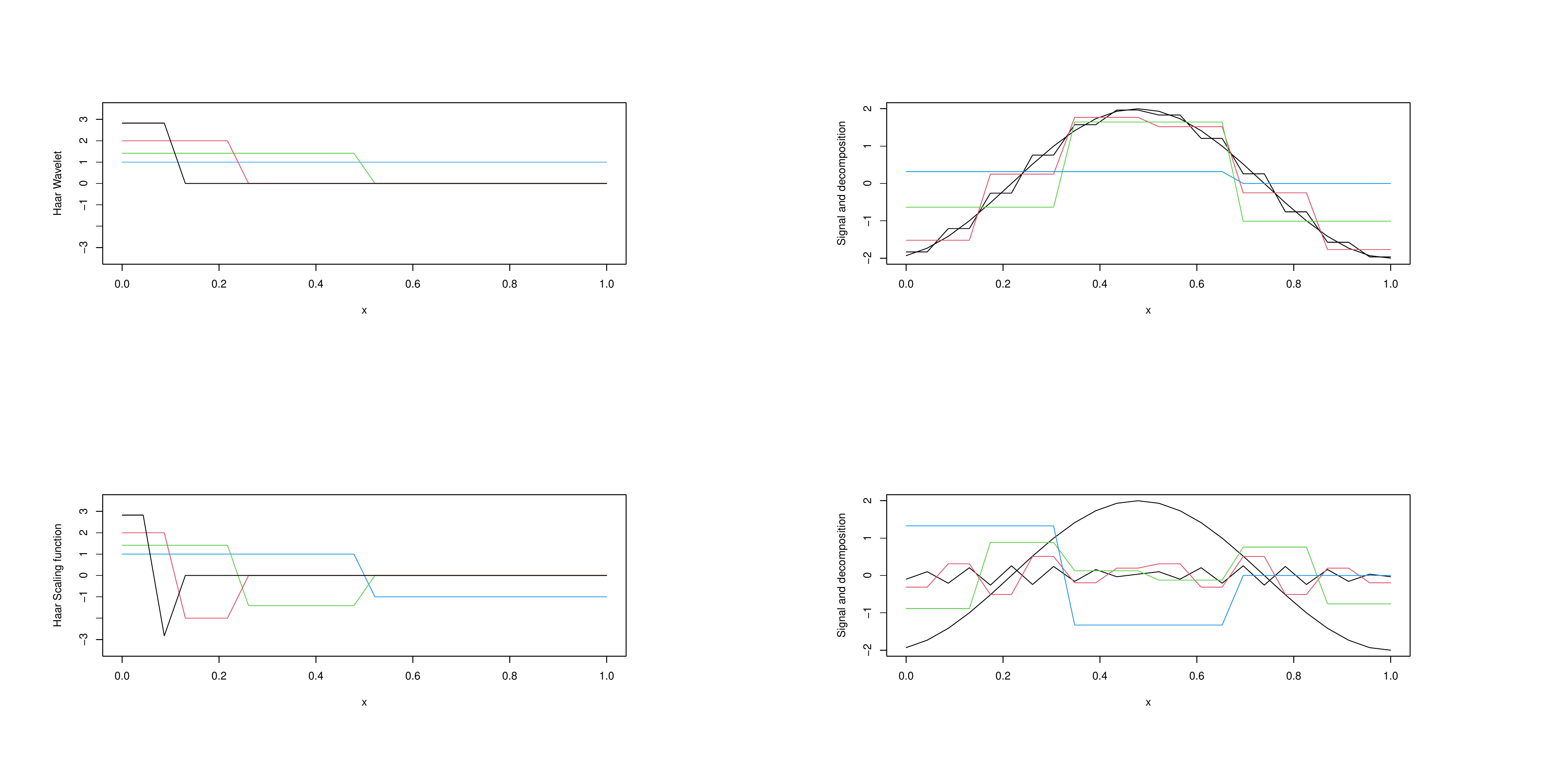}
   \caption{Illustration of the Haar basis. Top left: the Haar wavelet  and its scaled version. Bottom left: the Haar scaling function and its scaled version. Top right: reconstruction of the Haar coefficients at several levels. We focus in the analysis on the level 3, displayed in green. Bottom right: the residuals of the decomposition. } \label{fig:wavelet}
\end{figure}
We consider the wavelet coefficient dataset $(\mathbf{Y}_{i})_{1\leq i \leq n} = (\mathbf Y_{i.})_{1\leq i \leq n} \in (\mathbb{R}^{d \times p})^n$, which defines $n$ observations whose probability distribution is modeled by the following finite matrix-variate Gaussian mixture of segmentation model. 
As mentioned previously, the heterogeneous population is represented by $K$ clusters. For the cluster $k \in \{1,\ldots, K\}$, the heterogeneity in time is described by $L_k+1$ segments, defined by $\Lk$ break-points $T_{k0}<T_{k1}<\cdots<T_{k\Lk}<T_{k,\Lk+1}$. 
Then, for an observation $i\in \{1,\ldots,n\}$  in the cluster $k\in \{1,\ldots,K\}$, 
 for  $\ell\in\left\{0,\ldots,\Lk\right\}$ such that  $j\in\left\{T_{k, \ell}+1,T_{k, \ell}+2,\ldots,T_{k,\ell+1}\right\}$, we have:
\begin{align}
[\mathbf{Y}_{i}]_{j.} | (Z_{ik} = 1, W_{j\ell} = 1) = \bmu_{k \ell} + \varepsilon_{ij}
\label{model_multiv}
\end{align}
with $\bmu_{k \ell}\in\mathbb{R}^{\Cor{p}}$, $\varepsilon_{ij} \sim \mathcal{N}_p\left(0, \Sigma_{k\ell}\right)$ where $\Sigma_{k\ell}$ is diagonal with the values $\left(\sigma_{k \ell r}\right)_{1\leq r \leq p}$. 
\subsection{Identifiability of the  model}
In this section, we first establish the identifiability of the multivariate model \eqref{model_multiv}. 

\begin{theo}[Identifiability of \eqref{model_multiv}]\label{Th:Identifiability} Assume that:
\begin{enumerate}[label=(ID.\arabic*)]
    \item \label{ID:1} For every $k\in\{1,\ldots,K\}$ and $\ell\in\{0,\ldots,\Lk\}$, there exists at least one $r\in\{1,\ldots,p\}$ such that $\sigma_{k \ell r}\neq\sigma_{k,\ell+1,r}$ or $\mu_{k \ell r}\neq\mu_{k,\ell+1,r}$.
    \item \label{ID:2} We have:
    \[p \geq \underset{k\in\{1,\ldots,K\}}{\max}\Lk+1.\]
     \item \label{ID:3} If there exists $k\neq k'$ such that $\Lk=L_{k'}$ then:
    \begin{itemize}
        \item there exists $\ell\in\{0,\ldots,\Lk\}$ such that $T_{k \ell}\neq T_{k',\ell}$,
        \item or there exists $\ell\in\{0,\ldots,\Lk\}$ and $r\in\{1,\ldots,p\}$ such that:
    \[\sigma_{k \ell r}\neq\sigma_{k',\ell,r} \text{ or }\mu_{k \ell r}\neq\mu_{k',\ell,r}.\]
    \end{itemize}
    \item \label{ID.4} For every $k\in\{1,\ldots,K\}$, $\pi_k>0$.
\end{enumerate}
Under these assumptions, the model \eqref{model_multiv} is identifiable.
\end{theo}
The breakpoints models of each cluster are identifiable by the\linebreak assumption~(ID.1) and~(ID.2) (as introduced in \cite{lebarbier2005detecting,harchaoui2010multiple,brault2017efficient}). The  assumption~(ID.3) allows to differentiate the clusters (see \cite{droesbeke2013modeles}). 

Mixture models are known to be identifiable up to a label switching: two partitions can be the same while the cluster labels being reversed (see \cite{stephens2000dealing}, \cite[Chapter 10]{mengersen2011mixtures} and \cite{robert2021}). In this model, a natural order is to choose the labeling of each cluster such that
\[k\leq k'\Leftrightarrow L_k\leq L_{k'}.\]
This alleviates the problem of label switching; and it can be completely removed when the $(L_k)_{1\leq k\leq K}$ are all different.

\section{Estimation}\label{sec:estimation}
In this paper, we assume that $K$ the number of clusters is known, as well as the number of segment within each cluster $(L_k)_{1\leq  k \leq K}$.
\subsection{Maximum Likelihood Estimation}
Using the model \eqref{model_multiv}, under identifiability, by noting $\bT$ the set of the break points and $\btheta = ((\mu_{k\ell r},\sigma_{k\ell r})_{1\leq k \leq K,  0\leq \ell \leq L_K, 1\leq r \leq p}, (\pi_k)_{1\leq k \leq K})$ the set of parameters, we obtain the following likelihood: 
\[\text{lik}\left(\bY;K,\bT,\btheta\right)=\prod_{i=1}^n\sum_{k=1}^K\pi_k\prod_{\ell=0}^{\Lk}\prod_{j=T_{k \ell}+1}^{T_{k,\ell+1}}\prod_{r=1}^p\left[\frac{1}{\sqrt{2\pi\sigma_{k \ell r}}}e^{-\frac{1}{2\sigma_{k \ell r}}\left(Y_{ij r}-\mu_{k \ell r}\right)^2}\right].\]
The mixture model leads to  the product over individuals $i \in \{1,\ldots,n\}$ and the sum over the clusters $k \in \{1,\ldots,K\}$ while  the segmentation is related to the product over each segment $\ell \in \{0,\ldots, L_k\}$ and timepoints indexed by $j \in \{T_{k\ell}+1,\ldots, T_{k,\ell+1}\}$, for the cluster $k \in \{1,\ldots, K\}$. In addition to the parameters $K$, $\bT$ and $\btheta$, we search to estimate the partition $\bz$.  

We denote $\hat{\btheta}$ the maximum likelihood estimator. 
\subsection{EM algorithm}
Considering a mixture model, we use the \textit{Expectation Maximisation} (EM) algorithm~(\cite{dempster1977maximum}) to estimate the parameters. The principle is, for the step $(c)$, to fix a parameter $\bthetac$ and break points $\bT$, and to maximize the following function:
\[(\btheta,\bT)\mapsto Q\left(\btheta,\bT\left|\bthetac,\bTc\right.\right)=\Espc{\bz\left|\bY;\bthetac,\bTc \right. }{\log p\left(\bY,\bz;\btheta,\bT\right)}\]
where $p(\bY,\bz;\btheta,\bT)$ is the joint distribution of $\bY$ and $\bz$ for fixed parameters $\btheta$ and break points $\bT$. To do so, we alternate between  two steps:
\begin{itemize}
    \item[E-step:] For $i\in\{1,\ldots,n\}, k\in\{1,\ldots,K\}$ compute  the values $\sikc$
    defined by:
\[\sikc=\Prob{\zik=1\left|\bY;\bthetac,\bTc\right.};\]
    \item[M-step:] Maximization with respect to $\btheta$ and $\bT$ the function \linebreak$(\btheta,\bT)\mapsto Q\left(\btheta,\bT\left|\bthetac,\bTc\right.\right)$.
\end{itemize}

For $i\in\{1,\ldots,n\}, k\in\{1,\ldots,K\}$, the computation of $\sikc$ is explicit using the following proposition.

\begin{prop}[E-Step]\label{prop:stepE}
For every $i\in\{1,\ldots,n\}$ and $k\in\{1,\ldots,K\}$ we have:
\begin{align*}
    \sikc=\frac{\pic_{k}\prod_{\ell=0}^{\Lc_{k}}\prod_{j=\Tc_{k \ell}+1}^{\Tc_{k,\ell+1}}\prod_{r=1}^p\left[\frac{1}{\sqrt{2\pi\sigmac_{k \ell r}}}e^{-\frac{1}{2\sigmac_{k \ell r}}\left(Y_{i\ell r}-\muc_{k \ell r}\right)^2}\right]}{\sum_{k'=1}^K\pic_{k'}\prod_{\ell=0}^{\Lc_{k'}}\prod_{j=\Tc_{k'\ell}+1}^{\Tc_{k',\ell+1}}\prod_{r=1}^p\left[\frac{1}{\sqrt{2\pi\sigmac_{k'\ell r}}}e^{-\frac{1}{2\sigmac_{k'\ell r}}\left(Y_{i\ell r}-\muc_{k'\ell r}\right)^2}\right]}.
\end{align*}
\end{prop}

For the maximization, the problem is more difficult due to the unknown segmentation over time. The next proposition explicit the formulae for the proportions $(\pi_k)_{1\leq k \leq K}$.
\begin{prop}[Proportion in the M-Step]\label{prop:stepM}
For every $k\in\{1,\ldots,K\}$ we have:
\[\pikc=\frac{\spluskc}{n}=\frac{\sum_{i=1}^n\sikc}{n}.\]
\end{prop}

The other parameters $(\mu_{k\ell r},\sigma_{k\ell r})_{1\leq k \leq K, 1\leq \ell \leq L_K, 1\leq r \leq p}$ are given using the dynamic programming  \cite[see for example][]{bellman1957dynamic,kay1993fundamentals}.
We start by observing that
\begin{align}
\underset{\left(\btheta,\bT\right)}{\max}\;Q\left(\btheta,\bT\left|\bthetac,\bTc\right.\right)=&-\frac{1}{2}\underset{k=1}{\overset{K}{\sum}}\left[\underset{ T_{k.} \in \mathcal{T}}{\min}\underset{\ell=0}{\overset{\Lk}{\sum}}\Deltakc\left(T_{k \ell};T_{k,\ell+1}\right)\right]\nonumber\\
&-\frac{ndp}{2}\log\left(2\pi\right)+\underset{k=1}{\overset{K}{\sum}}\spluskc\log\pik,
\label{Eq:maxQ}
\end{align}
where $\mathcal T = \{0=T_{k0}<T_{k1}<\cdots<T_{k,\Lk+1}=d+1\}$
with for every $k\in\{1,\ldots,K\}$, $0\leq t_1<t_2\leq d$,
\begin{eqnarray}
&&\!\!\!\!\!\!\!\!\!\!\!\!\!\!\Deltakc\left(t_1;t_2\right)\nonumber\\
&=&\underset{\overset{\bmu \in \mathbb{R}^p,}{\underset{\bsigma\in (\mathbb{R}_+^*)^p}{}}}{\min}\underset{r=1}{\overset{p}{\sum}}\left[\spluskc(t_2-t_1)\log\left(\sigma_r\right)+\frac{1}{\sigma_r}\underset{j=t_1+1}{\overset{t_2}{\sum}}\underset{i=1}{\overset{n}{\sum}}\sikc\left(Y_{ijr}-\mu_r\right)^2\right].\label{Eq:Deltakc}
\end{eqnarray}
This optimization problem is explicitly solved in the following proposition, for each cluster independently. 
\begin{prop}[Form of $\Deltakc$]For all $k\in\{1,\ldots,K\}$ and $0\leq t_1<t_2\leq d$, we have:
\begin{eqnarray*}
&&\!\!\!\!\!\!\!\!\!\!\!\!\!\!\Deltakc\left(t_1;t_2\right)\\
&=&\underset{r=1}{\overset{p}{\sum}}\left[\spluskc(t_2-t_1)\log\left(\sigmah_{k t_1 r}\right)+\frac{1}{\sigmah_{k t_1 r}}\underset{j=t_1+1}{\overset{t_2}{\sum}}\underset{i=1}{\overset{n}{\sum}}\sikc\left(Y_{ijr}-\muh_{k t_1 r}\right)^2\right]
\end{eqnarray*}
with
\begin{eqnarray*}
\overline{\mathbf{Y}}_{i,t_1:t_2,r} &=& \frac{1}{t_2-t_1}\sum_{j=t_1+1}^{t_2}Y_{ijr}\\
\muh_{k t_1 r}&=&\sum_{i=1}^n\frac{\sikc}{\spluskc} \overline{\mathbf{Y}}_{i,t_1:t_2,r}\\
\sigmah_{k t_1 r}&=&\sum_{i=1}^n\frac{\sikc}{\spluskc}. \frac{1}{t_2-t_1}\sum_{j=t_1+1}^{t_2}\left(Y_{ijr}-\overline{\mathbf{Y}}_{i,t_1:t_2,r}\right)^2.
\end{eqnarray*}
\end{prop}

To optimize the computation time, we suggest to first compute all the means $\Yibarrc{t_1}{t_2}$ for $t_2>t_1$. In particular, in the case of $\bsigma$ depends only on the cluster $k$, we improve the complexity.

\begin{prop}[Complexity]\label{prop:complexity}
The complexity of the \textit{EM} algorithm with $N_{\text{algo}}$ iterations is $\mathcal{O}\left[pnd^2\max\left(d,KN_{\text{algo}}\right)\right]$.
\end{prop}

\begin{proof}
 The values $\overline{\mathbf{Y}}_{i,t_1:t_2,r}$ are computed for each $i$, each $r$ and each pair $t_1<t_2$: as the computation is a mean, the final complexity is $\mathcal{O}\left(pnd^3\right)$. For each iteration in the algorithm, the complexity of E-step is $\mathcal{O}\left(Kpnd\right)$ and, thanks the storage of the $\overline{\mathbf{Y}}_{i,t_1:t_2,r}$, the complexity of the computation of each matrix $\Deltakc$ is only $\mathcal{O}\left(npd^2\right)$. Finally, the update of the estimation of $\bmu$ and $\bsigma$ is $\mathcal{O}\left(Knpd\right)$. The combination of all these complexities gives the result.
\end{proof}

Remark that the two most time-consuming steps are the computation of the means $\overline{\mathbf{Y}}_{i,t_1:t_2,r}$ and the computations of $\Deltakc$ at each iteration; these two steps can be easily parallelized. Moreover, the table of averages can be stored for later use.

\subsection{Consistency}
In this section, we prove that the model introduced in Equation \eqref{model_multiv} is consistent. To simplify notations, we consider univariate functional data or  projection of the observed functions onto a 1-dimensional basis, such that $p=1$ in this section, but the conclusion would be the same. To simplify the notations, we set $\sigma_{k\ell}=1$, but the results can be extended as well to any variance.

First, we assume that the parameter space is bounded, as introduced in \cite{bickel2013asymptotic,brault2020consistency,mariadassou2020consistency}.

\begin{assumpCons}\label{ass:C.1}
There exists $M>0$ such that for all $k\in\{1,\ldots,K\}$ and $\ell\in\{0,\ldots,\Lk\}$,
\[\mu_{k\ell}\in[-M;M].\]
\end{assumpCons}

We  assume that there are enough observations in each segment:
\begin{assumpCons}\label{ass:C.2} There exists $\tau_{\min}>0$ such that for all $k\in\{1,\ldots,K\}$ and $\ell\in\{0,\ldots,\Lk\}$, 
\[T_{k,\ell+1}-T_{k\ell}>\tau_{\min}d.\]
\end{assumpCons}
Let $\mathcal{T}_{\tau_{\min}}$  the set of breakpoints satisfying Assumption \eqref{ass:C.2}. 

We need an assumption about the rate of convergence. 

\begin{assumpCons} \label{ass:C.3} We assume that $\log(n) / d \underset{n,d\to+\infty}{\longrightarrow}0$. 
\end{assumpCons}

We also want to distinguish between clusters. To do so, we introduce the notion of equivalent clusters and the related symmetry. 

\begin{defi}[Equivalent clusters]
Two partitions $\bzs$ and  $\bz$ with $K$ clusters are  \emph{equivalent}\index{partition@Partition!équivalente}, denoted $\bzs \sim \bz$, if there exists a permutation 
 $\sigma\in\mathfrak{S}(\left\{1,\ldots,K\right\})$ such that for all $i\in\{1,\ldots,n\}$ and $k\in\{1,\ldots,K\}$, $z_{i\sigma(k)}=\zs_{ik}$.

 By similarity, we denote $\btheta', \bT' \sim \btheta, \bT$ if there exists a permutation that leads to the same parameters. 
%
%
\end{defi}

We can thus define a distance between two partitions. 
\begin{defi}[Distance for partitions]
We define the distance $d_{0,\sim}$ between two partitions $\bz, \bzs$ with $K$ clusters by 
\begin{align}\label{d0sim}
d_{0,\sim}(\bzs,\bz)&=\underset{\sigma\in\mathfrak{S}(\left\{1,\ldots,K\right\})}{\min}\sum_{k=1}^K\sum_{i=1}^n \zs_{ik} z_{i\sigma(k)}.
\end{align}

For a partition $\bzs$ and a radius $r>0$, and $\mathcal{Z}$ is the set of all potential partitions,
let $\mathcal{B}\left(\bzs;r\right)$ the ball
\[\mathcal{B}\left(\bzs;r\right)=\left\{\bz\in\mathcal{Z}\left|d_{0,\sim}(\bz,\bzs)\leq rn\right.\right\}.\]
\end{defi}

For the consistency of the estimator, we need to distinguish between close clusters, assuming something stronger than Assumption \ref{ID:3}.
\begin{assumpId} \label{ID:3:s} If there exists $k\neq k'$ such that $\Lk=L_{k'}$ then we assume that there exists at least $\tau_{\min}d$ coordinates $j$ such that the distribution of $Y_{ij}|\zs_{ik}=1$ is different from the distribution of $ Y_{ij}|\zs_{ik'}=1$.
\end{assumpId}
This is needed when the models within each segment are the same, and only the segments are different. 

We also need an assumption stronger than \ref{ID.4} about the number of curves in each cluster:
\begin{assumpId} \label{ID:4:s}   There exists a constant $c>0$ such that for every $k\in\{1,\ldots,K\}$, $\pi_k>c$.
\end{assumpId}


Variances are supposed to be equal whatever the cluster, the segment and the dimension, and to have the value $\sigma_{kjr}=1$. Extension to any variance is straightforward but derivations of formula are more technical. 

\begin{theo}\label{th:consistence}
Let $\mathbf{Y}$ be a matrix of a $n\times T$ observations of the model \eqref{model_multiv} with true parameter $\bthetas, \bTs$ where the number of clusters $K$ and the number of segments $(L_k)_{1\leq k \leq,K}$ are known. 
We assume \ref{ID:1}, \eqref{ID:3:s}, \eqref{ID:4:s}, \eqref{ass:C.1}, \eqref{ass:C.2} and \eqref{ass:C.3}. Then, 
for every parameter $\btheta\in\bTheta$ and $\bT\in\mathcal{T}_{\tau_{\min}}$,
\[\frac{p(\bY;\btheta,\bT)}{p(\bY;\bthetas,\bTs)}=\max_{(\btheta',\bT')\sim(\btheta,\bT)}\frac{p\left(\bY,\bzs;\btheta',\bT'\right)}{p\left(\bY,\bzs;\bthetas,\bTs\right)}\left[1+o_P(1)\right]+o_P(1)\]
where $o_P$ are uniform over $\bTheta$ and $\mathcal{T}_{\tau_{\min}}$, and $(\btheta',\bT')\sim(\btheta,\bT)$ means for any parameter up to the label switching. 
\end{theo}

\textbf{Sketch of proof.}
We will prove that the complete likelihood with a bad clustering becomes small asymptotically with respect to the complete likelihood associated to the true partition.
To do so, we decompose the probability with respect to potential partitions.
\begin{align*}
p(\bY;\btheta,\bT)&=\sum_{\bz \in \mathcal{Z}}p(\bY,\bz;\btheta,\bT)\\
&= \sum_{\bz \sim \bzs}p(\bY,\bz;\btheta,\bT) + \!\!\!\!\!\!
\sum_{\underset{\bz\nsim\bzs}{\bz\in\mathcal{B}\left(\bzs;{c}\right)}}p(\bY,\bz;\btheta,\bT) +\!\!\!\!\!\!
\sum_{\bz\notin\mathcal{B}\left(\bzs;{c}\right)}p(\bY,\bz;\btheta,\bT).
\end{align*}
Each term is controlled by the following propositions, that are proved in Appendix \ref{app:Cons}.

\begin{prop}[Equivalent partitions]\label{prop:part_equi}
Under Assumptions (ID.1) and (ID.3), we have for all $\btheta\in\bTheta$ and $\bT\in\mathcal{T}_{\tau_{\min}}$:
\[\sum_{\bz\sim\bzs}\frac{p(\bY,\bz;\btheta,\bT)}{p\left(\bY,\bzs;\bthetas,\bTs\right)}=\max_{(\btheta',\bT')\sim(\btheta,\bT)}\frac{p\left(\bY,\bzs;\btheta',\bT'\right)}{p\left(\bY,\bzs;\bthetas,\bTs\right)}\left[1+o_P(1)\right]\]
where $o_P$ is uniform on $\bTheta$.
\end{prop}

    \begin{prop}[Partitions close to the true one]\label{prop:contrib_local}
Under Assumptions (ID.3.s), (ID.4.s), (C.1), (C.2) and (C.3), we have that for all  $\tilde{c}<c/4$:
\[\sup_{(\btheta,\bT)\in\bTheta\times\mathcal{T}_{\tau_{\min}} }\sum_{\underset{\bz\nsim\bzs}{\bz\in\mathcal{B}\left(\bzs;\tilde{c}\right)}}p\left(\bY,\bz;\btheta,\bT\right)=o_P\left[p\left(\bY,\bzs;\bthetas,\bTs\right)\right].\]
\end{prop}

    \begin{prop}[Partitions far from the true one]
Under Assumptions  (C.1) and (ID.4.s), asymptotically in  $n$ and $d$, if there exists a sequence of radius  $R_{nd}$ converging to 0 such that $R_{nd}>\max\left(\sqrt{\log K/d},\frac{4\text{Diam}(\bTheta)K^2}{\sqrt{nd}\delta(\bthetas)},\frac{4\log \left(\frac{1}{c}\right)}{d\delta(\bthetas)}\right)$, then for all  $\btheta\in\bTheta$ and all $\bT\in\mathcal{T}_{\tau_{\min}}$:
\begin{equation}
\sup_{(\btheta,\bT)\in\bTheta\times\mathcal{T}_{\tau_{\min}}} \sum_{\bz\notin\mathcal{B}\left(\bzs;R_{nd}\right)}p(\bY,\bz;\btheta,\bT)=p(\bY,\bzs;\bthetas,\bTs)o_P(1)
\end{equation}
with probability $1-\exp\left(-\varepsilon_{nd}^2nd\right)$ where $\varepsilon_{nd}=\min\left(\frac{\delta(\bthetas)R_{nd}}{16},1/\sqrt{2}\right)$.
\label{prop:part_eloign}
\end{prop}

Then, we get that 
\begin{align*}
p(\bY;\btheta,\bT)
&= \sum_{\bz \sim \bzs}p(\bY,\bz;\btheta,\bT) + 
p(\bY|\bzs;\bthetas,\bTs) o_P(1)
\end{align*}
which gives the result.


\begin{cor}\label{cor:consistency}

Let $\mathbf{Y}$ be a matrix of a $n\times T$ observations of the model \eqref{model_multiv} with true parameter $\bthetas, \bTs$ where the number of clusters $K$ and the number of segments $(L_k)_{1\leq k \leq K}$ are known. 
We assume \ref{ID:1}, \eqref{ID:3:s}, \eqref{ID:4:s}, \eqref{ass:C.1}, \eqref{ass:C.2} and \eqref{ass:C.3}. Then, 

\begin{align*}
\left(\widehat{\btheta},\widehat{\bT}\right)\underset{n, d \rightarrow +\infty}{\overset{\mathbb{P}}{\rightarrow}} \left(\btheta^\star,\bT^\star\right).
\end{align*}
\end{cor}

From Theorem~\ref{th:consistence}, we have that the likelihood focuses on the true partition of the data, hence the maximum likelihood estimator is asymptotically close to the complete maximum likelihood, given the true partitions. In this particular case, since the partitions are known, our problem boils down to a standard segmentation problem.

\section{Simulation study}
\label{sec:simu}
\Coragain{We first provide a toy example to motivate the combination of clustering and segmentation. Then,} we do an empirical evaluation of our model on univariate generated data. \Cor{We also compare our approach with clustering and segmentation performed independently.}

\subsection{On the importance of temporal coherence}

\Coragain{Hereafter, we provide a toy example to emphasize why our method is needed when considering clustering of functional data. For $n=60$ observations, $T=16$ and $d=30$, we consider two regimes, the neutral one, seen as a Gaussian noise $\mathcal{N}(0,0.1)$, and the active one, generated from the Gaussian $\mathcal{N}(2, 1)$. We consider $K=3$ clusters, which are not ordered, but also $L=3$ segments, which coincide for all the clusters. Every point is neutral, except the first cluster for the first segment, the second cluster for the second segment, and the third cluster for the third segment. Such scenario can easily occur in reality if we consider for instance electrical consumption by different entities that are active at different time of the day (morning, afternoon and night).} \Coragain{Even though the clustering problem at hand is simple, existing clustering method for functional data (e.g. funLBM \cite{bouveyron:2017}) will struggle in such cases. Indeed, up to the permutation of the column, there is only one row cluster (and vice-versa). However, we want to avoid allowing the permutation of the columns, as  we will loose the semantic behind the timepoints, that are ordered.}

\Coragain{To showcase this issue, we ran \texttt{funLBM} with $K=2$ clusters and $L=3$ segments (clusters in columns in their case). For $K=3$, the model was systematically obtaining an empty cluster\footnote{the implementation provided by the authors does neither handle model selection, nor the empty cluster problem.}.}

\Coragain{We report the results in Table \ref{tab:toyexample}.
As there are no constraints for the segmentation in \texttt{funLBM}, we compute the ARI. Unsurprisingly, our method performs very well and strongly outperforms funLBM. We observe this behavior for two reasons (1) funLBM encounters an identifiability problem for the case of $K=3$; (2) by allowing the model to shuffle the time dimension, the model focuses more on separating the neutral vs. active regimes independently from the time at which these regimes occur.}   
\Coragain{For this reason, we argue that our modelisation is best suited when dealing with functional data (or any ordered dependent data) with heteorgeneity in the population through clusters and segments in time.}

\begin{table}[t]
    \centering
    \caption{Mean (std) ARI for the toy example for our method and its latent block model counterpart (implemented in \texttt{funLBM}).}
    \label{tab:toyexample}
    \begin{tabular}{ccc}\toprule
     model / evaluation & clustering & segmentation \\ \hline
    our model & 1 (0.00) &  0.82 (0.05)\\
    \texttt{funLBM} & 0.48 (0.20)&0.46 (0.16)\\
    \bottomrule
    \end{tabular}
\end{table}

\begin{figure}[!ht]
\includegraphics[width=\textwidth]{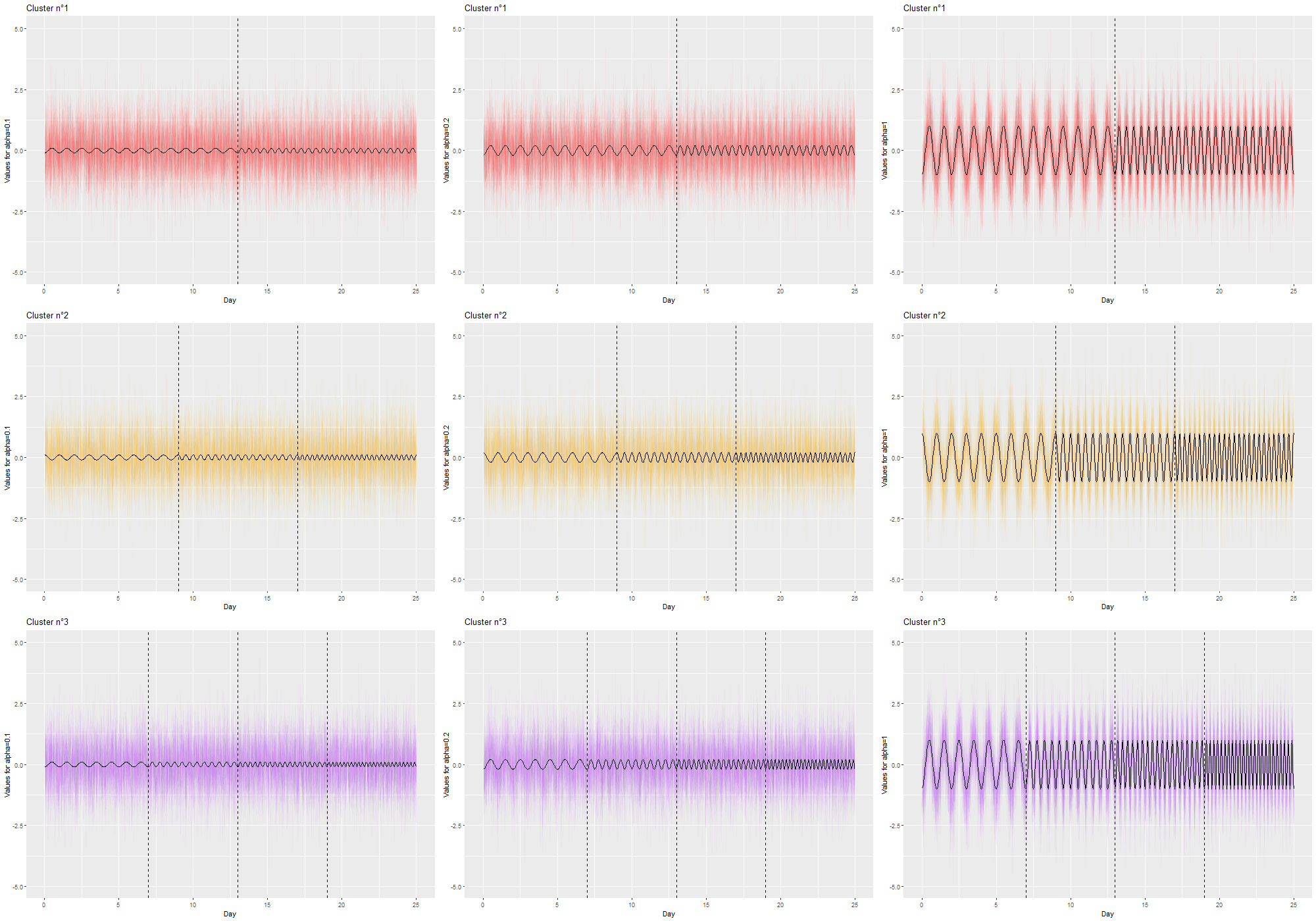}
\caption{Illustration of the different settings for $K=3$ with $(\alpha =0.1$ (left column), $0.2$ (middle column) and $1$ (right column). The number of breakpoints is $L_1=1, L_2=2, L_3=3$.}
\label{fig:dgp}
\end{figure} 
\subsection{Experimental Protocol}
\paragraph{Data generation process} We simulate data based on the following generative process. For a given number of multivariate observations $n\in\{100, 1000\}$ with $H=32$ and a number of days $d\in\{50, 100\}$, we start by generating a mixture model with $K=3$ clusters with equal proportions ($\pi_k=\frac{1}{3}$ for $k \in \{1,2,3\}$). For each $k\in \{1,\ldots, K\}$, we take:
$$f_{k\ell}(t_j) = (-1)^{k}. \alpha. \cos\left(\frac{2\pi t_j}{1+\ell}\right), $$
where $\alpha$ guides the amplitude of the generated curves, $\ell \in \{0, \ldots, L_k\}$ is the index of the current segment, and $j \in \{1, \ldots, T\}$ is the index of the timepoint. As a result $\alpha$ plays a dual role: the smaller $\alpha$ is, the harder it is to both differentiate between the clusters and to detect the break-points. In the sequel, we consider different settings by varying the values of $\alpha \in \{0.1, 0.2, 1\}$. Finally, we fix the variance $\sigma_{k\ell}=1$ for all cluster $k\in\{1,\ldots,3\}$ and segment $\ell \in\{1,\ldots,L_k\}$. 
Figure \ref{fig:dgp} illustrates all settings for $K=3$ and fixed $n$ and $d$.

\paragraph{Projection onto a wavelet basis}
The discretization of each component of the $p$-dimensionial curve is projected onto a wavelet basis\footnote{Note that any functional basis might be used to project the observed functions.}, that represents localized features of functions in a sparse way (\cite{Mallat}).  In our paper, the Discrete Wavelet Transform (DWT) is performed using a computationally fast pyramid algorithm (\cite{Mallat2, Poggi}).
We use both scaling functions to construct approximations of the function of interest, and the wavelet functions serve to provide the details not captured by successive approximations.

\paragraph{Evaluation Metrics}

In order to evaluate the quality of our output, we consider different metrics of evaluation. For the clustering part, we compute the Adjusted Rand Index (ARI) (\cite{hubert1985comparing}) and the Normalized Classification Error (NCE) (\cite{robert2021}). 

The ARI is a measure of agreement between two partitions defined by
\[
\text{ARI}(\mathbf{z}^{*}, \hat{\mathbf{z}}) = \frac{\sum_{k,k'}\binom{n_{kk'}}{2}-\left[\sum_k\binom{n_k}{2}\sum_{k'}\binom{\hat{n}_{k'}}{2}/\binom{n}{2} \right]}{\frac{1}{2}\left[\sum_k\binom{n_k}{2}+\sum_{k'}\binom{\hat{n}_{k'}}{2}\right] - \left[\sum_k\binom{n_k}{2}\sum_{k'}\binom{\hat{n}_{k'}}{2}\right]/\binom{n}{2} }, 
\]
where $n_k$ denotes the number of observations contained in the $k$-th cluster described by $\bzs$, $\hat{n}_k'$ is the number of observations in the estimated $k$-th cluster described by $\hat{\bz}$ and $n_{kk'}$ denotes the number of observations that are in the intersection between ground-truth cluster $k$ and the estimated cluster $k'$. The ARI lies between $0$ and $1$, with $1$ indicating a perfect agreement between the two partitions and $0$ that the two partitions are random. 

The NCE is defined by 
$$
\text{NCE}(\bzs,\hat{\bz}) = 1 - \frac{d_{0,\sim}(\bzs,\hat{\bz})}{n/K},
$$
where the distance $d_{0,\sim}$ between two partitions $\bz, \bzs$ is given by Eq. \eqref{d0sim}. 
The NCE lies in $[0, 1]$, with $0$ indicating a perfect estimation of $\bzs$.

Regarding the quality of the segmentation part, we compute the Hausdorff distance (\cite{BraultOSL18}), defined by
\begin{align*}
d_{\text{Haus}}(\mathbf{T}^*;\hat{\mathbf{T}}) = \max_{k\in \{1,2,3\}} \max_{1 \leq \ell \leq L_k} \left\{ \frac{|T_{k\ell}^* - \hat{T}_{k\ell}|}{d}\right\}.
\end{align*}
Note that $d_{\text{Haus}} \in [0,1]$, where $0$ indicates a perfect matching between the ground-truth and the estimated break-points. 

Finally, we propose to evaluate the quality of the parameter estimation with respect to the real values, by taking the difference between the ground-truth parameters (obtained by using the projection of $f_{k\ell}$, without additional noise) and the estimated ones. 


\subsection{Results}
ARI results are summarized in Table \ref{tab:ARI} (NCE results are consistent and presented in the Appendix). We observe that our method behaves as expected. As $\alpha$ increases (i.e. the task becomes easier), ARI and NCE increase and decrease, respectively. We obtain a perfect clustering when $\alpha=1$. In addition, we recover the results of Theorem~\ref{th:consistence}: one can see that as the number of observations ($d$) and the number of individuals ($n$) increases, the classification gets better (contrary to the one-dimensional mixture model where the proportion of errors tends toward a limit value even if the number of individuals keeps increasing). 

\begin{table}[t]
    \centering
    \caption{Mean (std) ARI for different settings: $(n,d)$ in row and $\alpha$ in columns.}
    \label{tab:ARI}
    \begin{tabular}{c}
    \begin{minipage}{0.95\textwidth}
    \begin{center}
    \begin{tabular}{cccc}
    \toprule
    \multicolumn{4}{c}{ARI $\uparrow$}\\
    \midrule
     $(n,d)$|$\alpha$&0.1&0.2&1\\\midrule
    (100,50)&0.13 (0.15)&0.61 (0.15)&1 (0)\\
    (100,100)&0.32 (0.19)&0.76 (0.16)&1 (0)\\
    (1000,50)&0.33 (0.13)&0.76 (0.1)&1 (0)\\
    (1000,100)&0.55 (0.099)&0.93 (0.046)&1 (0)\\
    \bottomrule
    \end{tabular}
    \end{center}
    \end{minipage}
     \end{tabular}
\end{table}
\begin{figure}[!t]
\begin{subfigure}[b]{0.24\textwidth}
\includegraphics[width=\linewidth]{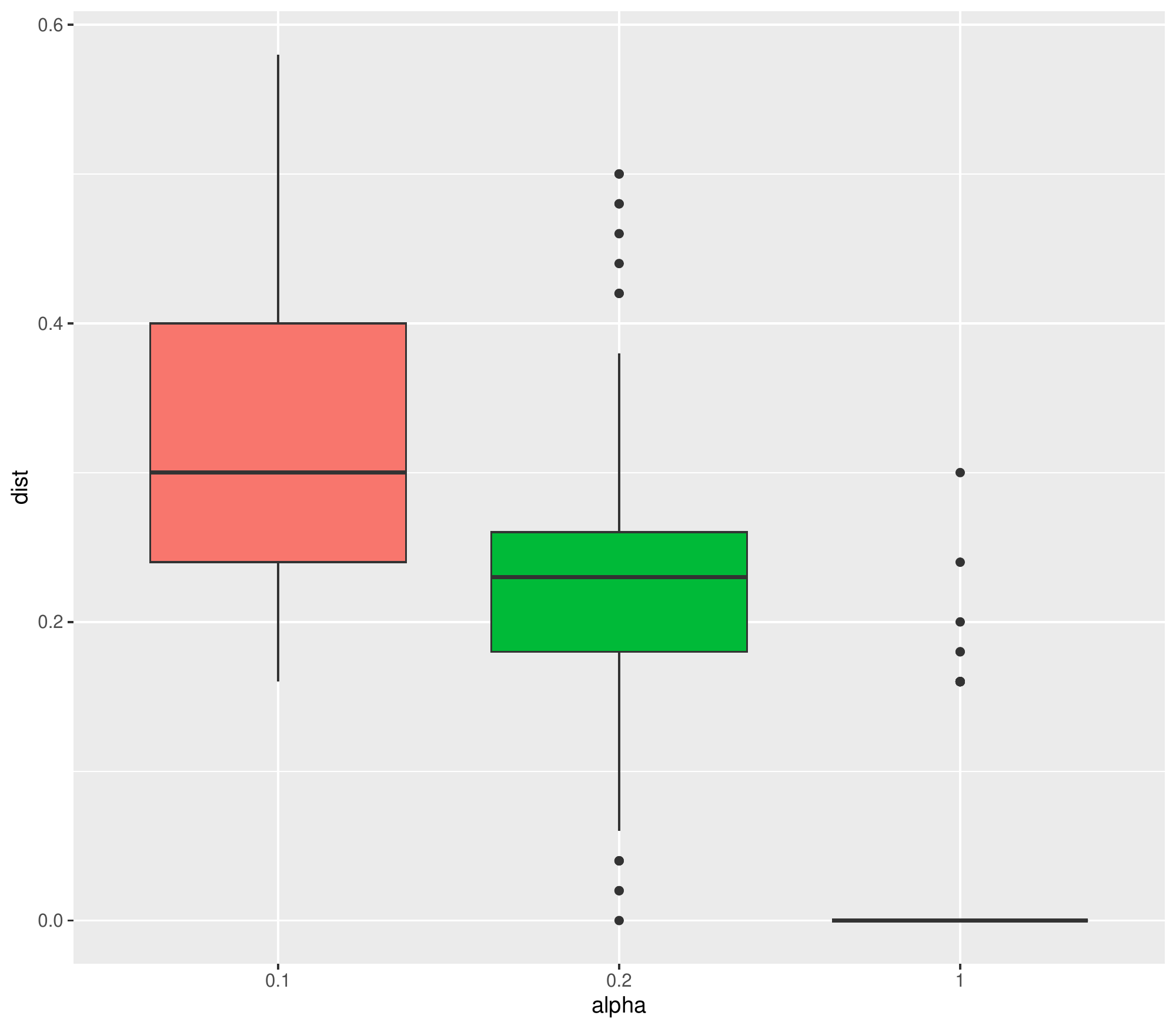}
\caption{n=100, d = 50}
\end{subfigure}
\begin{subfigure}[b]{0.24\textwidth}
\includegraphics[width=\linewidth]{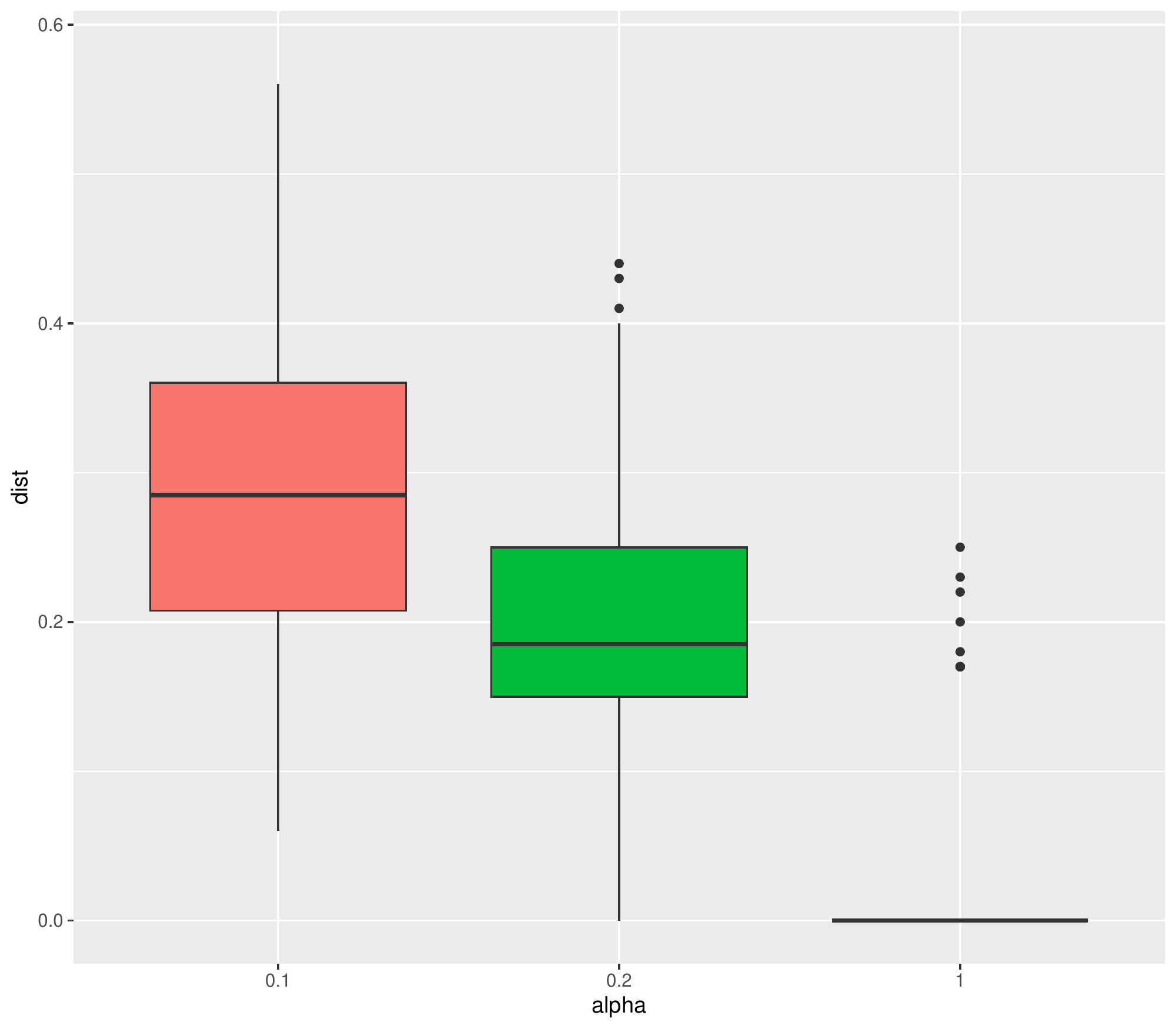}
\caption{n=100, d = 100}
\end{subfigure}
\begin{subfigure}[b]{0.24\textwidth}
\includegraphics[width=\linewidth]{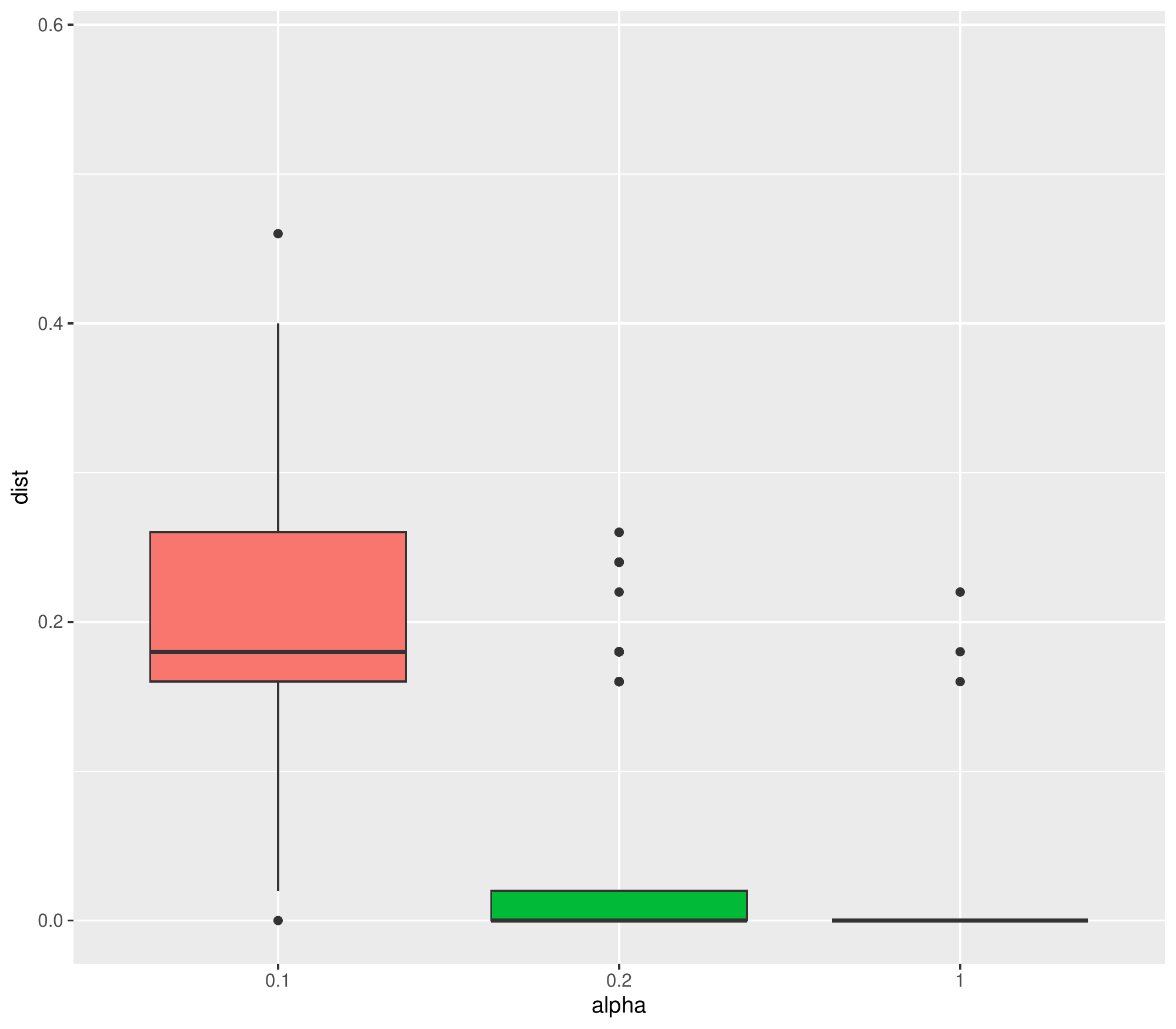}
\caption{n=1000, d = 50}
\end{subfigure}
\begin{subfigure}[b]{0.24\textwidth}
\includegraphics[width=\linewidth]{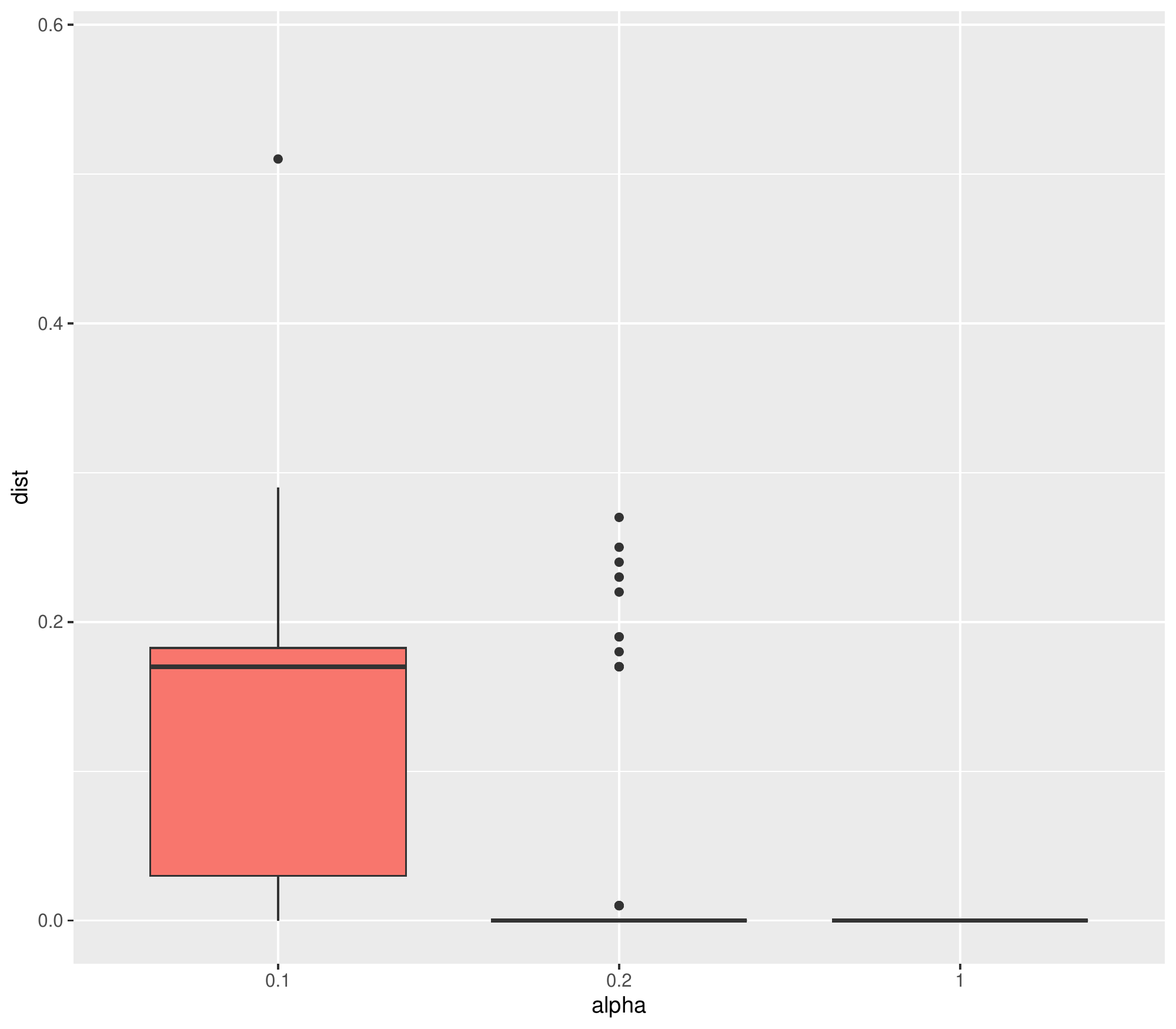}  
\caption{n=1000, d = 100}
\end{subfigure}
\caption{Boxplots of the Hausdorff distance for various combinations $(n,d)$. Color represents the complexity. }
\label{fig:Haus}
\end{figure} 
The same trend can be observed for the segmentation results presented in Figure \ref{fig:Haus}. We observe that the quality of the segmentation part (and hence the breakpoints localization) is correlated with the clustering performance aforementioned. 

\begin{figure}[!ht]
\begin{center}
\begin{tabular}{c|c|c}
&$d=50$&$d=100$\\
\hline
&&\\
\rotatebox{90}{\!\!\!\!\!$n=100$}&\begin{minipage}{0.44\textwidth}
\includegraphics[width=\linewidth]{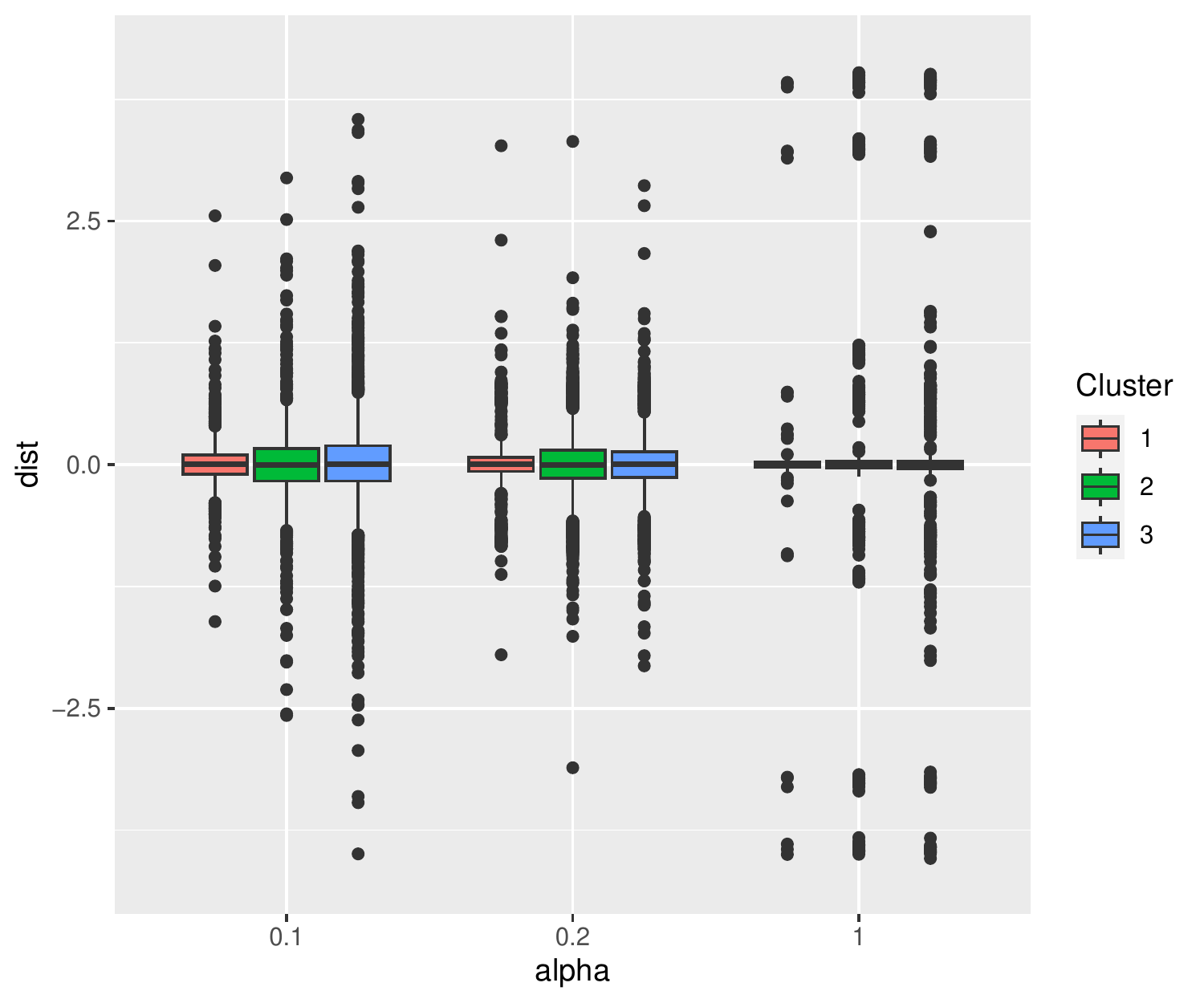}
\end{minipage}&
\begin{minipage}{0.44\textwidth}
\includegraphics[width=\linewidth]{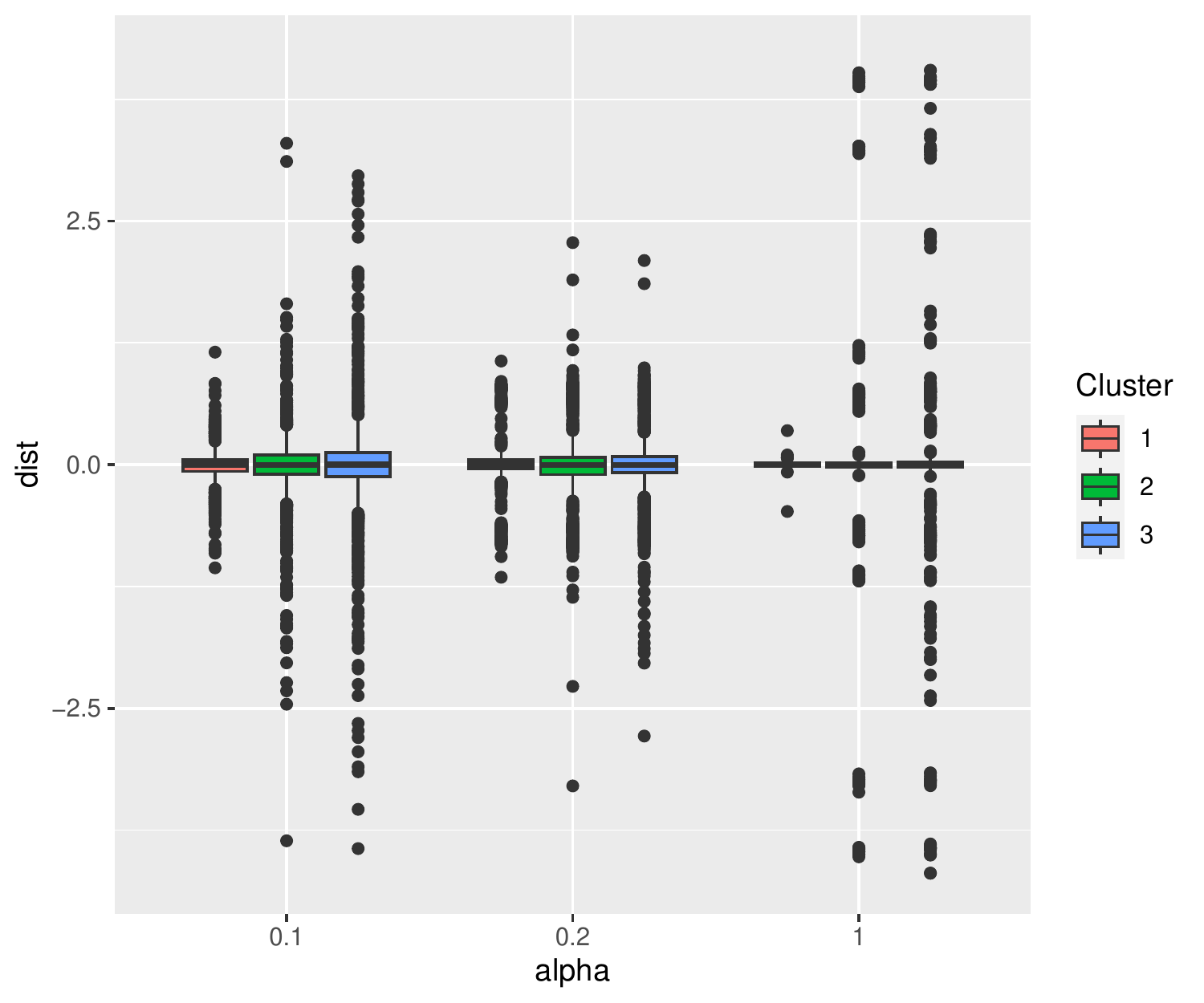}
\end{minipage}\\
\hline
&&\\
\rotatebox{90}{\!\!\!\!\!$n=1000$}&\begin{minipage}{0.44\textwidth}
\includegraphics[width=\linewidth]{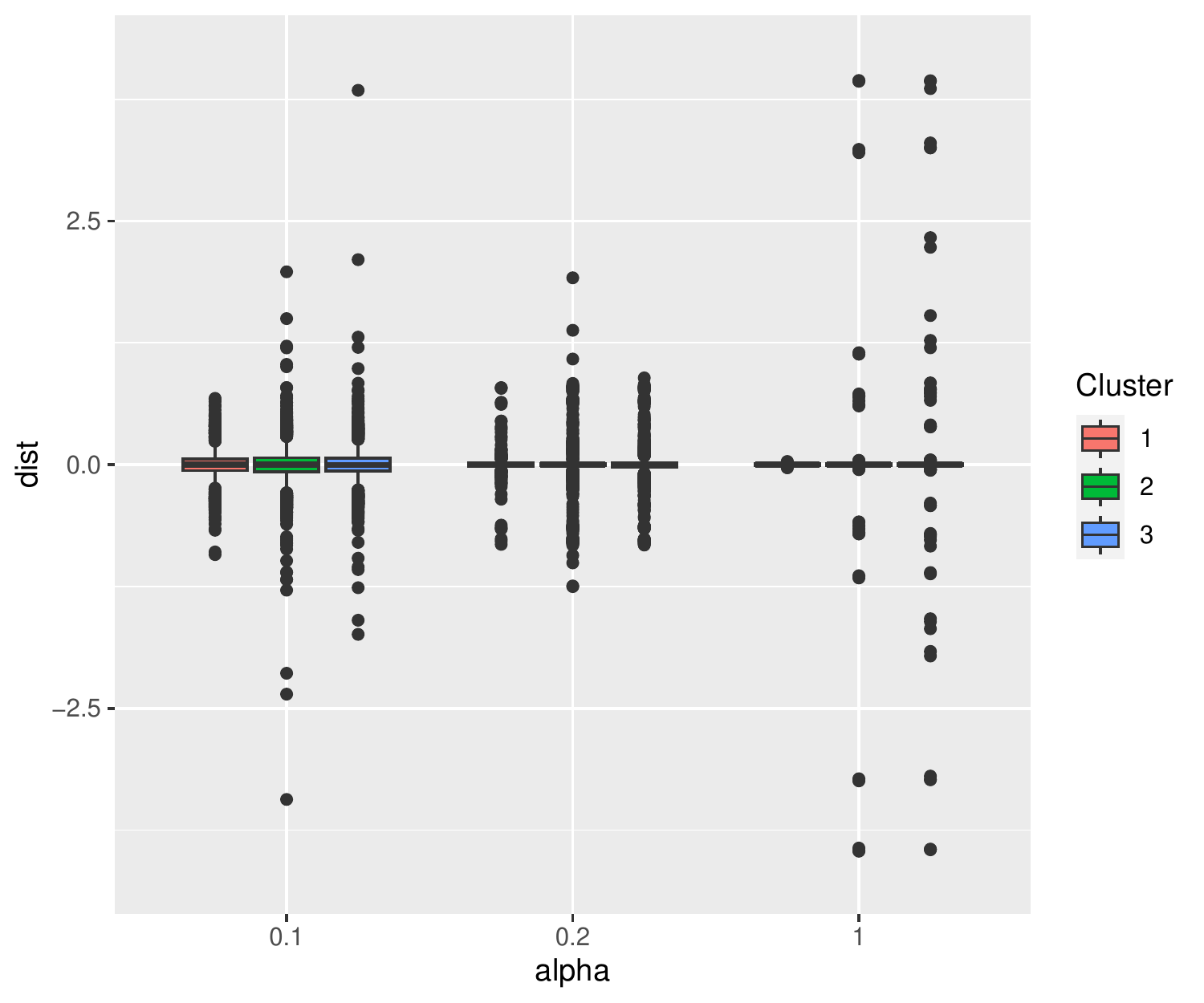}
\end{minipage}&
\begin{minipage}{0.44\textwidth}
\includegraphics[width=\linewidth]{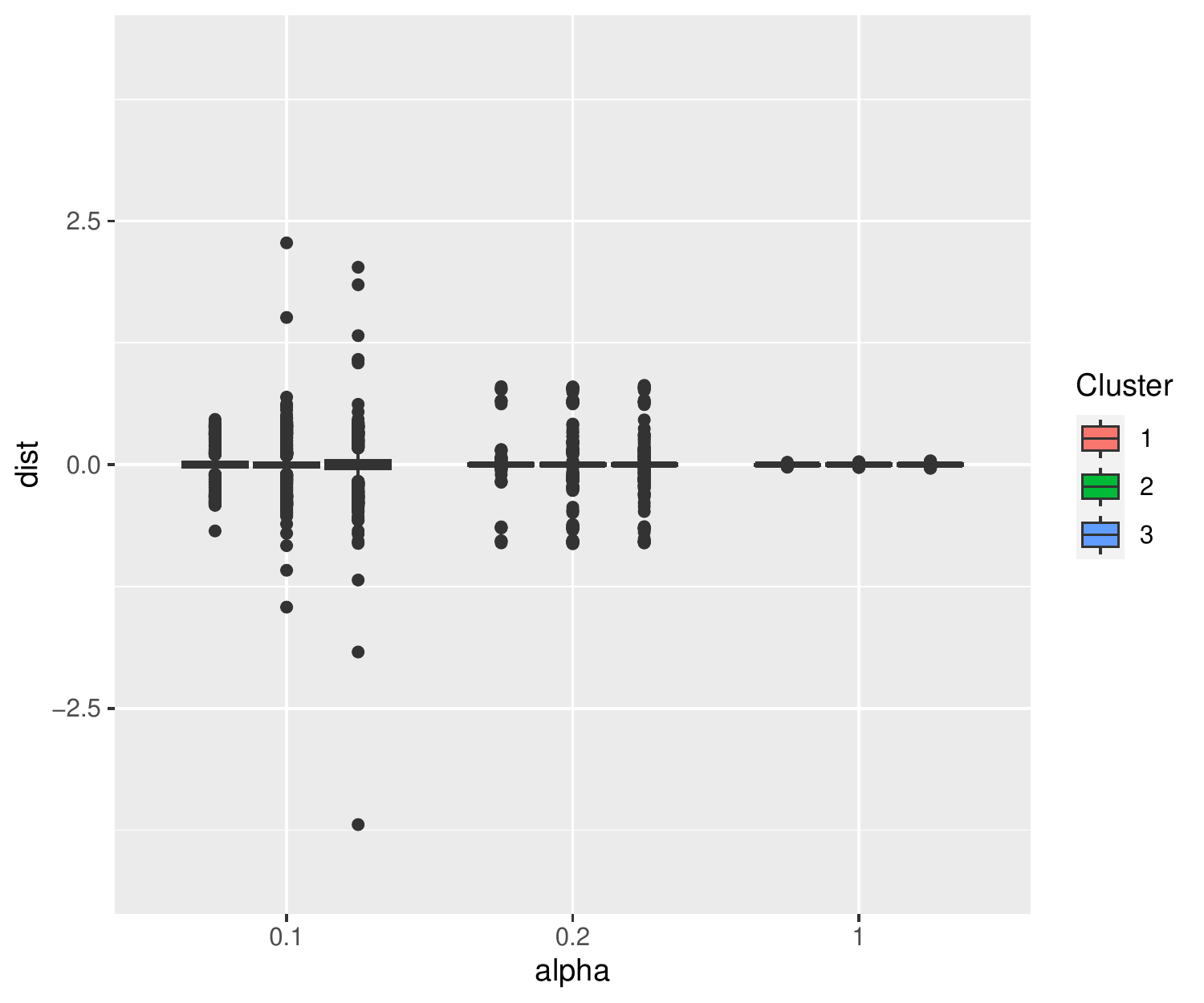}
\end{minipage}\\
\end{tabular}
\end{center}
\caption{Boxplots of the difference $\widehat{\mu}_{k\ell}$ and $\mu^\star_{k\ell}$ following the number of curves (rows), observations (columns), difficulties (x-axis of each graphic) and cluster (colors).}
\label{fig:TCL:Mu}
\end{figure} 

Finally, Figure \ref{fig:TCL:Mu} shows the performance of our approach on the estimation of the model parameters. First, we recover the consistency results stated in corollary~\ref{cor:consistency}. We also observe that the parameters from cluster one are better estimated than the ones from cluster three. This comes as no surprise as the number of breakpoints is less for cluster one, resulting in more observations to perform the estimation of the parameters. 

Next, we present comparative results between MixSeg and two simple baselines: a model that only performs clustering and a model that focuses on the segmentation.

\subsection{\Cor{Comparison with independent counterparts}}
\Cor{\textbf{Simple mixture model}}\\
\Cor{\noindent First, we run a simple mixture model following $\mathcal{N}(\mu_{k,r},\sigma_r^2)$ with at least 10 random initializations (and maximum 100 if some clusters are empty). ARI is reported in Table~\ref{tab:ARI:comp}. Our approach is referred to as MixSeg while the mixture model alone is referred to as SimpleMix. We note that on the most simple, i.e., configuration $\alpha=1$, both approaches perform similarly. However, for the two other cases, $\alpha=0.2$ and $0.1$, MixSeg significantly outperforms the simple mixture model, with ARI and NCE (see Appendix) that are 2 to 3 times better.}
\begin{table}[t]
    \centering
    \caption{ARI for different settings: $(n,d)$ in row and $\alpha$ in columns with mean and standard deviation.}
    \label{tab:ARI:comp}
    \begin{tabular}{ccccc}
    \toprule
    \multicolumn{5}{c}{ARI $\uparrow$}\\
         \midrule
         Setting&Model&0.1&0.2&1\\\midrule
 \multirow{2}{*}{(100,50)}&MixSeg&\textbf{0.13 (0.15)}&\textbf{0.61 (0.15)}&1 (0)\\
 &SimpleMix&0.05 (0.07)&\textbf{0.22 (0.17)}&0.98 (0.09)\\\midrule
 \multirow{2}{*}{(100,100)}&MixSeg&\textbf{0.32 (0.19)}&\textbf{0.76 (0.16)}&1 (0)\\
 &SimpleMix&0.13 (0.11)&0.39 (0.17)&0.99 (0.08)\\\midrule
 \multirow{2}{*}{(1000,50)}&MixSeg&\textbf{0.33 (0.13)}&\textbf{0.76 (0.10)}&1 (0)\\
 &SimpleMix&0.08 (0.07)&0.29 (0.14)&1 (9e-04)\\\midrule
 \multirow{2}{*}{(1000,100)}&MixSeg&\textbf{0.55 (0.10)}&\textbf{0.93 (0.05)}&1 (0)\\
 &SimpleMix&0.17 (0.12)&0.41 (0.14)&1 (0)\\
 \bottomrule
    \end{tabular}
\end{table}
\medskip 
\noindent\Cor{\textbf{ Simple change point detection}}\\
\Cor{\noindent For this part, let us assume that the number of change points to be detected is 5. This number corresponds to the total change points across all clusters. Looking at Figure 3, the first cluster has a single change point, while the second has two, and the last has three, one of which is common with the first cluster. For $\alpha$ and $n,d$ we consider the same scenario as previously. Hausdorff distances are presented in Figure~\ref{fig:Haus:Comp_5}, where the simple change point detection approach is referred to as SimpleSeg. We can see that the performance on the change point detection part for our approach is strongly correlated with the clustering performance (see Table \ref{tab:ARI:comp}). In particular, for $n,d$ and $\alpha$ relatively small, we observe that SimpleSeg achieves better results than our approach. For this scenario, the clustering results add noise to the change point detection part. However, as $n,d$ increases or when $\alpha=1$, the clustering results are better and enable our approach to achieve better or comparable results on change point detection.}

\begin{figure}[!ht]
\centering
\begin{subfigure}[b]{0.23\textwidth}
\includegraphics[width=\linewidth]{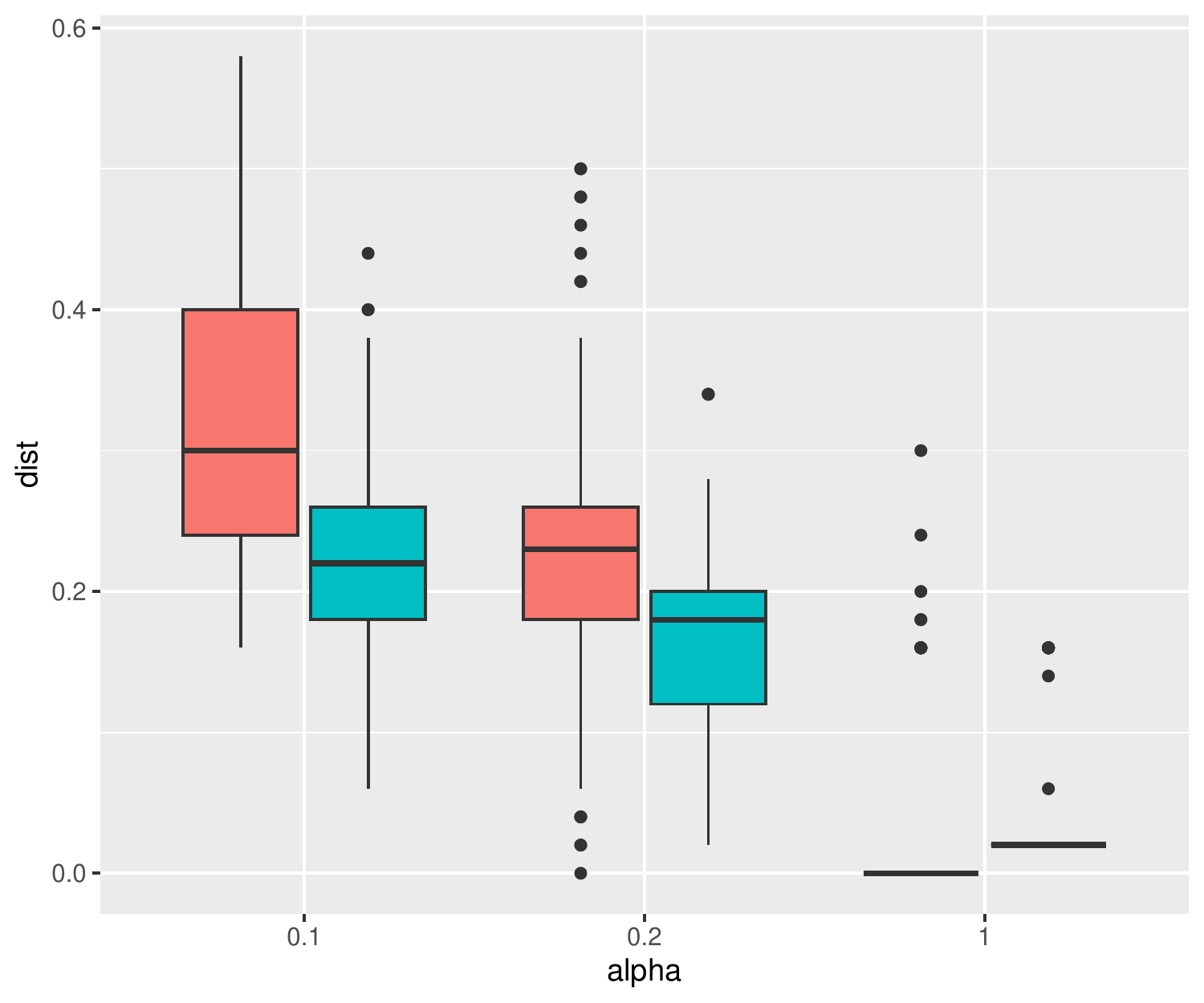}
\caption{\small$n$=100, $d$=50}
\end{subfigure}
\begin{subfigure}[b]{0.23\textwidth}
\includegraphics[width=\linewidth]{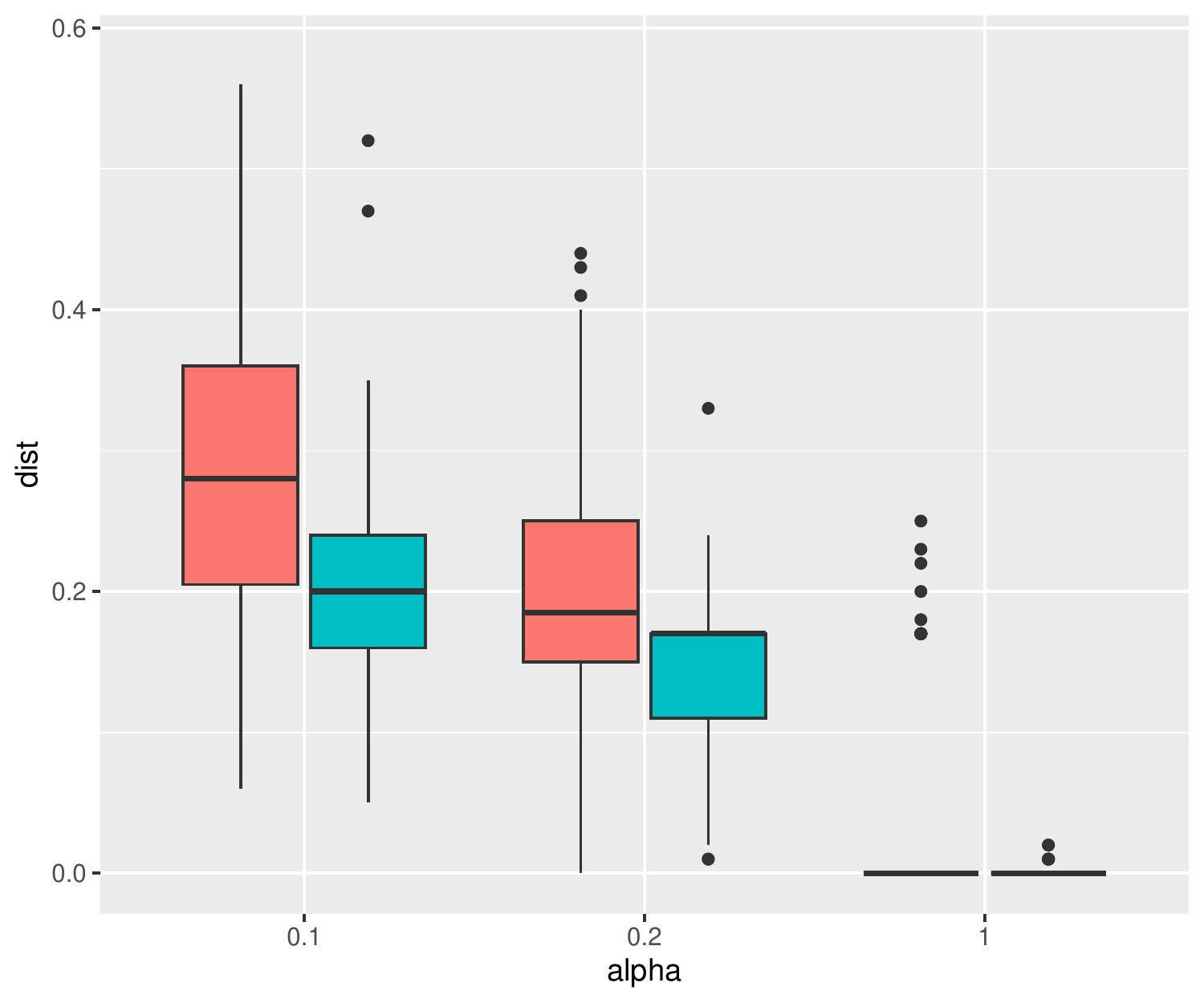}
\caption{\small$n$=100, $d$=100}
\end{subfigure}
\begin{subfigure}[b]{0.23\textwidth}
\includegraphics[width=\linewidth]{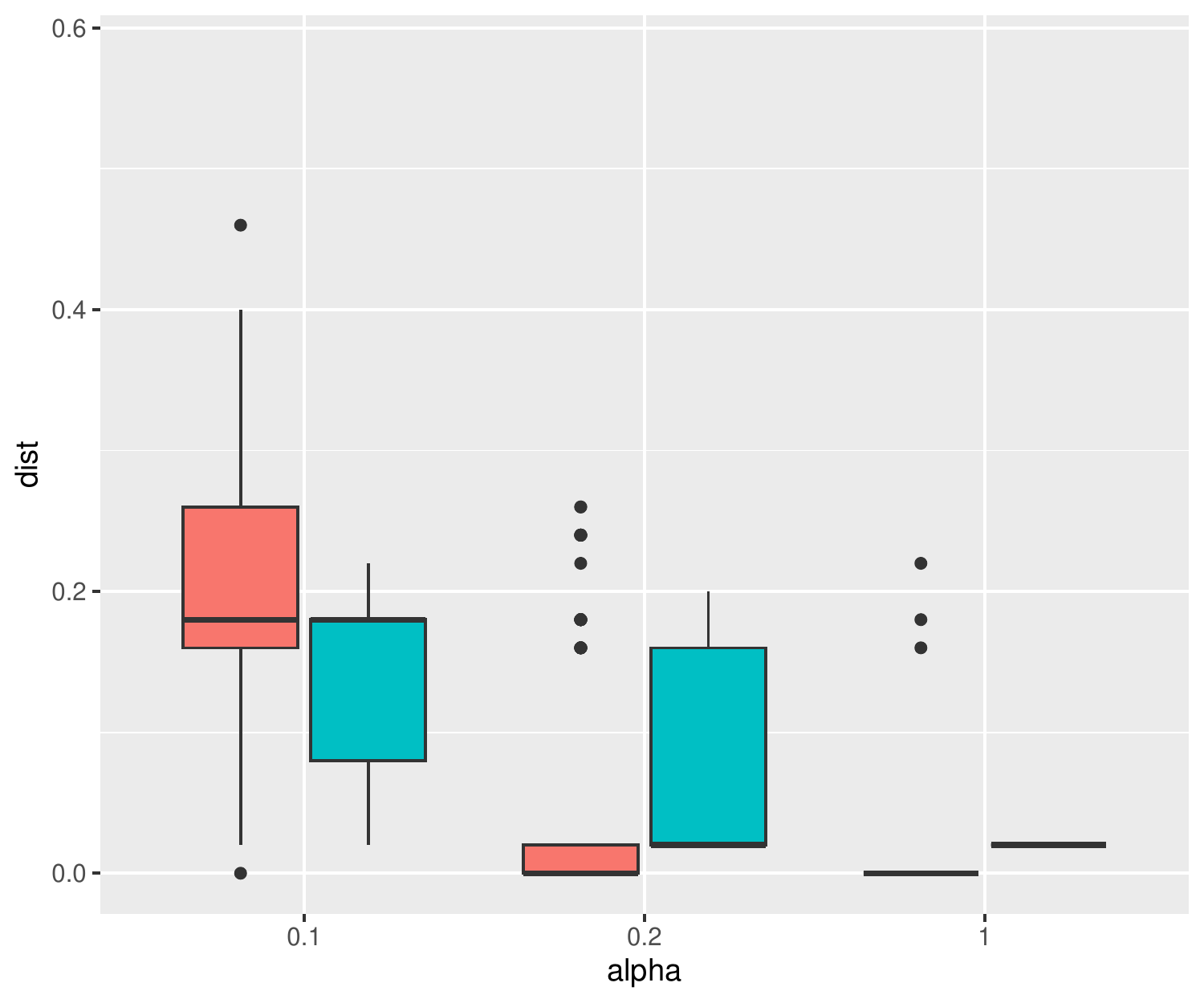}
\caption{\small$n$=1000, $d$=50}
\end{subfigure}
\begin{subfigure}[b]{0.267\textwidth}
\includegraphics[width=\linewidth]{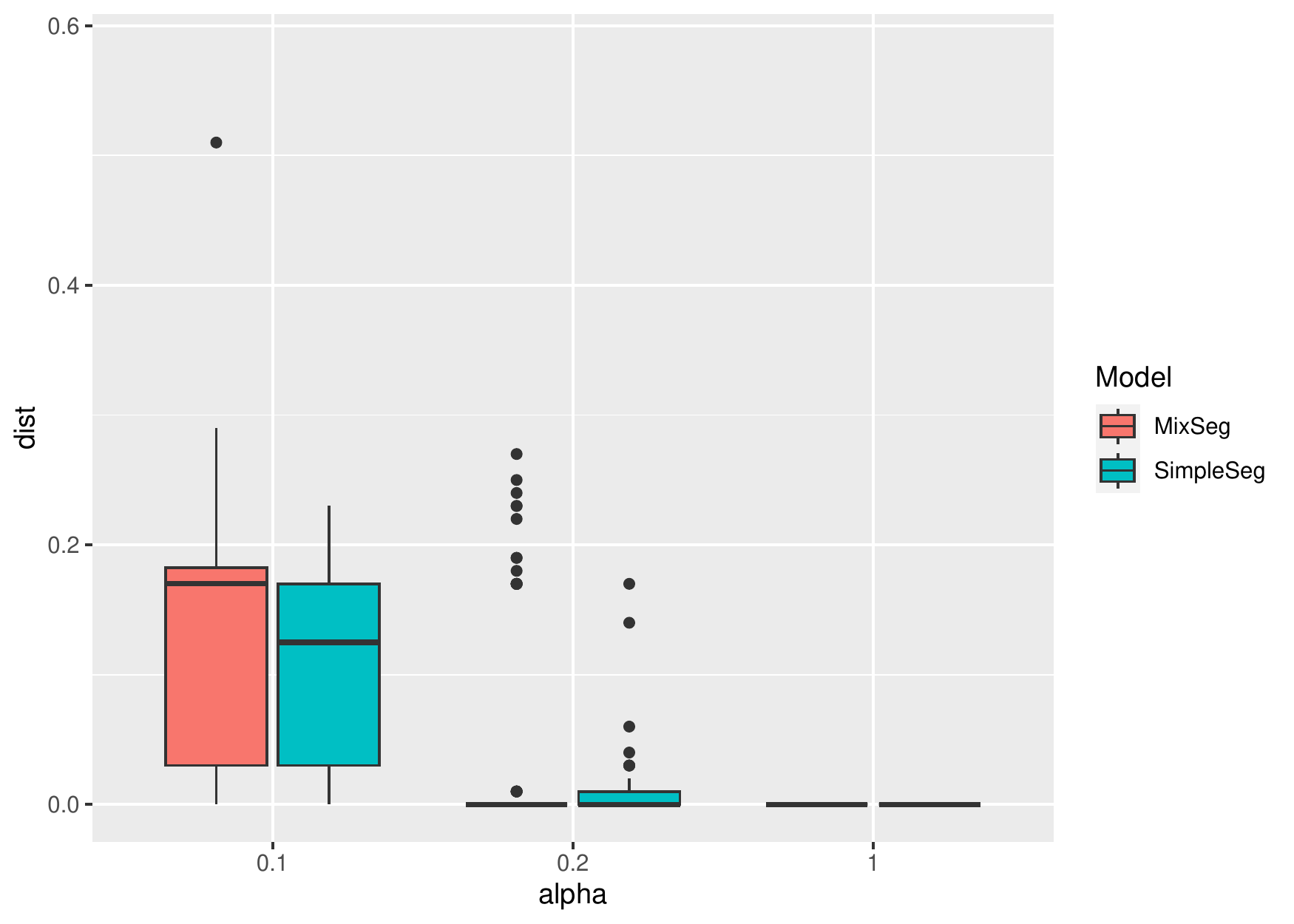}
\caption{\small$n$=1000, $d$=100}
\end{subfigure}
\caption{Boxplots of the Hausdorff distance varying the number of curves, the number of observations and difficulties (colors). }
\label{fig:Haus:Comp_5}
\end{figure}

\subsection{\Cor{Computational time}}
\Cor{In order to evaluate the time complexity of MixSeg, we follow the experimental protocol described a the beginning of this section and increase $n$ and $d$ by $50$ at each step. For each configuration $(n,d,\alpha)$, we simulate 20 matrices and for each matrix, we run our method 5 times.  Running times are recorded with the~\texttt{microbenchmark} package (version 1.4.10; see~\cite{olaf2021microbenchmark}) and simulation plans are launched on a cluster\footnote{Cluster \texttt{Luke44} -- 28 cores, 128Go RAM, GPU 2xK40m, 2xIntel Xeon E5-2680 2.40 GHz more information \url{https://scalde.gricad-pages.univ-grenoble-alpes.fr/web/pages/presentation.html}}. Results are presented in Figure \ref{fig:time:n} and consistent with Proposition \ref{prop:complexity}}. 

\begin{figure}[!ht]
    \centering
    \includegraphics[width=\linewidth]{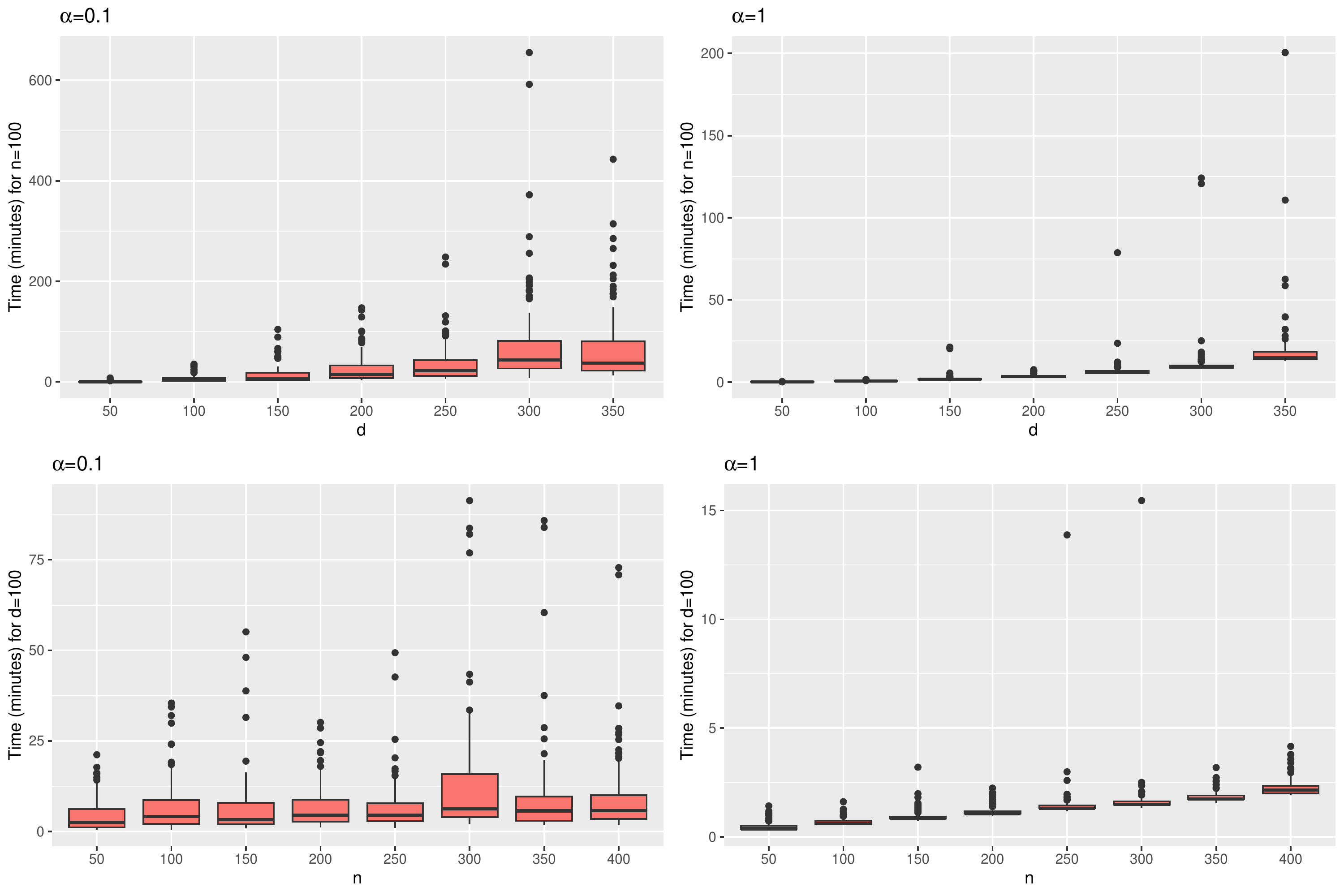}
    \caption{Computational time (in minutes) of $n$ fixed and $d$ varying (top line) and for $d$ fixed and varying $n$ (bottom line). We present results for $\alpha=0.1$ and $\alpha=1$.}
    \label{fig:time:n}
\end{figure}

\section{Real Data Analysis}
\label{sec:real_data}
We propose to apply our model to analyze electricity consumption using the Enedis Open Data Set\footnote{available at \href{https://data.enedis.fr/pages/accueil/?id=init}{https://data.enedis.fr/pages/accueil/?id=init}}. We focus on the year 2020 (52 weeks), corresponding to the outburst of the COVID-19 pandemic.
We built 984 observations by cross-referencing information on the type of contract subscribed to, the customer profile and the region of France. Out of these 984 observations, we removed those with missing values to obtain a final population of $889$ individuals. The curves are observed in kW every 30 minutes. We have chosen to analyze the curves by considering the week as a time unit of interest, hence we project those curves onto the Haar basis with $p=42$ ($d=52$ weeks).

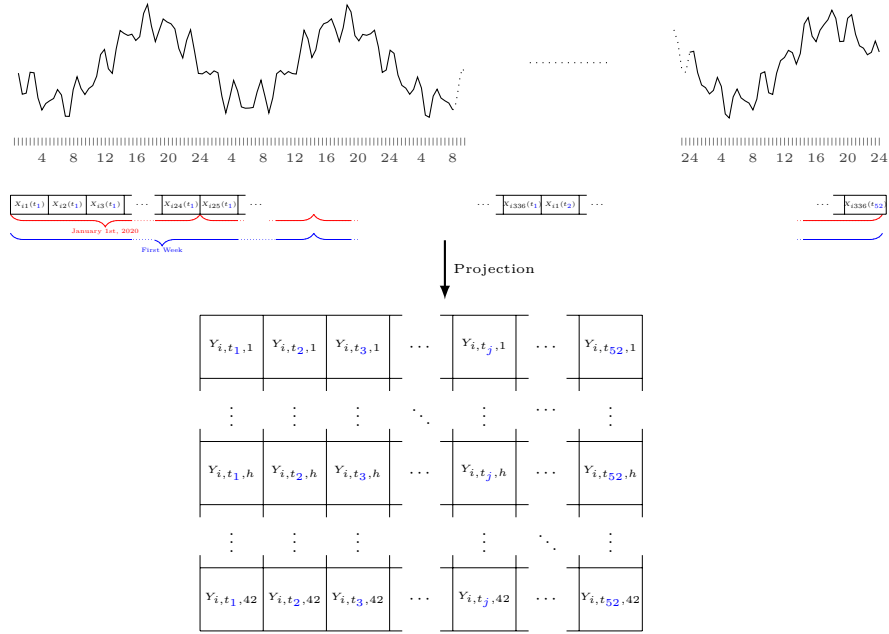
\begin{figure}[t]
    \centering
    \begin{tikzpicture}[scale=1.05]
\draw[domain=0.1:5.6, samples=112] plot (\x, {-sin(\x/2.4*2*pi r)/2+cos(\x*7*pi r)/6-cos(\x*11*pi r)/10}) ;
\draw[domain=5.6:5.75, samples=6,dotted] plot (\x, {-sin(\x/2.4*2*pi r)/2+cos(\x*7*pi r)/6-cos(\x*11*pi r)/10}) ;
\draw[dotted] (6.575,0) -- (7.575,0);
\draw[domain=5.4:5.6, samples=5,dotted] plot (\x+3, {-sin((\x-12.8)/2.4*2*pi r)/2+cos((\x-12.8)*7*pi r)/6-cos((\x-12.8)*11*pi r)/10}) ;
\draw[domain=5.6:8, samples=49] plot (\x+3, {-sin((\x-12.8)/2.4*2*pi r)/2+cos((\x-12.8)*7*pi r)/6-cos((\x-12.8)*11*pi r)/10}) ;
\foreach \t in {0.05,0.1,...,5.8}
    {\draw[opacity=0.5] (\t,-1.05) -- (\t,-0.95);}
\foreach \d in {0,2.4}{
    \foreach \t in {4,8,...,24}
        {\draw[opacity=0.75] (\t/10+\d,-1.05) node[below]{\tiny{\t}};}
        }
\foreach \t in {4,8}
    {\draw[opacity=0.75] (\t/10+4.8,-1.05) node[below]{\tiny{\t}};}
\foreach \t in {5.5,5.55,...,8.05}
    {\draw[opacity=0.5] (\t+3,-1.05) -- (\t+3,-0.95);}
\foreach \t in {4,8,12,16,20,24}
    {\draw[opacity=0.75] (\t/10+8.6,-1.05) node[below]{\tiny{\t}};}
\foreach \t in {24}
    {\draw[opacity=0.75] (\t/10+6.2,-1.05) node[below]{\tiny{\t}};}
\scalebox{0.48}{
\draw[red,decorate,decoration={brace,amplitude=10pt,mirror}] (0,-4) -- (5,-4) node [midway,below=8pt] {\tiny January 1st, 2020};
\fill[white] (3.2,-4) rectangle (3.8,-4.25);
\draw[red,decorate,decoration={brace,amplitude=10pt,mirror},dotted] (0,-4) -- (5,-4);
\draw[red,decorate,decoration={brace,amplitude=10pt,mirror}] (5,-4) -- (8,-4);
\fill[white] (6,-4) rectangle (7,-4.5);
\draw[red,decorate,decoration={brace,amplitude=10pt,mirror},dotted] (5,-4) -- (8,-4);
\fill[white] (6.2,-4) rectangle (6.8,-4.5);
\draw[red,decorate,decoration={brace,amplitude=10pt,mirror}] (8,-4) -- (23,-4);
\fill[white] (9,-4) rectangle (20.9,-4.5);
\draw[red,decorate,decoration={brace,amplitude=10pt,mirror},dotted] (8,-4) -- (23,-4);
\fill[white] (9.2,-4) rectangle (20.7,-4.5);
\draw[blue,decorate,decoration={brace,amplitude=10pt,mirror}] (0,-4.5) -- (8,-4.5) node [midway,below=8pt] {\tiny First Week};
\fill[white] (3.2,-4.5) rectangle (3.8,-4.75);
\fill[white] (6,-4.5) rectangle (7,-4.75);
\draw[blue,decorate,decoration={brace,amplitude=10pt,mirror},dotted] (0,-4.5) -- (8,-4.5);
\draw[blue,decorate,decoration={brace,amplitude=10pt,mirror}] (8,-4.5) -- (23,-4.5);
\fill[white] (9,-4.5) rectangle (20.9,-4.75);
\draw[blue,decorate,decoration={brace,amplitude=10pt,mirror},dotted] (8,-4.5) -- (23,-4.5);
\fill[white] (9.2,-4.5) rectangle (20.7,-5);
\draw (0,-3.5) grid[ystep=0.5,xstep=1] (3.2,-4);
\foreach \h in {1,2,3}{
    \draw (\h-0.5,-3.75) node{\tiny$X_{i\h}(t_{\textcolor{blue}{1}})$};
    }
\draw (3.5,-3.75) node{$\cdots$};
\draw (3.8,-3.5) grid[ystep=0.5,xstep=1] (6.2,-4);
\draw (4.5,-3.75) node{\tiny$X_{i24}(t_{\textcolor{blue}{1}})$};
\draw (5.5,-3.75) node{\tiny$X_{i25}(t_{\textcolor{blue}{1}})$};
\draw (6.5,-3.75) node{$\cdots$};
\draw (12.5,-3.75) node{$\cdots$};
\draw (12.8,-3.5) grid[ystep=0.5,xstep=1] (15.2,-4);
\draw (13.5,-3.75) node{\tiny$X_{i336}(t_{\textcolor{blue}{1}})$};
\draw (14.5,-3.75) node{\tiny$X_{i1}(t_{\textcolor{blue}{2}})$};
\draw (15.5,-3.75) node{$\cdots$};
\draw (21.45,-3.75) node{$\cdots$};
\draw (21.7,-3.5) grid[ystep=0.5,xstep=1.1] (23.1,-4);
\draw (22.55,-3.75) node{\tiny$X_{i336}(t_{\textcolor{blue}{52}})$};}
\begin{scope}[shift={(0,4)}]
\scalebox{0.8}{
\draw (3,-8) grid (6.2,-9.2);
\draw (6.8,-8) grid (8.2,-9.2);
\draw (8.8,-8) grid (10,-9.2);
\draw (3,-9.8) grid (6.2,-11.2);
\draw (6.8,-9.8) grid (8.2,-11.2);
\draw (8.8,-9.8) grid (10,-11.2);
\draw (3,-11.8) grid (6.2,-13);
\draw (6.8,-11.8) grid (8.2,-13);
\draw (8.8,-11.8) grid (10,-13);
\foreach \j in {1,2,3}{
    \foreach \h in {1}{
        \draw (\j+2.5,-7.5-\h) node{\tiny$Y_{i,t_{\textcolor{blue}{\j}},\h}$};
    }
}
\foreach \j in {1,2,3}{
    \draw (\j+2.5,-9.5) node{$\vdots$};
}
\foreach \h in {1}{
    \draw (6.5,-7.5-\h) node{$\cdots$};
}
\foreach \h in {1}{
        \draw (7.5,-7.5-\h) node{\tiny$Y_{i,t_{\textcolor{blue}{j}},\h}$};
}
\draw (7.5,-9.5) node{$\vdots$};
\foreach \h in {1,2}{
    \draw (8.5,-7.5-\h) node{$\cdots$};
}
\foreach \h in {1}{
        \draw (9.5,-7.5-\h) node{\tiny$Y_{i,t_{\textcolor{blue}{52}},\h}$};
}
\draw (9.5,-9.5) node{$\vdots$};
\foreach \j in {1,2,3}{
    \draw (\j+2.5,-10.5) node{\tiny$Y_{i,t_{\textcolor{blue}{\j}},h}$};
}
\foreach \j in {1,2,3}{
    \draw (\j+2.5,-11.5) node{$\vdots$};
}
\draw (6.5,-10.5) node{$\cdots$};
\draw (7.5,-10.5) node{\tiny$Y_{i,t_{\textcolor{blue}{j}},h}$};
\draw (7.5,-11.5) node{$\vdots$};
\draw (8.5,-10.5) node{$\cdots$};
\draw (9.5,-10.5) node{\tiny$Y_{i,t_{\textcolor{blue}{52}},h}$};
\draw (9.5,-11.5) node{$\vdots$};
\foreach \j in {1,2,3}{
    \draw (\j+2.5,-12.5) node{\tiny$Y_{i,t_{\textcolor{blue}{\j}},42}$};
}
\draw (6.5,-12.5) node{$\cdots$};
\draw (7.5,-12.5) node{\tiny$Y_{i,t_{\textcolor{blue}{j}},42}$};
\draw (8.5,-12.5) node{$\cdots$};
\draw (9.5,-12.5) node{\tiny$Y_{i,t_{\textcolor{blue}{52}},42}$};
\draw (6.5,-9.5) node{$\ddots$};
\draw (8.5,-11.5) node{$\ddots$};}
\draw[thick,->,>=latex] (5.5,-6.25) -- (5.5,-7) node[midway,right]{\tiny Projection};
\end{scope}
    \end{tikzpicture}
    \vspace*{-2cm}
    \caption{Illustration of the real data, for an individual: curves (top), vectorial notation (middle) and matrix projection (bottom).}
    \label{fig:enter-label}
\end{figure}

\paragraph{Remark} We first ran our approach on such data. The result of the model selection (see next paragraph) resulted in 2 clusters. After having analyzed these results, we found that the model isolated all profiles associated with public lighting contracts, which are not subject to a notion of energy consumption behavior in the sense that interests us.
For this reason, we chose to discard these observations in the following, which give us a final population of $791$ observations.

\subsection{Model Selection}
When dealing with real data, we have no access to the \textit{true} number of clusters or the number of breakpoints. We therefore propose to adapt a model selection strategy proposed in \cite{zhang2007modified} relying on the Bayesian information criterion (BIC) (\cite{Schwarz_1978}). We obtain the following criterion for our model:
\begin{eqnarray*}
\text{BIC}(K, \mathbf{L})&=&\max_{\btheta,\bT}\text{lik}\left(\bY;K,\bT,\btheta\right)-\underbrace{\frac{K-1}{2}\log (n)}_{\text{for }\bpi}\\
&&\underbrace{-\frac{1}{2}\sum_{k=1}^K\left[3p(\Lk+1)\log(ndp)+\sum_{\ell=0}^{\Lk}\log\left(\frac{\widehat{T}_{k,\ell+1}-\widehat{T}_{k\ell}}{d}np\right)\right]}_{\text{for }\bmu\text{ and }\bT}\\
&&\underbrace{-\frac{K}{2}\log(ndp)}_{\text{for }\bsigma}.\\
\end{eqnarray*}
The general form of the penalty in this criterion allow us to account for the specificities of all parameters.

Since an exhaustive exploration of the number of clusters and breakpoints is not possible, we adapt the \texttt{bikm1} strategy proposed by \cite{robert2021bikm1}. Given a reference configuration (the current state of the model), we proceed as follows: 
\begin{itemize}
    \item \textbf{Backward Search}: we remove a cluster (K possible options).  For each of these options, we make 10 random initializations as well as a random distribution of the observations of the deleted cluster in the remaining ones. 
    \item \textbf{Forward Search}: we add a cluster with 1 breakpoint and make 10 random initializations. 
    \item \textbf{Number of breakpoints}: we proceed with the same principles (backward and forward searches) for the number of breakpoints (with $K$ fixed). 
\end{itemize}

In the end, we obtain $K=3$ clusters and the number of breakpoints within each cluster is given by $L_1=1$, $L_2=2$ and $L_3=6$. Figure \ref{fig:rep-cluster} shows the distribution of our observations within clusters and the locations of the different breakpoints.
\begin{figure}
    \centering
    \includegraphics[width=\linewidth]{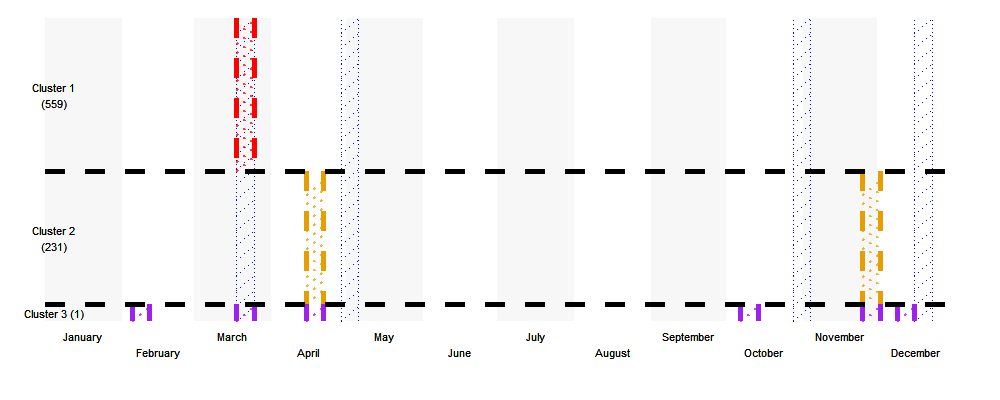}
    \caption{Distribution of the observations across clusters (horizontal dashed lines) and localization of the breakpoints within each cluster (vertical dashed lines). Vertical dot points areas correspond to the start and end weeks of the lockdown periods. }
    \label{fig:rep-cluster} 
\end{figure}
We observe that cluster 3 only contains one observation and hence overfits on the number of breakpoints. Based on the information at our disposal (type of contract, profile, and geographic area), we have not been able to draw interesting conclusions for this particular observation\footnote{We are currently investigating this point with the help of Enedis}. Aside from that particular point, the obtained clusters are essentially shaped to distinguish the profiles built by Enedis (see Table \ref{tab:cluster-profil}). We believe that choosing a less penalizing BIC criterion could reveal less discriminating effects than the profile, such as a regional effect for instance.

\begin{table}
    \centering
    \caption{Contingency table between clusters and Enedis profiles.}
    \label{tab:cluster-profil}
    \begin{tabular}{c|c||c|c|c}
            \multicolumn{2}{c||}{}&ENT/PRO&RES1, 3 or 4 & Other RES\\\hline\hline
  \multirow{3}{*}{\rotatebox{90}{$\!$\small{Cluster}}}&1&     341&       36&       182\\\cline{2-5}
  &2&       0&       0&        231\\\cline{2-5}
  &3&       0&         0&        1\\
    \end{tabular}
\end{table}

\subsection{Results Discussion}
Let us begin by recalling the two periods of lockdown observed in France in 2020: the first one lasted from March $17^{th}$ until May $11^{th}$, while the second one started on October $30^{th}$ and ended on December $15^{th}$. For cluster 1, which is mainly composed of enterprises and professionals, we observe a unique breakpoint that matches with the beginning of the first lockdown. We note that for this cluster no other breakpoints are observed. We can make two assumptions regarding this point: (1) it indicates that the regime change operated by these consumers did not return to \textit{normal} \footnote{here normal refers to prior to the crisis} after the beginning of the crisis; (2) the return to this so-called \textit{normal} regime did not occur at the same moment for all individuals. For the second cluster, corresponding mainly to the Residential profiles, we observe two breakpoints, each appearing during the lockdown period (mid April and end of November). Once again, two hypothesis are in order: (1) a delayed effect of each lockdown or (2) the impact of outdoor temperatures, particularly high in France between mid-April and the end of November. Finally, for the last cluster, the fact that it only contains one observation does not allow us to draw significant conclusions. We can only note that most of the breakpoints matches breakpoints from both cluster 1 and 2. Additional results are provided in Figures \ref{SI:fig:cluster1}, \ref{SI:fig:cluster2} and \ref{SI:fig:cluster3} in the Supplementary.

\section{Conclusion}
\label{sec:conc}
In this paper, we define a novel model to analyze multivariate functional data by performing clustering and segmentation simultaneously. We derive an EM algorithm where the maximization step is carried out by dynamic programming.   From a theoretical point of view, we establish the identifiability and the consistency of the proposed model. We apply this model on synthetic data to control the behavior and validate our theoretical statements. We also demonstrate the usefulness of our model on real electricity consumption data by focusing on the year 2020 corresponding to the outbreak of the COVID pandemic.

This work can be further extended in different directions. On the estimation part, three developments can be considered. To speed up the parameter estimation, we could adapt the pruned dynamic programming algorithm proposed by \cite{rigaill2015pruned} and recently improved by~\cite{maidstone2017optimal}, who proposed to prune the set of candidate change-points rather than looking at all possible cases. In order to circumvent this speed problem, a second alternative could be to replace the current maximization step with a group-lasso procedure for the parameter estimation. For instance \cite{brault2017efficient} proposed to transform the problem into an equivalent estimate of a linear regression whose parameter of interest would be sparse and thus to use the LASSO (Least Absolute Shrinkage and Selection Operator) procedures. The main challenge of these methods being the regularization parameter of the LASSO method and the multiplication of the number of estimated breakpoints. Moreover, in relation to the real observed data, a non parametric extension based on rank statistics (see \cite{brault2018nonparametric}) can be studied.
Finally, regarding the model selection criterion proposed in the experimental part, a theoretical study of the latter would allow us to define the most appropriate form for the penalty. Indeed, the first results appear to indicate that the proposed version is too penalizing, not allowing to highlight fine-grained information. 

\appendix

\section{Identifiability}
\label{App:identifiability}
To prove Theorem~\ref{Th:Identifiability}, we need the following lemma:
\begin{lem}[Identifiability for the breakpoints model]\label{lem:Identifiability} We define a $L$-breakpoints model of $p$-dimensional spherical Gaussian of length $d$ with the parameters $0=T_0<T_1<\cdots<T_L<T_{L+1}=d$, $\bmu\in\mathcal{M}_{(L+1)\times p}\left(\setR\right)$ and $\bsigma\in\left(\setR_{\star}^+\right)^{L+1}$ with the following likelihood for every $\bx\in\setR^{d\times p}$:
\[p(\bx;\bT,\bmu,\bsigma)=\prod_{\ell=0}^L\prod_{j=T_{\ell}+1}^{T_{\ell+1}}\prod_{r=1}^pf\left(x_{jr};\mu_{lr},\sigma_{\ell r}\right)\]
where $f$ is the density of an univariate Gaussian distribution. We assume that:
\begin{enumerate}[label=(ID.\alph*)]
    \item For every $\ell\in\{0,\ldots,L\}$,  there exists $r\in\{1,\ldots,p\}$ such that:
    \[\sigma_{\ell r}\neq\sigma_{\ell+1,r} \text{ or }\mu_{\ell r}\neq\mu_{\ell+1,r}.\]
    \item We have $d\geq L+1$.
\end{enumerate}
Under the assumptions (ID.a) and (ID.b), the model is identifiable.
\end{lem}
\begin{proof}
Let two parameters $\left(L,\bT,\bmu,\bsigma\right)$ and $\left(L',\bT',\bmu',\bsigma'\right)$ satisfying the assumptions~(ID.a) and~(ID.b), and $\bX$ and $\bX'$ the random matrices depending of each parameter. We assume that $\bX$ and $\bX'$ have the same distribution. To prove that the parameters are equal, we use the characteristic function defined for every $\bxi\in\setR^{d\times p}$ by:
\[\Phi_{\bX}\left(\bxi\right)=\Esp{e^{\Scal{i\bxi}{\bX}}}\]
where $\Scal{\cdot}{\cdot}$ is the scalar product and $i$ is the imaginary number, satisfying $i^2=-1$. In our case, we have, for every $\bxi\in\setR^{d\times p}$:
\begin{align*}
\Phi_{\bX}\left(\bxi\right)&=\Esp{e^{\Scal{i\bxi}{\bX}}} =\prod_{j=1}^d\Esp{e^{\Scal{i\bxi_j}{\bX_j}}}\\
&=\prod_{\ell=0}^L\prod_{j=T_{\ell}+1}^{T_{\ell+1}}\prod_{r=1}^p\Esp{e^{\Scal{i\xi_{jr}}{X_{jr}}}} =\prod_{\ell=0}^L\prod_{j=T_{\ell}+1}^{T_{\ell+1}}\prod_{r=1}^p\left(e^{i\xi_{jr}\mu_{\ell r}-\frac{\sigma_{\ell r}^2\xi_{jr}^2}{2}}\right)\\
&=\exp\left[\sum_{\ell=0}^L\sum_{j=T_{\ell}+1}^{T_{\ell+1}}\sum_{r=1}^p\left(i\xi_{jr}\mu_{\ell r}-\frac{\sigma_{\ell r}^2\xi_{jr}^2}{2}\right)\right]\\
&=\exp\left[i\sum_{\ell=0}^L\sum_{r=1}^p\mu_{\ell r}\sum_{j=T_{\ell}+1}^{T_{\ell+1}}\xi_{jr}-\sum_{\ell=0}^L\sum_{r=1}^p\frac{\sigma_{\ell r}^2}{2}\sum_{j=T_{\ell}+1}^{T_{\ell+1}}\xi_{jr}^2\right].
\end{align*}
Let $\bw$ be the vector of $\{0,\ldots,L\}^{d}$ with $w_{j}=\ell$ if and only if the $T_{\ell}+1\leq j\leq T_{\ell+1}$. Then, 
\begin{align*}
\Phi_{\bX}\left(\bxi\right)&=\exp\left[i\sum_{j=1}^d\sum_{r=1}^p\xi_{jr}\mu_{w_{j}r}-\frac{1}{2}\sum_{j=1}^d\sum_{r=1}^p\xi_{jr}^2\sigma_{w_jr}^2\right].
\end{align*}
The distribution of $\bX$ and $\bX'$ are equal if and only if the characteristic functions are equal: for all $\bxi\in\setR^{d\times p}$, we have
\begin{eqnarray*}
&&\!\!\!\!\!\!\!\!\!\Phi_{\bX}\left(\bxi\right)=\Phi_{\bX'}\left(\bxi\right)\\
&\Leftrightarrow&i\sum_{j=1}^d\sum_{r=1}^p\xi_{j r}\mu_{w_{j}r}-\frac{1}{2}\sum_{j=1}^d\sum_{r=1}^p\xi_{j r}^2\sigma_{w_j,r}^2=\\
&&\quad\quad\quad\quad\quad i\sum_{j=1}^d\sum_{r=1}^p\xi_{j r}\mu'_{w'_{j}r}-\frac{1}{2}\sum_{j=1}^d\sum_{r=1}^p\xi_{j r}^2{\sigma'_{w'_j,r}}^2\\
&\Leftrightarrow&i\sum_{j=1}^d\sum_{r=1}^p\xi_{j r}\left(\mu_{w_{j}r}-\mu'_{w'_{j}r}\right)-\frac{1}{2}\sum_{j=1}^d\sum_{r=1}^p\xi_{j r}^2\left(\sigma_{w_jr}^2-{\sigma'}_{w'_jr}^2\right)=0.
\end{eqnarray*}
A polynomial is null if and only all coefficients are null:
\[\text{for every $j\in\{1,\ldots,d\}$, }\left\{\begin{array}{l}
 \text{for every $r\in\{1,\ldots,p\}$, }\mu_{w_{j}r}=\mu'_{w'_{j}r}\\
 \text{ and }\sigma_{w_jr}^2={\sigma'}_{w'_j}^2. 
\end{array}\right.\]
By the assumption (ID.b), we know that there is at least one observation by segment (by definition of $\bT$ and $\bT'$) then, by definition also, $w_1=0=w'_1$ which implies that
\[\text{for every $r\in\{1,\ldots,p\}$, }\mu_{0,r}=\mu'_{0,r}
 \text{ and }\sigma_{0,r}^2={\sigma'}_{0,r}^2.\]
Then, if we assume that $T_1\neq T'_1$ (for example, $T'_1>T_1$), we observe that $w_{T_1+1}=0$ and $w'_{T_1+1}=1$ but we know that:
\begin{eqnarray*}
&&\!\!\!\!\!\!\!\!\!
 \text{for every $r\in\{1,\ldots,p\}$, }\mu_{w_{T_1+1},r}=\mu'_{w'_{T_1+1},r}
 \text{ and }\sigma_{w_{T_1+1}}^2={\sigma'}_{w'_{T_1+1}}^2\\
 &\Leftrightarrow& \text{for every $r\in\{1,\ldots,p\}$, }\mu_{0,r}=\mu'_{1r}
 \text{ and }\sigma_{0,r}^2={\sigma'}_{1r}^2\\
 &\Leftrightarrow& \text{for every $r\in\{1,\ldots,p\}$, }\mu'_{0,r}=\mu'_{1r}
 \text{ and }{\sigma'}_{0,r}^2={\sigma'}_{1r}^2\\
\end{eqnarray*}
by the previous results. This affirmation contradicts the assumption (ID.b), therefore $T_1= T'_1$.
We then continue with $w_{T_1+1}=1=w'_{T_1+1}$ and by an identical reasoning, we show that
\[\text{for every $r\in\{1,\ldots,p\}$, }\mu_{1r}=\mu'_{1r}
 \text{ and }\sigma_{1r}^2={\sigma'}_{1r}^2\]
and so on until the segment $\left[T_{\min{L,L'}}+1;T_{\min{L,L'}+1}\right]$. If $L=L'$, the proof is finished since we prove that all the parameters are identical.
If $L\neq L'$, for example $L>L'$, we observe that $w_{T_{L'+1}+1}=L+1$ (since by the previous reasoning $T_L=T'_{L}$ and $T_{L'}<T_{L'+1}<d$) and $w'_{T_{L'+1}+1}=L'=L$ then, by the same reasoning, this implies:
\begin{eqnarray*}
&&\!\!\!\!\!\!\!\!\!
 \text{for every $r\in\{1,\ldots,p\}$, }\mu_{w_{T_{L'+1}+1},r}=\mu'_{w'_{T_{L'+1}+1},r}\\
 && 
 \text{ and }\sigma_{w_{T_{L'+1}+1},r}^2={\sigma'}_{w'_{T_{L'+1}+1},r}^2\\
 &\Leftrightarrow& \text{for every $r\in\{1,\ldots,p\}$, }\mu_{L'+1,r}=\mu'_{L',r}
 \text{ and }\sigma_{L'+1,r}^2={\sigma'}_{L',r}^2\\
 &\Leftrightarrow& \text{for every $r\in\{1,\ldots,p\}$, }\mu_{L'+1,r}=\mu_{L',r}
 \text{ and }{\sigma}_{L'+1,r}^2={\sigma}_{L',r}^2
\end{eqnarray*}
which again contradicts the assumption~(ID.b).
Then, $L=L'$ and, the parameters being identical, the model is identifiable. 
\end{proof}

To prove Theorem~\ref{Th:Identifiability}, we start by observing that the assumptions~(ID.1) and~(ID.3) imply the assumptions~(ID.a) and~(ID.b) of Lemma~\ref{lem:Identifiability} and, with the assumptions~(ID.2), we have that the distribution of each cluster is unique.
As the image of the distribution functions by any isomorphism defined on the vector subspace generated by the set of distribution functions is a free family in the arrival space, then model is identifiable (see~\cite{droesbeke2013modeles}).

\section{Consistency}
\label{app:Cons}

\subsection{Notations and first results}
Let $\bzs$ be the true partition, and  $\bz$ any partition.

We denote $\RK$ the matrix of size $K$ about contingency of intersection of the two clusterings: for all $(k,k')\in\{1\ldots,K\}^2$,
\begin{align}
\RK_{k,k'}=\sum_{i=1}^n\zs_{ik}z_{i,k'}.
\label{RK}
\end{align}
Particularly, remark that marginals give the contingency table: 
\begin{align}
&\forall k\in\{1,\ldots,K\}, \RK_{k,+}=\zs_{+,k} \label{Rk+}\\
&\forall k'\in\{1,\ldots,K\}, \RK_{+,k'}=\zs_{k',+}\label{R+k'}.
\end{align}

Also, if the two partitions are equal up to label switching, then $\RK$ is diagonal up to a permutation of rows and columns. 

Similarly, we denote $\NLk$ the same matrix for the segments: for all 
$(k,k')\in\{1\ldots,K\}^2$, for all $(\ell,\ell')\in\{1\ldots,L_k\}\times\{1\ldots,L_{k'}\}$,
\begin{align}
\NLk_{\ell,\ell'} 
&=\left|{D^\star}_{\ell}^{k}\bigcap{D}_{\ell'}^{k'}\right|
\label{NLk}
\end{align}
where
\begin{align*}
{D}_{\ell'}^{k'}&=\left\{j\in\{1,\ldots,d\}\left|T_{k'\ell'}+1\leq j\leq T_{k',\ell'+1}\right.\right\},\\
{D^\star}_{\ell'}^{k'}&=\left\{j\in\{1,\ldots,d\}\left|T^\star_{k'\ell'}+1\leq j\leq T^\star_{k',\ell'+1}\right.\right\}. 
\end{align*}
Remark that 
\begin{align}
\sum_{\ell_1=1}^{L_{k_1}}\NL{k_1}{k}_{\ell_1,\ell}
=\left|{D^\star}_{\ell}^{k}\right|.
\label{sumNk=Dk}
\end{align}

For every $(\btheta,\bT,\bz)$, we denote 
\[\Fntz=\log\frac{p(\bY|\bz;\btheta,\bT)}{p(\bY|\bzs;\bthetas,\bTs)},\]
and its expectation 
\[\Gntz=\Espc{\bY|\bzs;\bthetas,\bTs}{\Fntz}.\]

In the following, we study the Kullback-Leibler divergence between two Gaussian distributions with variance 1. As the distributions only depend on the respective means, we denote $\KL{\mu}{\mu'}$ this divergence. Remark that 
\begin{align}
\label{KL}
KL(\mu,\mu') = \frac12 (\mu-\mu')^2.
\end{align}
\begin{prop}[Distribution of $\Fntz$]\label{prop:Fn}
For all $\btheta$, $\bT\in\mathcal{T}_{\tau_{\min}}$ and $\bz$, we get
\begin{eqnarray*}
\Fntz  
&\sim&\mathcal{N}\left(\Gntz,\Gntz\right).
\end{eqnarray*}
with 
\begin{eqnarray*}
&&\!\!\!\!\!\!\Gntz\\
&=&\!\!\!\!-\!\sum_{k=1}^K\!\sum_{k'=1}^K\!\sum_{\ell=1}^{\Lk}\!\sum_{\ell'=1}^{L_{k'}}\KL{\mus_{k\ell}}{\mu_{k'\ell'}}\RK_{k,k'}\NLk_{\ell,\ell'}.
\end{eqnarray*}
\end{prop}

The proof stands in Section \ref{sec:tools}.

We are also particularly interested in the maximum of $g_n$ 
\begin{eqnarray*}
\Lambdat(\bz,\bT)&=&\max_{\btheta\in\Theta}\Gntz\\
\end{eqnarray*}

\begin{prop}\label{prop:lambdat}
For all $\bz, \bT$, 
\begin{eqnarray*}
&&\!\!\!\!\!\!\!\!\!\!\!\!\Lambdat(\bz,\bT)\\
&=&\!\!\!\!-\frac{1}{2}\sum_{k=1}^K\sum_{\ell=1}^{L_{k}}\frac{1}{\left|{D^\star}_{\ell}^{k}\right|\bzs_{+,k}}\sum_{k_1=1}^K\sum_{\ell_1=1}^{L_{k_1}}\RK_{k_1,k}\NL{k_1}{k}_{\ell_1,\ell}\\
&&\times\sum_{k_2=1}^K\sum_{\ell_2=1}^{L_{k_2}}\RK_{k_2,k}\NL{k_2}{k}_{\ell_2,\ell}\KL{\mus_{k_1,\ell_1}}{\mus_{k_2,\ell_2}}\\
\end{eqnarray*}
\end{prop}

The proof stands in Section \ref{sec:tools}.

Then, we need a measure of the minimal difference between two different parameters, computed through the Kullback-Leibler divergence. 
\begin{defi}[Minimal Kullback-Leibler divergence]
Let $\delta(\bthetas)$ be the minimal nonzero Kullback-Leibler divergence: 
\[\delta(\bthetas)=\min_{\underset{\mus_{k_1,\ell_1}\neq \mus_{k_2,\ell_2}}{k_1,\ell_1,k_2,\ell_2}}\KL{\mus_{k_1,\ell_1}}{\mus_{k_2,\ell_2}}>0.\]
\end{defi}

\subsection{Proof of Proposition \ref{prop:part_equi}: equivalent partitions}
Let $\sigma\in\mathfrak{S}(\left\{1,\ldots,K\right\})$ be a permutation. We say that $(\bpi, \bmu,\bT)$ has a \emph{symmetry} for   $\sigma$ if we have, for all  $k\in\{1,\ldots,K\}$:
\[\pi_{\sigma(k)}=\pik,\,\forall\ell\in\{1,\ldots,\Lk\},\mu_{\sigma(k),\ell}=\mu_{k\ell}\text{ and }T_{\sigma(k),\ell}=T_{k\ell}.\]
We denote  $\text{Sym}(\btheta, \bT)$ the set of permutations such that  $\left(\btheta,\bT\right)$ has a symmetry. 

Remark that under Assumption (ID.3), we have
\[\#\text{Sym}\left(\btheta,\bT\right)=1,\]
which makes the next computations easier than in \cite{brault2020consistency}, where we directly get, in our particular case,  $\{\sigma\in\mathfrak{S}_{\bz,\bzs}\} = \{(\btheta',\bT')\sim(\btheta,\bT)\}$.

Let $\sigma$  be a permutation such that for all  $i\in\{1,\ldots,n\}$ and $k\in\{1,\ldots,K\}$, we have
\[z_{i\sigma(k)}=\zsik,\]
$\mathfrak{S}_{\bz,\bzs}$ the set of all possible permutations and we denote $\bz^{(\sigma)} :=\left(z_{i\sigma(k)}\right)_{ik}$. We have: 
\[p(\bY,\bz;\btheta,\bT)=p\left(\bY,{\bzs}^{(\sigma)};\btheta,\bT\right)=p\left(\bY,\bzs;\btheta^{(\sigma)},\bT^{(\sigma)}\right).\]
If   $\sigma\in\text{Sym}(\btheta)$, it leads to 
$p(\bY,\bz;\btheta,\bT)=p\left(\bY,\bzs;\btheta,\bT\right)$. By summing, we get
\begin{eqnarray*}
\sum_{\bz\sim\bzs}p(\bY,\bz;\btheta,\bT)&=&\sum_{\sigma\in\mathfrak{S}_{\bz,\bzs}}p\left(\bY,{\bzs}^{(\sigma)};\btheta,\bT\right)\\
&=&\sum_{\sigma\in\mathfrak{S}_{\bz,\bzs}}p\left(\bY,\bzs;\btheta^{(\sigma)},\bT^{(\sigma)}\right)\\
&=&\sum_{(\btheta',\bT')\sim(\btheta,\bT)}p\left(\bY,\bzs;\btheta',\bT'\right).\\
\end{eqnarray*}
However, the function $\btheta\mapsto p\left(\bY,\bzs;\btheta,\bT\right)$ is unimodal and maximal for the maximum of the complete likelihood. As the estimator is consistent as soon as we have the true partition, under Assumption (ID.1), we have $p\left(\bY,\bzs;\btheta,\bT\right)=\mathcal{O}_P\left[p\left(\bY,\bzs;\bthetas,\bTs\right)\right]$ when $\theta$ is in a neighborhood of $\thetas$ and $p\left(\bY,\bzs;\btheta,\bT\right)=o_P\left[p\left(\bY,\bzs;\bthetas,\bTs\right)\right]$ elsewhere. If  $\btheta$ is close to  $\bthetas$, the set of equivalent $\btheta'$ but not symmetric are far and we get:
\[p\left(\bY,\bzs;\btheta',\bT'\right)=o_P\left[p\left(\bY,\bzs;\bthetas,\bTs\right)\right].\]
Then,
\[\sum_{(\btheta',\bT')\sim(\btheta,\bT)}\frac{p\left(\bY,\bzs;\btheta',\bT'\right)}{p\left(\bY,\bzs;\bthetas,\bTs\right)}=\max_{(\btheta',\bT')\sim(\btheta,\bT)}\frac{p\left(\bY,\bzs;\btheta',\bT'\right)}{p\left(\bY,\bzs;\bthetas,\bTs\right)}\left[1+o_P(1)\right].\]

\subsection{Proof of Proposition \ref{prop:contrib_local}: partitions that are close}

\begin{lem}
Assume Ass. (ID.4.s) with parameter $c>0$.
Let the event 
$$\Omega_1(c) = \left\{\bzs\in\mathcal{Z}\left|\forall k\in\{1,\ldots,K\},\,\zs_{+,k}\geq nc/2\right.\right\},$$
 where $\mathcal{Z}$ is the set of all possible partitions. 
Then, 
\[\mathbb{P}_{\thetas}\left(\overline{\Omega_1(c)}\right)\leq K e^{-\frac{nc^2}{2}}.\]
\end{lem}

\begin{proof}
First, remark that for all $k\in\{1,\ldots,K\}$, $\Zs_{+,k} \sim \mathcal{B}(\pis_k)$ with $\pis_k>c$ according to Assumption (ID.4.s). Then, using Hoeffding inequality,

\begin{eqnarray*}
\mathbb{P}_{\thetas}\left(\overline{\Omega_1(c)}\right)&=&\mathbb{P}_{\thetas}\left(\bigcup_{k=1}^K\{\Zs_{+,k}< nc/2\}\right)
\leq \sum_{k=1}^K\mathbb{P}_{\thetas}\left(\Zs_{+,k}< nc/2\right)\\
&\leq&\sum_{k=1}^K\mathbb{P}_{\thetas}\left(\Zs_{+,k}< n\pis_k/2\right)
\leq \sum_{k=1}^K\mathbb{P}_{\thetas}\left(\Zs_{+,k}-n\pis_k< -n\pis_k/2\right)\\
&\leq&\sum_{k=1}^K\exp\left[-\frac{2\left(-n\pis_k/2\right)^2}{\sum_{i=1}^n(1-0)^2}\right]
\leq \sum_{k=1}^K\exp\left[-\frac{2n^2{\pis_k}^2}{4n}\right]\\
&\leq&\sum_{k=1}^K\exp\left[-\frac{n{\pis_k}^2}{2}\right]
\leq \sum_{k=1}^K\exp\left[-\frac{nc^2}{2}\right]
\leq K\exp\left[-\frac{nc^2}{2}\right].
\end{eqnarray*}
\end{proof}

In the next lemma, we split the balls into slices. 
\begin{lem}[Upper bounding of  $\Fntz$]
Under Assumptions (ID.3.s), (C.1) and (C.2), for all $r\geq 1/n$, for all $\btheta\in\bTheta$ and $\bz\in\mathcal{Z}$ such that\linebreak$d_{0,\sim}(\bz,\bzs)=rn$, we have:
\begin{equation}\label{eq:cor:majFn}
    \Fntz\leq-\frac{d\tau_{\min}\delta\left(\bthetas\right)}{2}rn\left[1+o_P(1)\right].
\end{equation}
\end{lem}

\begin{proof}
First, we remark that for all $\btheta$, $\bT\in\mathcal{T}_{\tau_{\min}}$ and $\bz$, we have:
\begin{eqnarray*}
\Fntz&\leq&\Fntz-\Gntz+\Lambdat\left(\bz\right)\\
&\leq&\Fntz-\Gntz-\frac{d\tau_{\min}\delta(\bthetas)}{2}d_{0,\sim}(\bz,\bzs)\\
&\leq&\Fntz-\Gntz-\frac{d\tau_{\min}\delta(\bthetas)}{2}rn.
\end{eqnarray*}
By Lemma~\ref{prop:Fn},  $\Fntz-\Gntz$ is a centered Gaussian.  
We also know that 
\begin{align*}
(\mu_{k\ell} - \mu_{k'\ell'})^2 &\leq (\text{Diam} \Theta)^2\\
\sum_{\ell=1}^{\Lk}\sum_{\ell'=1}^{L_{k'}}\NLk_{\ell,\ell'} &\leq d\\
\sum_{k=1}^K\sum_{k'=1}^K\RK_{k,k'}&\leq n,
\end{align*}
then using Lemma~\ref{lem:inegalitemesouviensplus}, we get for all $t>0$:
\begin{eqnarray*}
\!\!\!\!\!\!\!\!\!\!\!\!\Prob{\Fntz-\Gntz\geq t} \leq\exp\left\{-\frac{t^2}{2\text{Diam}(\bTheta)^2nd}\right\},
\end{eqnarray*}
then for all partition such that  $d_{0,\sim}(\bz,\bzs)=rn$, we have:
\begin{eqnarray*}
&&\Prob{\Fntz-\Gntz\geq \frac{dd_{0,\sim}(\bz,\bzs)\tau_{\min}\delta\left(\bthetas\right)}{2}}\\
&\leq&\exp\left\{-\frac{\frac{d^2\tau_{\min}^2d_{0,\sim}(\bz,\bzs)^2\delta\left(\bthetas\right)^2}{2^2}}{2\text{Diam}(\bTheta)^2nd}\right\}\\
&\leq&\exp\left\{-\frac{\delta\left(\bthetas\right)^2d\tau_{\min}d_{0,\sim}(\bz,\bzs)r}{8\text{Diam}(\bTheta)^2}\right\} \underset{n, d \rightarrow +\infty}\longrightarrow 0, 
\end{eqnarray*}
where $d_{0,\sim}$ depends on $n$. Then,
\begin{eqnarray*}
\Fntz&\leq& o_P\left[\frac{d\tau_{\min}d_{0,\sim}(\bz,\bzs)\delta\left(\bthetas\right)}{2}\right]-\frac{d\tau_{\min}\delta\left(\bthetas\right)}{2}rn\\
&\leq&-\frac{d\tau_{\min}\delta\left(\bthetas\right)}{2}rn\left[1+o_P(1)\right]
\end{eqnarray*}
\end{proof}

\begin{lem}[Having different partitions]
For all $c>0$, considering the event $\Omega_1(c)$, we have for all $\tilde{c}\leq c/4$ and $\bz\in\mathcal{B}\left(\bzs,\tilde{c}\right)$:
\begin{equation}\label{eq:demo:borneProbaParition}
    \frac{p\left(\bz;\btheta,\bT\right)}{p\left(\bzs;\bthetas,\bTs\right)}\leq \mathcal{O}_P\left(1\right)e^{M_{c/4}\left\|\bz-\bzs\right\|_0}.
\end{equation}
\end{lem}

\begin{proof}
Considering $\Omega_1(c)$: for all  $k\in\{1,\ldots,K\}$, we have $\zs_{+,k}\geq nc/2$. As $\bz\in\mathcal{B}\left(\bzs,\tilde{c}\right)$ with $\tilde{c}\leq c/4$ then for all $k\in\{1,\ldots,K\}$, we have  $z_{+,k}\geq nc/4$.

For a partition $\bz$, let $\pih{\bz}$ be the maximum of $\bpi\mapsto p\left(\bz;\bpi\right)$:  for all $k\in\{1,\ldots,K\}$,
\[\pih{\bz}_k=\frac{z_{+,k}}{n}.\]

Then, we have
\begin{eqnarray*}
\frac{p\left(\bz;\btheta,\bT\right)}{p\left(\bzs;\bthetas,\bTs\right)}&=&\frac{p\left(\bz;\bpi\right)}{p\left(\bzs;\bpis\right)}
\leq\frac{p\left(\bz;\pih{\bz}\right)}{p\left(\bzs;\pih{\bzs}\right)}\times\frac{p\left(\bz;\pih{\bzs}\right)}{p\left(\bzs;\bpis\right)},
\end{eqnarray*}
by definition of $\pih{\bz}$. Using \cite{brault2020consistency}[Lemma D.2] we get:
\begin{eqnarray*}
\log\left[\frac{p\left(\bz;\pih{\bz}\right)}{p\left(\bzs;\pih{\bzs}\right)}\right]
&=&\log p\left(\bz;\pih{\bz}\right)-\log p\left(\bzs;\pih{\bzs}\right)\\
&=&\sum_{i=1}^n\sum_{k=1}^K\zik\log\pih{\bz}_k-\sum_{i=1}^n\sum_{k=1}^K\zsik\log\pih{\bzs}_k\\
&=&n\sum_{k=1}^K\left[\pih{\bz}_k\log\pih{\bz}_k-\pih{\bzs}_k\log\pih{\bzs}_k\right].
\end{eqnarray*}
Let $H(\bpi)=\sum_{k=1}^K\pik\log\pik$. This function is differentiable and we can use the mean value theorem to the function $x\mapsto -x\log x$ with derivative $x\mapsto \log x+1$. Then, for all $k\in\{1,\ldots,K\}$, there exists $\kappa_k\in]\pisk;\pik[$ such that 
\begin{eqnarray*}
&&\frac{\pik\log\pik-\pisk\log\pisk}{\pik-\pisk}=-\log \kappa_k-1\\
\Rightarrow&&\left|\frac{\pik\log\pik-\pisk\log\pisk}{\pik-\pisk}\right|=\left|-\log \kappa_k-1\right|.
\end{eqnarray*}
However, as $\bpi$ and $\bpis$ are regular enough, we know that $\kappa_k\in]c/4;1-c/4[$ then there exists a constant $M_{c/4}$ such that
\begin{eqnarray*}
\left|\frac{\pik\log\pik-\pisk\log\pisk}{\pik-\pisk}\right| \leq \frac{M_{c/4}}{2}.
\end{eqnarray*}
By summing, we get
\begin{eqnarray*}
\left|H(\bpi)-H(\bpis)\right|&=&\left|\sum_{k=1}^K\pik\log\pik-\sum_{k=1}^K\pisk\log\pisk\right|\\
&\leq&\sum_{k=1}^K\frac{M_{c/4}}{2}\left|\pik-\pisk\right|
\leq\frac{M_{c/4}}{2}\left\|\bpi-\bpis\right\|_1.
\end{eqnarray*}
Finally, we have
\begin{eqnarray*}
&&\!\!\!\!\!\!\!\!\left|\log\left[\frac{p\left(\bz;\pih{\bz}\right)}{p\left(\bzs;\pih{\bzs}\right)}\right]\right|\\
&=&n\left|H(\pih{\bz})-H(\pih{\bzs})\right|
\leq\frac{nM_{c/4}}{2}\left\|\pih{\bz}-\pih{\bzs}\right\|_1\\
&\leq&\frac{nM_{c/4}}{2}\sum_{k=1}^K\left|\frac{z_{+,k}}{n}-\frac{\zs_{+,k}}{n}\right|\\
&\leq&\frac{M_{c/4}}{2}\sum_{i=1}^n\sum_{k=1}^K\left|z_{ik}-\zs_{ik}\right|
\leq \frac{M_{c/4}}{2}\sum_{i=1}^n\sum_{k=1}^K\mathds{1}_{\{z_{ik}\neq\zs_{ik}\}}\\
&\leq&  M_{c/4}\left\|\bz-\bzs\right\|_0.
\end{eqnarray*}
Indeed, if $i$ does not belong to the true partition, there are two nonzero terms. 

On the other side, by the law of large numbers, we know that
$\pih{\bzs}=\mathcal{O}_P\left(\bpis\right)$ and as $\bpis$ is regular, we have:
\begin{equation}\label{eq:demo:bornefrac2}
    \frac{p\left(\bz;\pih{\bzs}\right)}{p\left(\bzs;\bpis\right)}=\mathcal{O}_P\left(1\right).
\end{equation}
\end{proof}

\begin{proof}[Proof of Proposition \ref{prop:contrib_local}]
Considering $\tilde{c}<c/4$, $\btheta\in\bTheta$ and $\bT\in\mathcal{T}_{\tau_{\min}}$, considering the event $\Omega_1(c)$:
\begin{eqnarray*}
&&\!\!\!\!\!\!\!\!\!\!\!\!\sum_{\underset{\bz\nsim\bzs}{\bz\in\mathcal{B}\left(\bzs;\tilde{c}\right)}}p\left(\bY,\bz;\btheta,\bT\right)\\
&=&\sum_{\underset{\bz\nsim\bzs}{\bz\in\mathcal{B}\left(\bzs;\tilde{c}\right)}}p\left(\bY|\bz;\btheta,\bT\right)p\left(\bz;\btheta,\bT\right)\frac{p\left(\bY|\bzs;\bthetas,\bTs\right)}{p\left(\bY|\bzs;\bthetas,\bTs\right)}\\
&=&\sum_{\underset{\bz\nsim\bzs}{\bz\in\mathcal{B}\left(\bzs;\tilde{c}\right)}}\frac{p\left(\bY,\bzs;\bthetas,\bTs\right)}{p\left(\bzs;\bthetas,\bTs\right)}p\left(\bz;\btheta,\bT\right)\frac{p\left(\bY|\bz;\btheta,\bT\right)}{p\left(\bY|\bzs;\bthetas,\bTs\right)}\\
&=&p\left(\bY,\bzs;\bthetas,\bTs\right)\sum_{\underset{\bz\nsim\bzs}{\bz\in\mathcal{B}\left(\bzs;\tilde{c}\right)}}\frac{p\left(\bz;\btheta,\bT\right)}{p\left(\bzs;\bthetas,\bTs\right)}e^{\Fntz}\\
&=&p\left(\bY,\bzs;\bthetas,\bTs\right)\sum_{\underset{\bz\nsim\bzs}{\bz\in\mathcal{B}\left(\bzs;\tilde{c}\right)}}\mathcal{O}_P\left(1\right)e^{M_{c/4}\left\|\bz-\bzs\right\|_0}e^{-\frac{d\tau_{\min}\delta\left(\bthetas\right)}{2}\left\|\bz-\bzs\right\|_0\left[1+o_P(1)\right]}
\end{eqnarray*}
using Eqs. ~\eqref{eq:cor:majFn} and ~\eqref{eq:demo:borneProbaParition}. Then, 
\begin{eqnarray*}
&&\!\!\!\!\!\!\!\!\!\!\!\!\sum_{\underset{\bz\nsim\bzs}{\bz\in\mathcal{B}\left(\bzs;\tilde{c}\right)}}p\left(\bY,\bz;\btheta,\bT\right)\\
&\leq&p\left(\bY,\bzs;\bthetas,\bTs\right)\mathcal{O}_P\left(1\right)\sum_{R=1}^{\left[\tilde{c}n\right]}\binom{n}{R}K^Re^{M_{c/4}R-\frac{d\tau_{\min}\delta\left(\bthetas\right)}{2}R\left[1+o_P(1)\right]}
\end{eqnarray*}
because there are at most  $\binom{n}{R}K^R$ possible partitions at distance  $R$, and then
\begin{eqnarray*}
&&\!\!\!\!\!\!\!\!\!\!\!\!\sum_{\underset{\bz\nsim\bzs}{\bz\in\mathcal{B}\left(\bzs;\tilde{c}\right)}}p\left(\bY,\bz;\btheta,\bT\right)\\
&\leq&p\left(\bY,\bzs;\bthetas,\bTs\right)\mathcal{O}_P\left(1\right)\sum_{R=1}^{\left[\tilde{c}n\right]}\binom{n}{R}\left(e^{\log K+M_{c/4}-\frac{d\tau_{\min}\delta\left(\bthetas\right)}{2}\left[1+o_P(1)\right]}\right)^R\\
&\leq&p\left(\bY,\bzs;\bthetas,\bTs\right)\mathcal{O}_P\left(1\right)\left(1+e^{\log K+M_{c/4}-\frac{d\tau_{\min}\delta\left(\bthetas\right)}{2}\left[1+o_P(1)\right]}\right)^n\\
&\leq&p\left(\bY,\bzs;\bthetas,\bTs\right)\mathcal{O}_P\left(1\right)\exp\left\{n\log\left[1+e^{\log K+M_{c/4}-\frac{d\tau_{\min}\delta\left(\bthetas\right)}{2}\left[1+o_P(1)\right]}\right]\right\}\\
&\leq&p\left(\bY,\bzs;\bthetas,\bTs\right)\mathcal{O}_P\left(1\right)\exp\left\{ne^{\log K+M_{c/4}-\frac{d\tau_{\min}\delta\left(\bthetas\right)}{2}\left[1+o_P(1)\right]}[1+o\left(1\right)]\right\}\\
&\leq&p\left(\bY,\bzs;\bthetas,\bTs\right)\mathcal{O}_P\left(1\right)\\
&&\quad\times\exp\left\{e^{-d\left\{\frac{\log(n)}{d}+\frac{\log K+M_{c/4}}{d}-\frac{\delta\left(\bthetas\right)}{2}\left[1+o_P(1)\right]\right\}}[1+o\left(1\right)]\right\}\\
&\leq&p\left(\bY,\bzs;\bthetas,\bTs\right)o_P(1)\text{ using Assumption (C.3).}
\end{eqnarray*}
As this is true for all  $\theta\in\bTheta$ and all $\bT\in\mathcal{T}_{\tau_{\min}}$, this is also true for the maximum. 
\end{proof}

\subsection{Proof of Proposition \ref{prop:part_eloign}: partitions that are far}

\begin{prop}[Separability]\label{prop:separability}
Assume (ID.3.s) and (C.2). There exists  $R>0$ and a constant $B(R)$  such that 
\begin{align*}
\max_{\bT} \Lambdat(\bz, \bT)&\leq -\frac{d\tau_{\min}\delta(\bthetas) \max_k L_k}{2}d_{0,\sim}(\bz,\bzs) \hspace{1cm} \text{ for } \bz\in \mathcal{B}\left(\bzs;R\right);\\
 \max_{\bT} \Lambdat(\bz,\bT)&\leq -B(R)dn \hspace{4.4cm} \text{ for }\bz\notin\mathcal{B}\left(\bzs;R\right).
\end{align*}
\end{prop}
\begin{proof}

If $\bz\notin\mathcal{B}\left(\bzs;R\right)$, because the two partitions $\bz$ and $\bzs$ are far from each other (at least a radius $R>0$), there exists a constant $B(R)$ such that the second inequality holds. 

Else, let assume that $\bz\in\mathcal{B}\left(\bzs;R\right)$.
From Assumptions (ID.3.s) and (C.2), for all  $k\neq k'$, there exists at least  $\tau_{\min}d$ columns such that the Kullback-Leibler divergence is strictly positive, then
$$ \sum_{\ell_2=1}^{L_{k_2}}\NL{k_2}{k}_{\ell_2,\ell}\KL{\mus_{k_1,\ell_1}}{\mus_{k_2,\ell_2}} \geq d\tau_{\min}\delta(\bthetas).$$

Then, according to Proposition~\ref{prop:lambdat}, we get
\begin{eqnarray*}
&&\!\!\!\!\!\!\!\!\!\!\!\!\Lambdat(\bz, \bT)\\
&\leq&-\frac{d\tau_{\min}\delta(\bthetas)}{2}\sum_{k=1}^K\sum_{\ell=1}^{L_{k}}\frac{1}{\left|{D^\star}_{\ell}^{k}\right|\bzs_{+,k}}\sum_{k_1=1}^K\RK_{k_1,k} \\
&&\quad\times\sum_{\ell_1=1}^{L_{k_1}}\NL{k_1}{k}_{\ell_1,\ell}\\
&&\quad\quad\times\left[\sum_{k_2=1}^K\RK_{k_2,k}-\RK_{k_1,k}\right]\\
&\leq&-\frac{d\tau_{\min}\delta(\bthetas)}{2}\sum_{k=1}^K\sum_{\ell=1}^{L_{k}}\frac{1}{\bzs_{+,k}}\sum_{k_1=1}^K\RK_{k_1,k}\\
&&\quad\quad\times\left[\sum_{k_2=1}^K\RK_{k_2,k}-\RK_{k_1,k}\right], 
\end{eqnarray*}
using Eq. \eqref{sumNk=Dk} in the last inequality. 
Then, using Eq. \eqref{R+k'}, and the fact that 
$\RK_{k_1,k} \leq \bzs_{+,k}$,

\begin{eqnarray*}
&&\!\!\!\!\!\!\!\!\!\!\!\!\Lambdat(\bz,\bT)\\
&\leq&-\frac{d\tau_{\min}\delta(\bthetas)\max_kL_k}{2}\sum_{k=1}^K\frac{1}{\bzs_{+,k}}\sum_{k_1=1}^K\RK_{k_1,k}\\
&&\quad\times\left[\sum_{k_2=1}^K\RK_{k_2,k}-\RK_{k_1,k}\right]\\
&\leq&-\frac{d\tau_{\min}\delta(\bthetas)\max_kL_k}{2}\left[n-\sum_{k=1}^K\sum_{k_1=1}^K\frac{\RK_{k_1,k}^2}{\bzs_{+,k}}\right]\\
&\leq&-\frac{d\tau_{\min}\delta(\bthetas)\max_kL_k}{2}\left[n-\sum_{k=1}^K\sum_{k_1=1}^K\RK_{k_1,k}\right]\\
&\leq&-\frac{d\tau_{\min}\delta(\bthetas)\max_kL_k}{2}d_{0,\sim}(\bz,\bzs).
\end{eqnarray*}

\end{proof}

\begin{lem}[Large deviation]\label{propo:largeDeviation}
Under assumption (C.1), and for all $\varepsilon_{nd}<1/\sqrt{2}$, we get:
\begin{align*}
&\Prob{\sup_{\btheta,\bT,\bz}\left[\Fntz-\Lambdat(\bz,\bT)\right]\geq \text{Diam}(\bTheta)\sqrt{nd}K^2+4\varepsilon_{nd}\text{Diam}(\bTheta)nd}\\
\leq& K^n\exp\left(-\varepsilon_{nd}^2nd\right).
\end{align*}
\end{lem}
\begin{proof}
By definition of $\Fntz$ and $\Gntz$:
\begin{eqnarray*}
&&\!\!\!\!\!\!\!\!\!\!\!\!\!\!\Fntz-\Lambdat(\bz,\bT)\leq\Fntz-\Gntz\\
&\leq&-\sum_{i=1}^n\sum_{k=1}^K\sum_{k'=1}^K\zs_{ik}z_{i,k'}\sum_{\ell=1}^{\Lk}\sum_{\ell'=1}^{L_{k'}}\\
&&\quad\times\sum_{j\in{D^\star}_{\ell}^{k}\bigcap{D}_{\ell'}^{k'}}\left(Y_{ij}-\Espc{\bY|\bzs;\bthetas,\bTs}{Y_{ij}}\right)\left(\mus_{k\ell}-\mu_{k'\ell'}\right)\\
&\leq&\sup_{\underset{\|\bGamma\|_\infty\leq \text{Diam}(\bTheta)}{\bGamma\in\mathbb{R}^{K\times K}}}\sum_{k=1}^K\sum_{k'=1}^K\sum_{\ell=1}^{\Lk}\sum_{\ell'=1}^{L_{k'}}\Gamma_{k,k'}\sum_{i=1}^n\zs_{ik}z_{i,k'}\sum_{j\in{D^\star}_{\ell}^{k}\bigcap{D}_{\ell'}^{k'}}\left(\mus_{k\ell}-Y_{ij}\right)\\
&\leq&\sup_{\underset{\|\bGamma\|_\infty\leq \text{Diam}(\bTheta)}{\bGamma\in\mathbb{R}^{K\times K}}}\sum_{k=1}^K\sum_{k'=1}^K\Gamma_{k,k'}\sum_{i=1}^n\zs_{ik}z_{i,k'}\sum_{\ell=1}^{\Lk}\sum_{j\in{D^\star}_{\ell}^{k}}\left(\mus_{k\ell}-Y_{ij}\right).
\end{eqnarray*}
Let 
$$ W_{k,k'}= \sum_{i=1}^n\zs_{ik}z_{i,k'}\sum_{\ell=1}^{\Lk}\sum_{j\in{D^\star}_{\ell}^{k}}\left(\mus_{k\ell}-Y_{ij}\right) \sim \mathcal{N}\left(0,d\RK_{k',k}\right). $$

Then, 
using \cite{brault2020consistency}[Proposition C.4], for $\bGamma \in \mathbb{R}^{K\times K}$ such that $\|\bGamma\|_\infty\leq \text{Diam}(\bTheta)$,
\[\Espc{\bthetas}{\sum_{k=1}^K\sum_{k'=1}^K\Gamma_{k,k'} W_{k,k'}}\leq \text{Diam}(\bTheta)\sqrt{nd}K^2.\]

Then, for all nonnegative sequence $(\varepsilon_{nd})_{(n,d)\in\left(\mathbb{N}^\star\right)^2}$ satisfying\linebreak $4\varepsilon_{nd}\text{Diam}(\bTheta)nd\leq\frac{8\text{Diam}(\bTheta)^2nd}{2\sqrt{2}\text{Diam}(\bTheta)}=2\sqrt{2}\text{Diam}(\bTheta)nd$, using Lemma \ref{lem:inegalitemesouviensplus}:

\begin{eqnarray*}
&&\!\!\!\!\!\!\!\!\!\!\!\!\Prob{\sum_{k=1}^K\sum_{k'=1}^K\Gamma_{k,k'} W_{k,k'} - \Espc{\bthetas}{\sum_{k=1}^K\sum_{k'=1}^K\Gamma_{k,k'} W_{k,k'}}  \geq 4\varepsilon_{nd}\text{Diam}(\bTheta)nd}\\
&\leq&\exp\left[-\frac{\left(4\varepsilon_{nd}\text{Diam}(\bTheta)nd\right)^2}{16\text{Diam}(\bTheta)^2nd}\right]
\leq\exp\left(-\varepsilon_{nd}^2nd\right).
\end{eqnarray*}
Then, 
\begin{eqnarray*}
&&\!\!\!\!\!\!\!\!\!\!\!\!\Prob{\Fntz-\Lambdat(\bz,\bT)\geq \text{Diam}(\bTheta)\sqrt{nd}K^2+4\varepsilon_{nd}\text{Diam}(\bTheta)nd}\\
&\leq&\exp\left(-\varepsilon_{nd}^2nd\right).
\end{eqnarray*}
This inequality holds for all $(\btheta,\bT)$ and $\bz$, then we get:
\begin{eqnarray*}
&&\mathbb{P}\left(\sup_{\btheta,\bT,\bz}\left[\Fntz-\Lambdat(\bz,\bT)\right]\geq\right.\\
&&\quad\quad\left. \text{Diam}(\bTheta)\sqrt{nd}K^2+4\varepsilon_{nd}\text{Diam}(\bTheta)nd\right)\\
&\leq&\sum_{\bz\in\mathcal{Z}}\mathbb{P}\left(\sup_{\btheta,\bT}\left[\Fntz-\Lambdat(\bz,\bT)\right]\geq\right.\\
&&\quad\quad\left. \text{Diam}(\bTheta)\sqrt{nd}K^2+4\varepsilon_{nd}\text{Diam}(\bTheta)nd\right)\\
&\leq&\sum_{\bz\in\mathcal{Z}}\exp\left(-\varepsilon_{nd}^2nd\right) \leq K^n\exp\left(-\varepsilon_{nd}^2nd\right).
\end{eqnarray*}
\end{proof}

Let $R>0$. Then, according to Proposition \ref{prop:separability},
\begin{align*}
\Lambdat(\bz,\bT)&\leq -B(C)nd\text{ if }\bz\notin\mathcal{B}\left(\bzs;R\right)\\
\Lambdat(\bz,\bT)&\leq-\frac{d\tau_{\min}\delta(\bthetas)}{2}d_{0,\sim}(\bz,\bzs)\\
&\leq-\frac{d\tau_{\min}\delta(\bthetas)}{2}nR_{nd}\text{ if }\bz\in\mathcal{B}\left(\bzs;R\right)\backslash\mathcal{B}\left(\bzs;R_{nd}\right).
\end{align*}

Using Lemma~\ref{propo:largeDeviation}, with $\varepsilon_{nd}=\min\left(\frac{\delta(\bthetas)\tau_{\min}R_{nd}}{16},1/\sqrt{2}\right)$, we get with probability   $1-K^n\exp\left(-\varepsilon_{nd}^2nd\right)$,
\begin{eqnarray*}
&&\Fntz-\Lambdat(\bz,\bT)+\Lambdat(\bz,\bT)\\
&\leq&\Fntz-\Lambdat(\bz,\bT)-\frac{d\tau_{\min}\delta(\bthetas)}{2}nR_{nd}\\
&\leq&\text{Diam}(\bTheta)\sqrt{nd}K^2+4\varepsilon_{nd}\text{Diam}(\bTheta)nd-\frac{d\tau_{\min}\delta(\bthetas)}{2}nR_{nd}\\
&\leq&\text{Diam}(\bTheta)\sqrt{nd}K^2+4\frac{\delta(\bthetas)\tau_{\min}R_{nd}}{16}\text{Diam}(\bTheta)nd-\frac{d\tau_{\min}\delta(\bthetas)}{2}nR_{nd}\\
&\leq& \text{Diam}(\bTheta)\sqrt{nd}K^2+\frac{d\tau_{\min}\delta(\bthetas)}{4}nR_{nd}-\frac{d\tau_{\min}\delta(\bthetas)}{2}nR_{nd}\\
&\leq&\text{Diam}(\bTheta)\sqrt{nd}K^2-\frac{d\tau_{\min}\delta(\bthetas)}{4}nR_{nd}\\
&\leq&-dnR_{nd}\frac{\delta(\bthetas)\tau_{\min}}{4}\left[1-\frac{4\text{Diam}(\bTheta)K^2}{\sqrt{nd}R_{nd}\tau_{\min}\delta(\bthetas)}\right]\\
&\leq&-dnR_{nd}\frac{\tau_{\min}\delta(\bthetas)}{4}
\end{eqnarray*}
 for $n$ and $d$ large enough. 

 We also know that
 \begin{eqnarray*}
{p(\bzs;\bthetas,\bTs)}
&\leq& e^{-\log \prod_{i=1}^n\prod_{k=1}^K\pik^{\zs_{ik}}} 
\leq e^{-\sum_{i=1}^n\sum_{k=1}^K\zs_{ik}\log \pik} \\
&\leq& e^{-\sum_{i=1}^n\sum_{k=1}^K\zs_{ik}\log c} \text{ by Assumption (C.1),}\\
&\leq& e^{-\sum_{i=1}^n\log c}
\leq e^{n \log \frac{1}{c}}.
\end{eqnarray*}
Then, it leads to

\begin{eqnarray*}
\!\!\!\!\!\!\!\!\!\!\sum_{z\notin\mathcal{B}\left(\bzs;R_{nd}\right)}p(\bY,\bz;\btheta,\bT)&=&\sum_{z\notin\mathcal{B}\left(\bzs;R_{nd}\right)}p(\bY|\bz;\btheta,\bT)p(\bz;\btheta,\bT)\\
&=&p(\bY|\bzs;\bthetas,\bTs)\sum_{z\notin\mathcal{B}\left(\bzs;R_{nd}\right)}p(\bz;\btheta,\bT)\frac{p(\bY|\bz;\btheta,\bT)}{p(\bY|\bzs;\bthetas,\bTs)}\\
&=&p(\bY|\bzs;\bthetas,\bTs)\sum_{z\notin\mathcal{B}\left(\bzs;R_{nd}\right)}p(\bz;\btheta,\bT)e^{\Fntz}\\
&\leq&p(\bY,\bzs;\bthetas,\bTs)e^{-n\left[dR_{nd}\frac{\delta(\bthetas)\tau_{\min}}{4}-\log \frac{1}{c}\right]}.
\end{eqnarray*}
Then, with probability  $1-K^n\exp\left(-\varepsilon_{nd}^2nd\right)$, under Assumption  (C.1), we get:
\[\sum_{z\notin\mathcal{B}\left(\bzs;R_{nd}\right)}p(\bY,\bz;\btheta,\bT)=p(\bY,\bzs;\bthetas,\bTs)o_P(1)\]
for all $\btheta \in \bTheta$ and $\bT\in\mathcal{T}_{\tau_{\min}}$.

\subsection{Tools and details}
\label{sec:tools}
\begin{lem}[Chernoff's lemma]
\label{lem:inegalitemesouviensplus}
Let $Z\sim\mathcal{N}\left(0,\sigma^2\right)$. Then, for all $t>0$:
\[\Prob{Z\geq t}\leq e^{-\frac{t^2}{2\sigma^2}}.\]
\end{lem}

\begin{proof}[Proof of Proposition \ref{prop:Fn}]
By definition, 
\begin{eqnarray*}
\Fntz&=&-\sum_{i=1}^n\sum_{k=1}^K\sum_{k'=1}^K\zs_{ik}z_{i,k'}\sum_{\ell=1}^{\Lk}\sum_{\ell'=1}^{L_{k'}}\sum_{j\in{D^\star}_{\ell}^{k}\bigcap{D}_{\ell'}^{k'}}\log\left[\frac{\varphi\left(Y_{ij};\mus_{k\ell}\right)}{\varphi\left(Y_{ij};\mu_{k'\ell'}\right)}\right].
\end{eqnarray*}
The computation gives:
\begin{eqnarray}
\log\left[\frac{\varphi\left(Y_{ij};\mus_{k\ell}\right)}{\varphi\left(Y_{ij};\mu_{k'\ell'}\right)}\right]&=&\log\frac{\frac{1}{\sqrt{2\pi}}e^{-\frac{1}{2}\left(Y_{ij}-\mus_{k\ell}\right)^2}}{\frac{1}{\sqrt{2\pi}}e^{-\frac{1}{2}\left(Y_{ij}-\mu_{k'\ell'}\right)^2}}\nonumber\\
&=&-\frac{1}{2}\left(Y_{ij}-\mus_{k\ell}\right)^2+\frac{1}{2}\left(Y_{ij}-\mu_{k'\ell'}\right)^2\nonumber\\
&=&Y_{ij}\left(\mus_{k\ell}-\mu_{k'\ell'}\right)-\frac{1}{2}\left({\mus_{k\ell}}^2-\mu_{k'\ell'}^2\right) \label{ratioLogLik}
\end{eqnarray}
which leads to 
\begin{eqnarray*}
&&\Fntz=-\sum_{i=1}^n\sum_{k=1}^K\sum_{k'=1}^K\zs_{ik}z_{i,k'}\sum_{\ell=1}^{\Lk}\sum_{\ell'=1}^{L_{k'}}\sum_{j\in{D^\star}_{\ell}^{k}\bigcap{D}_{\ell'}^{k'}}Y_{ij}\left(\mus_{k\ell}-\mu_{k'\ell'}\right)\\
&&-\frac{1}{2}\sum_{k=1}^K\sum_{k'=1}^K\sum_{\ell=1}^{\Lk}\sum_{\ell'=1}^{L_{k'}}\RK_{k,k'}\NLk_{\ell,\ell'}\left({\mus_{k\ell}}^2-\mu_{k'\ell'}^2\right).
\end{eqnarray*}
We recognize the linear combination of independent Gaussian variable, then the distribution of $\Fntz$ is also Gaussian. The computation of the expectation and the variance are straightforward.  
\end{proof}

\begin{proof}[Proof of Proposition \ref{prop:lambdat}]
The function $\btheta\mapsto\Gntz$ is maximal for
\[\widehat{\mu}_{k'\ell'}=\frac{\sum_{k=1}^K\sum_{\ell=1}^{\Lk}\RK_{k,k'}\NLk_{\ell,\ell'}\mus_{k\ell}}{\sum_{k=1}^K\sum_{\ell=1}^{\Lk}\RK_{k,k'}\NLk_{\ell,\ell'}}.\]
Indeed, the Kullback-Leibler divergence is equal to 
\[\KL{\mu}{\mu'}=\frac{1}{2}\left(\mu-\mu'\right)^2,\]
and the maximum is get by differentiating in each value. Then,
\begin{align*}
\mus_{k_1,\ell_1}-\widehat{\mu}_{k\ell}
&= \mus_{k_1,\ell_1}-\frac{\sum_{k_2=1}^K\sum_{\ell_2=1}^{L_{k_2}}\RK_{k_2,k}\NL{k_2}{k}_{\ell_2,\ell}\mus_{k_2,\ell_2}}{\sum_{k_2=1}^K\sum_{\ell_2=1}^{L_{k_2}}\RK_{k_2,k}\NL{k_2}{k}_{\ell_2,\ell}}\\
&= \frac{\sum_{k_2=1}^K\sum_{\ell_2=1}^{L_{k_2}}\RK_{k_2,k}\NL{k_2}{k}_{\ell_2,\ell}\left(\mus_{k_1,\ell_1}-\mus_{k_2,\ell_2}\right)}{\sum_{k_2=1}^K\sum_{\ell_2=1}^{L_{k_2}}\RK_{k_2,k}\NL{k_2}{k}_{\ell_2,\ell}}.
\end{align*}
However, if ones want to take the square of the previous formulae, we get 
\begin{align*}
\left(\mus_{k_1,\ell_1}-\mus_{k_2,\ell_2}\right) \left(\mus_{k_1,\ell_1}-\mus_{k'_2,\ell'_2}\right) 
&= {\mus_{k_1,\ell_1}}^2-\mus_{k_1,\ell_1}\mus_{k_2,\ell_2}-\mus_{k'_2,\ell'_2}\left(\mus_{k_1,\ell_1}-\mus_{k_2,\ell_2}\right),
\end{align*}
then it leads to
\begin{align*}
(\mus_{k_1,\ell_1}-\widehat{\mu}_{k\ell})^2
&= \frac{\sum_{k_2=1}^K\sum_{\ell_2=1}^{L_{k_2}} \sum_{k'_2=1}^K\sum_{\ell'_2=1}^{L_{k'_2}} \RK_{k_2,k} \NL{k_2}{k}_{\ell_2,\ell}}{(\sum_{k_2=1}^K\sum_{\ell_2=1}^{L_{k_2}}\RK_{k_2,k}\NL{k_2}{k}_{\ell_2,\ell})^2}\\
&\quad\times \RK_{k'_2,k}\NL{k'_2}{k}_{\ell'_2,\ell} \left({\mus_{k_1,\ell_1}}^2-\mus_{k_1,\ell_1}\mus_{k_2,\ell_2} \right)\\
& -  \frac{\sum_{k_2=1}^K\sum_{\ell_2=1}^{L_{k_2}} \sum_{k'_2=1}^K\sum_{\ell'_2=1}^{L_{k'_2}} \RK_{k_2,k} \NL{k_2}{k}_{\ell_2,\ell}}{(\sum_{k_2=1}^K\sum_{\ell_2=1}^{L_{k_2}}\RK_{k_2,k}\NL{k_2}{k}_{\ell_2,\ell})^2}\\
&\quad\times \RK_{k'_2,k}\NL{k'_2}{k}_{\ell'_2,\ell} \mus_{k'_2,\ell'_2}\left(\mus_{k_1,\ell_1}-\mus_{k_2,\ell_2}\right)\\
&= \frac{\sum_{k_2=1}^K\sum_{\ell_2=1}^{L_{k_2}}  \RK_{k_2,k} \NL{k_2}{k}_{\ell_2,\ell}  \left({\mus_{k_1,\ell_1}}^2-\mus_{k_1,\ell_1}\mus_{k_2,\ell_2} \right)}{\sum_{k_2=1}^K\sum_{\ell_2=1}^{L_{k_2}}\RK_{k_2,k}\NL{k_2}{k}_{\ell_2,\ell}}\\
& -  \frac{\sum_{k_2=1}^K\sum_{\ell_2=1}^{L_{k_2}} \sum_{k'_2=1}^K\sum_{\ell'_2=1}^{L_{k'_2}} \RK_{k_2,k} \NL{k_2}{k}_{\ell_2,\ell}}{(\sum_{k_2=1}^K\sum_{\ell_2=1}^{L_{k_2}}\RK_{k_2,k}\NL{k_2}{k}_{\ell_2,\ell})^2}\\
&\quad\times \RK_{k'_2,k}\NL{k'_2}{k}_{\ell'_2,\ell} \mus_{k'_2,\ell'_2}\left(\mus_{k_1,\ell_1}-\mus_{k_2,\ell_2}\right).
\end{align*}
When considering $\Lambdat(\bz,\bT)$, we are summing with respect to $k_1,\ell_1$ as well, and the second term becomes $0$. Then, using the explicit form of $F_n$ given in the proof of Proposition \ref{prop:Fn}, it leads to
\begin{eqnarray*}
&&\!\!\!\!\!\!\!\!\!\!\!\!\Lambdat(\bz,\bT)\\
&=&-\frac{1}{2}\sum_{k=1}^K\sum_{\ell=1}^{L_{k}}
\frac{\sum_{k_1=1}^K\sum_{\ell_1=1}^{L_{k_1}} \sum_{k_2=1}^K\sum_{\ell_2=1}^{L_{k_2}} \RK_{k_1,k}\RK_{k_2,k}}{\sum_{k_2=1}^K\sum_{\ell_2=1}^{L_{k_2}}\RK_{k_2,k}\NL{k_2}{k}_{\ell_2,\ell}}\\
&&\quad\times\NL{k_1}{k}_{\ell_1,\ell}\NL{k_2}{k}_{\ell_2,\ell}\frac{1}{2}\left({\mus_{k_1,\ell_1}}-\mus_{k_2,\ell_2}\right)^2.
\end{eqnarray*}
Then, using \ref{Rk+}, \ref{sumNk=Dk} and \ref{KL}, 
we finally get the desired formulae. 
\end{proof}

\begin{acks}[Acknowledgments]
This work has been partially supported by MIAI@Grenoble Alpes (ANR-19-P3IA-0003). All the computations presented in this paper were performed using the GRICAD infrastructure (\url{https://gricad.univ-grenoble-alpes.fr}), which is supported by Grenoble research communities.
\end{acks}

\begin{supplement}

\section*{Additional results on the simulation}
In this part, additional results for the section~\ref{sec:simu} are presented. 
Hereafter, we present additional results for the simulation part. Table \ref{tab:NCE} presents the NCE obtained with both MixSeg and the simple mixture model. Figure \ref{fig:time:d} presents the computation time recorded for MixSeg when $n$, $d$ and $\alpha$ vary.

\begin{table}[!htpb]
\caption{NCE MixSeg vs Simple mixture model}\label{tab:NCE}
    \begin{tabular}{ccccc}
    \toprule
    \multicolumn{5}{c}{NCE $\downarrow$}\\
         \midrule
         Setting&Model&0.1&0.2&1\\\midrule
 \multirow{2}{*}{(100,50)}&MixSeg&$\mathbf{0.76(0.16)}$&$\mathbf{0.34(0.18)}$&0 (0)\\
 &Simple&0.84 (0.1)&0.67 (0.18)&0.02 (0.09)\\\midrule
 \multirow{2}{*}{(100,100)}&MixSeg&$\mathbf{0.60(0.17)}$&$\mathbf{0.19(0.17)}$&0 (0)\\
 &Simple&0.75 (0.12)&0.51 (0.18)&0.01 (0.07)\\\midrule
 \multirow{2}{*}{(1000,50)}&MixSeg&$\mathbf{0.56(0.15)}$&$\mathbf{0.17(0.12)}$&0 (0)\\
 &Simple&0.82 (0.11)&0.62 (0.14)&0.00015 (0.00045)\\\midrule
 \multirow{2}{*}{(1000,100)}&MixSeg&$\mathbf{0.36(0.14)}$&$\mathbf{0.04(0.03)}$&0 (0)\\
 &Simple&0.72 (0.14)&0.51 (0.16)&0 (0)\\
 \bottomrule
    \end{tabular}
\end{table}

\begin{figure}[!ht]
    \centering
    \includegraphics[width=\linewidth]{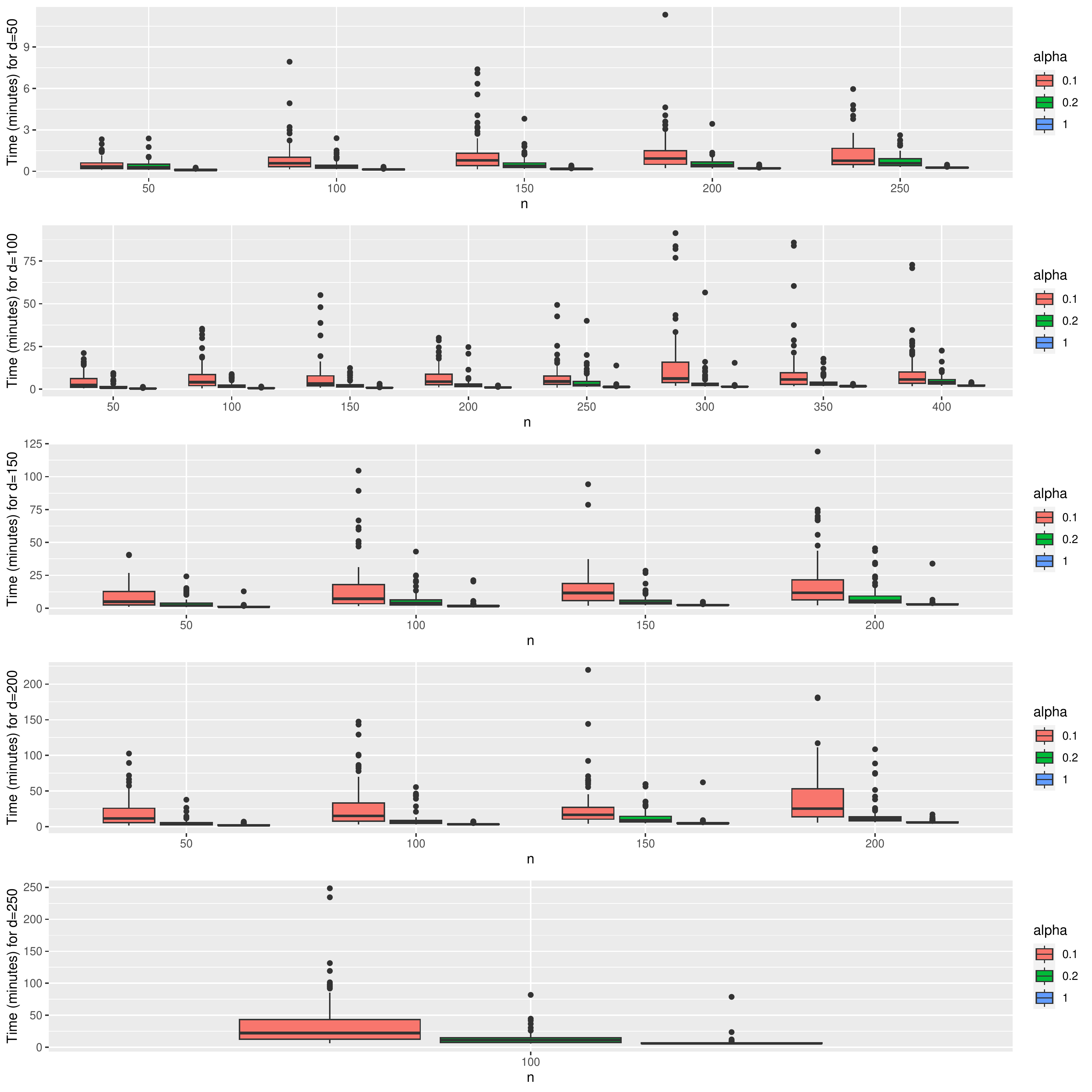}
    \caption{Computation time for varying $n$ and $d$}
    \label{fig:time:d}
\end{figure}

\FloatBarrier

\section*{Additional figures for Enedis data analysis}
In this part, three additional figures (\ref{SI:fig:cluster1}, \ref{SI:fig:cluster2} and~\ref{SI:fig:cluster3}) for the section~\ref{sec:real_data} are presented.

\begin{figure}[!ht]
    \centering
    \includegraphics[width=\linewidth]{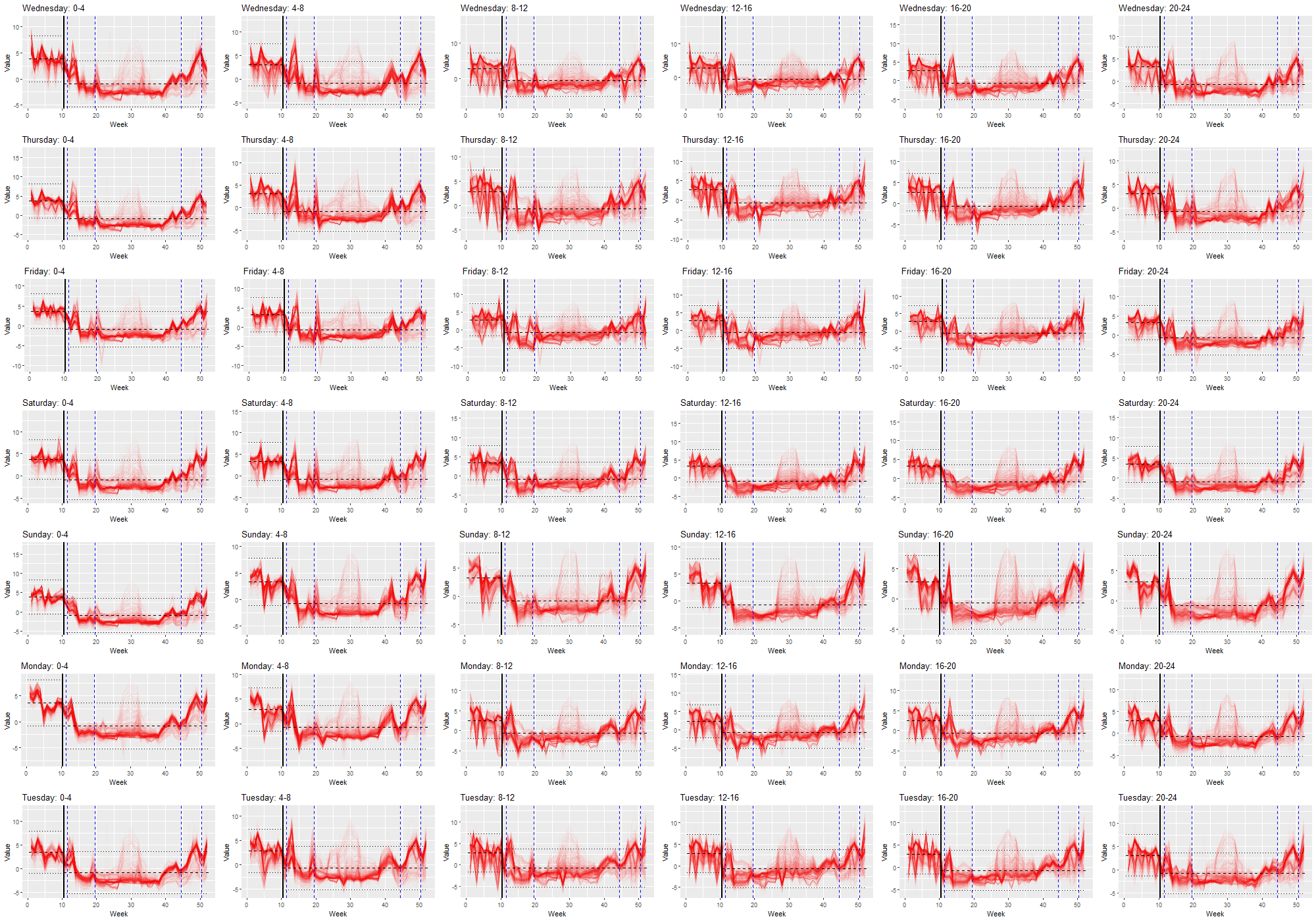}
    \caption{Evolution of the coefficients grouped per day (rows) and time slots (columns) for the cluster~1. For each graphic, the evolution week per week of each curve is represented in color, the vertical black line corresponds to the estimated break-point, the horizontal black dashed line at the mean, and the horizontal black dotted lines to the confident intervals. The vertical blue dashed lines correspond at the beginning and the end of both lock-down.}
    \label{SI:fig:cluster1}
\end{figure}

\begin{figure}[!ht]
    \centering
    \includegraphics[width=\linewidth]{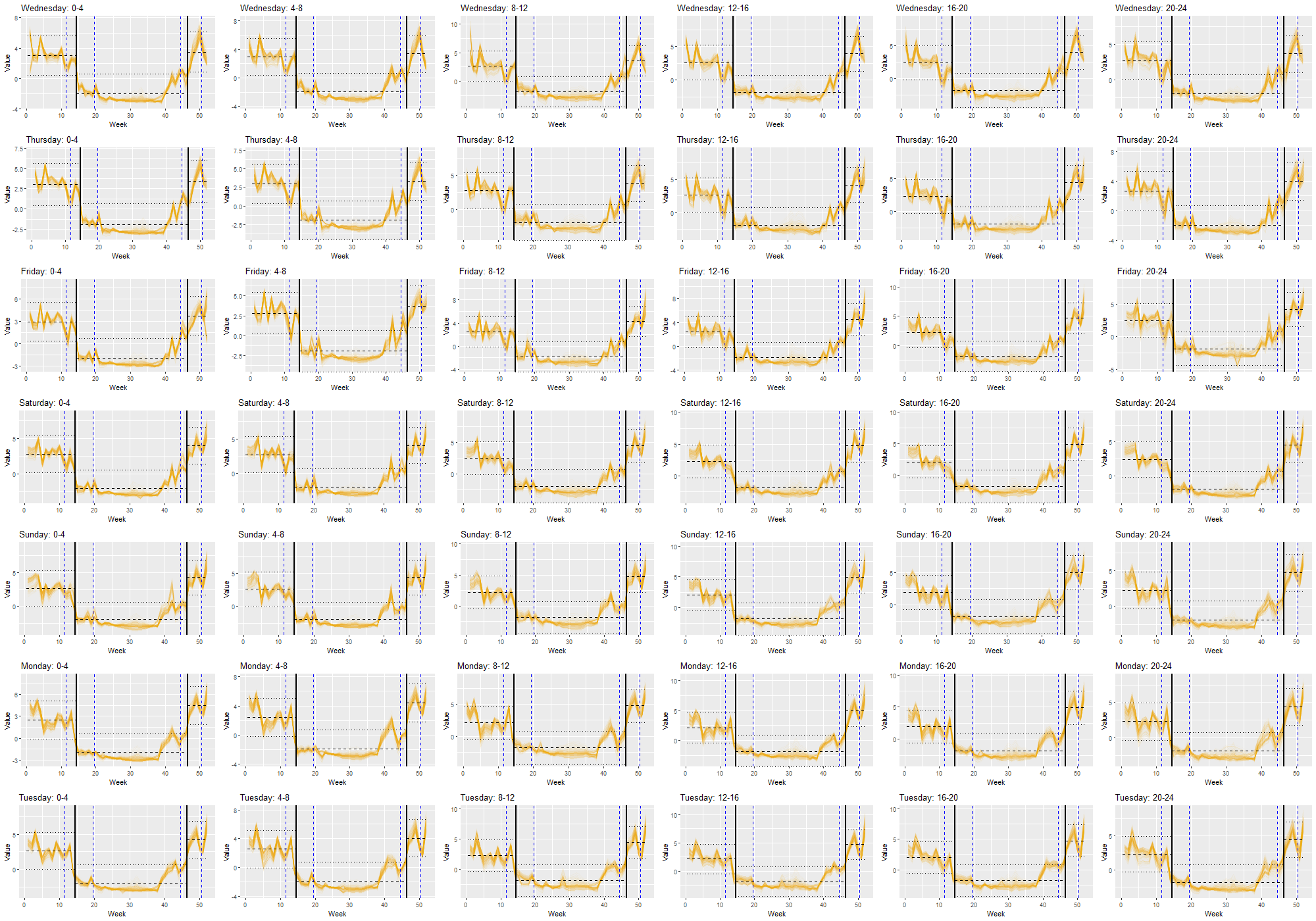}
    \caption{Evolution of the coefficients grouped per day (rows) and time slots (columns) for the cluster~2. For each graphic, the evolution week per week of each curve is represented in color, the vertical black lines correspond at the estimated break-points, the horizontal black dashed line at the mean, and the horizontal black dotted lines at the confident intervals. The vertical blue dashed lines correspond at the beginning and the end of both lock-down.}
    \label{SI:fig:cluster2}
\end{figure}

\begin{figure}[!ht]
    \centering
    \includegraphics[width=\linewidth]{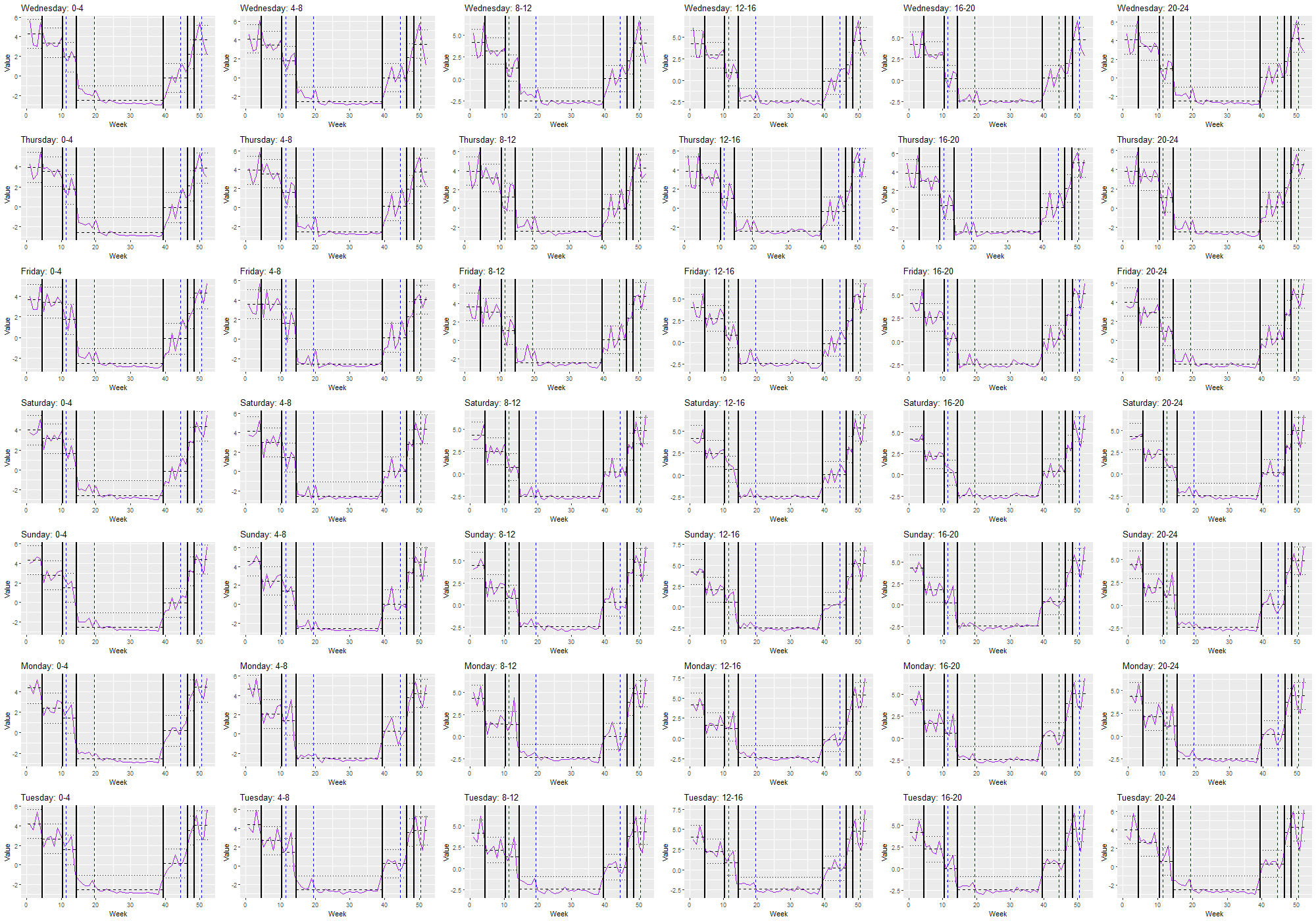}
    \caption{Evolution of the coefficients grouped per day (rows) and time slots (columns) for the cluster~3. For each graphic, the evolution week per week of each curve is represented in color, the vertical black lines correspond at the estimated break-points, the horizontal black dashed line at the mean, and the horizontal black dotted lines at the confident intervals. The vertical blue dashed lines correspond at the beginning and the end to both lock-down.}
    \label{SI:fig:cluster3}
\end{figure}

\textcolor{white}{Plus rien ne marche ?}\\
\end{supplement}

\clearpage
\bibliographystyle{imsart-nameyear}
\bibliography{_references}

\begin{thebibliography}{50}

\bibitem[\protect\citeauthoryear{Alon et~al.}{2003}]{alon2003}
\begin{binproceedings}[author]
\bauthor{\bsnm{Alon},~\bfnm{J.}\binits{J.}},
  \bauthor{\bsnm{Sclaroff},~\bfnm{S.}\binits{S.}},
  \bauthor{\bsnm{Kollios},~\bfnm{G.}\binits{G.}} \AND
  \bauthor{\bsnm{Pavlovic},~\bfnm{V.}\binits{V.}}
(\byear{2003}).
\btitle{Discovering clusters in motion time-series data}.
In \bbooktitle{2003 IEEE Computer Society Conference on Computer Vision and
  Pattern Recognition, 2003. Proceedings.}
\end{binproceedings}
\endbibitem

\bibitem[\protect\citeauthoryear{Bellman and Kalaba}{1957}]{bellman1957dynamic}
\begin{barticle}[author]
\bauthor{\bsnm{Bellman},~\bfnm{Richard}\binits{R.}} \AND
  \bauthor{\bsnm{Kalaba},~\bfnm{Robert}\binits{R.}}
(\byear{1957}).
\btitle{Dynamic programming and statistical communication theory}.
\bjournal{Proceedings of the National Academy of Sciences of the United States
  of America}
\bvolume{43}
\bpages{749}.
\end{barticle}
\endbibitem

\bibitem[\protect\citeauthoryear{Bickel et~al.}{2013}]{bickel2013asymptotic}
\begin{barticle}[author]
\bauthor{\bsnm{Bickel},~\bfnm{Peter}\binits{P.}},
  \bauthor{\bsnm{Choi},~\bfnm{David}\binits{D.}},
  \bauthor{\bsnm{Chang},~\bfnm{Xiangyu}\binits{X.}} \AND
  \bauthor{\bsnm{Zhang},~\bfnm{Hai}\binits{H.}}
(\byear{2013}).
\btitle{Asymptotic normality of maximum likelihood and its variational
  approximation for stochastic blockmodels}.
\end{barticle}
\endbibitem

\bibitem[\protect\citeauthoryear{Bouveyron and Jacques}{2011}]{bouveyron2011}
\begin{barticle}[author]
\bauthor{\bsnm{Bouveyron},~\bfnm{Charles}\binits{C.}} \AND
  \bauthor{\bsnm{Jacques},~\bfnm{Julien}\binits{J.}}
(\byear{2011}).
\btitle{{Model-based Clustering of Time Series in Group-specific Functional
  Subspaces}}.
\bjournal{Advances in Data Analysis and Classification}
\bpages{281-300}.
\end{barticle}
\endbibitem

\bibitem[\protect\citeauthoryear{Bouveyron et~al.}{2017}]{bouveyron:2017}
\begin{barticle}[author]
\bauthor{\bsnm{Bouveyron},~\bfnm{Charles}\binits{C.}},
  \bauthor{\bsnm{Bozzi},~\bfnm{Laurent}\binits{L.}},
  \bauthor{\bsnm{Jacques},~\bfnm{Julien}\binits{J.}} \AND
  \bauthor{\bsnm{Jollois},~\bfnm{Fran{\c c}ois-Xavier}\binits{F.-X.}}
(\byear{2017}).
\btitle{{The Functional Latent Block Model for the Co-Clustering of Electricity
  Consumption Curves}}.
\bjournal{{Journal of the Royal Statistical Society: Series C Applied
  Statistics}}.
\end{barticle}
\endbibitem

\bibitem[\protect\citeauthoryear{Bouveyron et~al.}{2021a}]{bouveyron:21}
\begin{barticle}[author]
\bauthor{\bsnm{Bouveyron},~\bfnm{Charles}\binits{C.}},
  \bauthor{\bsnm{Jacques},~\bfnm{Julien}\binits{J.}},
  \bauthor{\bsnm{Schmutz},~\bfnm{Amandine}\binits{A.}},
  \bauthor{\bsnm{Simoes},~\bfnm{Fanny}\binits{F.}} \AND
  \bauthor{\bsnm{Bottini},~\bfnm{Silvia}\binits{S.}}
(\byear{2021}a).
\btitle{{Co-Clustering of Multivariate Functional Data for the Analysis of Air
  Pollution in the South of France}}.
\bjournal{{Annals of Applied Statistics}}.
\end{barticle}
\endbibitem

\bibitem[\protect\citeauthoryear{Bouveyron et~al.}{2021b}]{bouveyron:2021_2}
\begin{barticle}[author]
\bauthor{\bsnm{Bouveyron},~\bfnm{Charles}\binits{C.}},
  \bauthor{\bsnm{Casa},~\bfnm{Alessandro}\binits{A.}},
  \bauthor{\bsnm{Erosheva},~\bfnm{Elena}\binits{E.}} \AND
  \bauthor{\bsnm{Menardi},~\bfnm{Giovanna}\binits{G.}}
(\byear{2021}b).
\btitle{{Co-clustering of Time-Dependent Data via the Shape Invariant Model}}.
\bjournal{{Journal of Classification}}.
\bdoi{10.1007/s00357-021-09402-8}
\end{barticle}
\endbibitem

\bibitem[\protect\citeauthoryear{Brault, Chiquet and
  L\'evy-Leduc}{2017}]{brault2017efficient}
\begin{barticle}[author]
\bauthor{\bsnm{Brault},~\bfnm{Vincent}\binits{V.}},
  \bauthor{\bsnm{Chiquet},~\bfnm{Julien}\binits{J.}} \AND
  \bauthor{\bsnm{L\'evy-Leduc},~\bfnm{C\'eline}\binits{C.}}
(\byear{2017}).
\btitle{Efficient block boundaries estimation in block-wise constant matrices:
  An application to HiC data}.
\bjournal{Electron. J. Statist.}
\bvolume{11}
\bpages{1570-1599}.
\bdoi{10.1214/17-EJS1270}
\end{barticle}
\endbibitem

\bibitem[\protect\citeauthoryear{Brault et~al.}{2018a}]{BraultOSL18}
\begin{barticle}[author]
\bauthor{\bsnm{Brault},~\bfnm{Vincent}\binits{V.}},
  \bauthor{\bsnm{Ouadah},~\bfnm{Sarah}\binits{S.}},
  \bauthor{\bsnm{Sansonnet},~\bfnm{Laure}\binits{L.}} \AND
  \bauthor{\bsnm{L{\'{e}}vy{-}Leduc},~\bfnm{C{\'{e}}line}\binits{C.}}
(\byear{2018}a).
\btitle{Nonparametric multiple change-point estimation for analyzing large Hi-C
  data matrices}.
\bjournal{J. Multivar. Anal.}
\bvolume{165}
\bpages{143--165}.
\bdoi{10.1016/j.jmva.2017.12.005}
\end{barticle}
\endbibitem

\bibitem[\protect\citeauthoryear{Brault
  et~al.}{2018b}]{brault2018nonparametric}
\begin{barticle}[author]
\bauthor{\bsnm{Brault},~\bfnm{Vincent}\binits{V.}},
  \bauthor{\bsnm{Ouadah},~\bfnm{Sarah}\binits{S.}},
  \bauthor{\bsnm{Sansonnet},~\bfnm{Laure}\binits{L.}} \AND
  \bauthor{\bsnm{L{\'e}vy-Leduc},~\bfnm{C{\'e}line}\binits{C.}}
(\byear{2018}b).
\btitle{Nonparametric multiple change-point estimation for analyzing large Hi-C
  data matrices}.
\bjournal{Journal of Multivariate Analysis}
\bvolume{165}
\bpages{143--165}.
\end{barticle}
\endbibitem

\bibitem[\protect\citeauthoryear{Brault et~al.}{2020}]{brault2020consistency}
\begin{barticle}[author]
\bauthor{\bsnm{Brault},~\bfnm{Vincent}\binits{V.}},
  \bauthor{\bsnm{Keribin},~\bfnm{Christine}\binits{C.}},
  \bauthor{\bsnm{Mariadassou},~\bfnm{Mahendra}\binits{M.}} \betal{et~al.}
(\byear{2020}).
\btitle{Consistency and asymptotic normality of Latent Block Model estimators}.
\bjournal{Electronic journal of statistics}
\bvolume{14}
\bpages{1234--1268}.
\end{barticle}
\endbibitem

\bibitem[\protect\citeauthoryear{Bugni et~al.}{2009}]{bugni2009}
\begin{barticle}[author]
\bauthor{\bsnm{Bugni},~\bfnm{Federico~A.}\binits{F.~A.}},
  \bauthor{\bsnm{Hall},~\bfnm{Peter}\binits{P.}},
  \bauthor{\bsnm{Horowitz},~\bfnm{Joel~L.}\binits{J.~L.}} \AND
  \bauthor{\bsnm{Neumann},~\bfnm{George~R.}\binits{G.~R.}}
(\byear{2009}).
\btitle{Goodness-of-fit tests for functional data}.
\bjournal{The Econometrics Journal}
\bvolume{12}
\bpages{S1--S18}.
\end{barticle}
\endbibitem

\bibitem[\protect\citeauthoryear{Chamroukhi}{2016}]{chamroukhi_2016}
\begin{barticle}[author]
\bauthor{\bsnm{Chamroukhi},~\bfnm{Faicel}\binits{F.}}
(\byear{2016}).
\btitle{{Piecewise Regression Mixture for Simultaneous Functional Data
  Clustering and Optimal Segmentation}}.
\bjournal{Journal of Classification}
\bvolume{33}
\bpages{374-411}.
\bdoi{10.1007/s00357-016-9212-8}
\end{barticle}
\endbibitem

\bibitem[\protect\citeauthoryear{Dempster, Laird and
  Rubin}{1977}]{dempster1977maximum}
\begin{barticle}[author]
\bauthor{\bsnm{Dempster},~\bfnm{Arthur~P}\binits{A.~P.}},
  \bauthor{\bsnm{Laird},~\bfnm{Nan~M}\binits{N.~M.}} \AND
  \bauthor{\bsnm{Rubin},~\bfnm{Donald~B}\binits{D.~B.}}
(\byear{1977}).
\btitle{Maximum likelihood from incomplete data via the EM algorithm}.
\bjournal{Journal of the Royal Statistical Society: Series B (Methodological)}
\bvolume{39}
\bpages{1--22}.
\end{barticle}
\endbibitem

\bibitem[\protect\citeauthoryear{Devijver}{2017}]{devijver2017}
\begin{barticle}[author]
\bauthor{\bsnm{Devijver},~\bfnm{Emilie}\binits{E.}}
(\byear{2017}).
\btitle{Model-Based Regression Clustering for High-Dimensional Data:
  Application to Functional Data}.
\bjournal{Advances in Data Analysis and Classification}
\bvolume{11}
\bpages{243–279}.
\end{barticle}
\endbibitem

\bibitem[\protect\citeauthoryear{Devijver, Goude and
  Poggi}{2020}]{devijver2020}
\begin{barticle}[author]
\bauthor{\bsnm{Devijver},~\bfnm{Emilie}\binits{E.}},
  \bauthor{\bsnm{Goude},~\bfnm{Yannig}\binits{Y.}} \AND
  \bauthor{\bsnm{Poggi},~\bfnm{Jean‐Michel}\binits{J.}}
(\byear{2020}).
\btitle{Clustering electricity consumers using high‐dimensional regression
  mixture models}.
\bjournal{Applied Stochastic Models in Business and Industry}
\bvolume{36}
\bpages{159-177}.
\end{barticle}
\endbibitem

\bibitem[\protect\citeauthoryear{Droesbeke, Saporta and
  Thomas-Agnan}{2013}]{droesbeke2013modeles}
\begin{bbook}[author]
\bauthor{\bsnm{Droesbeke},~\bfnm{Jean-Jacques}\binits{J.-J.}},
  \bauthor{\bsnm{Saporta},~\bfnm{Gilbert}\binits{G.}} \AND
  \bauthor{\bsnm{Thomas-Agnan},~\bfnm{Christine}\binits{C.}}
(\byear{2013}).
\btitle{Mod{\`e}les {\`a} variables latentes et mod{\`e}les de m{\'e}lange}.
\bpublisher{Editions TECHNIP}.
\end{bbook}
\endbibitem

\bibitem[\protect\citeauthoryear{Ferraty and Vieu}{2006}]{FerratyView2006}
\begin{bbook}[author]
\bauthor{\bsnm{Ferraty},~\bfnm{F.}\binits{F.}} \AND
  \bauthor{\bsnm{Vieu},~\bfnm{P.}\binits{P.}}
(\byear{2006}).
\btitle{Nonparametric Functional Data Analysis. Theory and Practice}.
\bseries{Springer Series in Statistics}.
\bpublisher{Springer}.
\end{bbook}
\endbibitem

\bibitem[\protect\citeauthoryear{Fontana, Tavoni and Vantini}{2019}]{Fontana}
\begin{barticle}[author]
\bauthor{\bsnm{Fontana},~\bfnm{Matteo}\binits{M.}},
  \bauthor{\bsnm{Tavoni},~\bfnm{Massimo}\binits{M.}} \AND
  \bauthor{\bsnm{Vantini},~\bfnm{Simone}\binits{S.}}
(\byear{2019}).
\btitle{Functional Data Analysis of high-frequency load curves reveals drivers
  of residential electricity consumption}.
\bjournal{PLOS ONE}
\bvolume{14}
\bpages{1-16}.
\bdoi{10.1371/journal.pone.0218702}
\end{barticle}
\endbibitem

\bibitem[\protect\citeauthoryear{Galvani et~al.}{2021}]{GALVANI2021107219}
\begin{barticle}[author]
\bauthor{\bsnm{Galvani},~\bfnm{Marta}\binits{M.}},
  \bauthor{\bsnm{Torti},~\bfnm{Agostino}\binits{A.}},
  \bauthor{\bsnm{Menafoglio},~\bfnm{Alessandra}\binits{A.}} \AND
  \bauthor{\bsnm{Vantini},~\bfnm{Simone}\binits{S.}}
(\byear{2021}).
\btitle{FunCC: A new bi-clustering algorithm for functional data with
  misalignment}.
\bjournal{Computational Statistics \& Data Analysis}
\bvolume{160}
\bpages{107219}.
\bdoi{https://doi.org/10.1016/j.csda.2021.107219}
\end{barticle}
\endbibitem

\bibitem[\protect\citeauthoryear{Giacofci et~al.}{2013}]{giacofci2013}
\begin{barticle}[author]
\bauthor{\bsnm{Giacofci},~\bfnm{M.}\binits{M.}},
  \bauthor{\bsnm{Lambert-Lacroix},~\bfnm{S.}\binits{S.}},
  \bauthor{\bsnm{Marot},~\bfnm{G.}\binits{G.}} \AND
  \bauthor{\bsnm{Picard},~\bfnm{F.}\binits{F.}}
(\byear{2013}).
\btitle{Wavelet-Based Clustering for Mixed-Effects Functional Models in High
  Dimension}.
\bjournal{Biometrics}
\bvolume{69}
\bpages{31--40}.
\end{barticle}
\endbibitem

\bibitem[\protect\citeauthoryear{Harchaoui and
  L{\'e}vy-Leduc}{2010}]{harchaoui2010multiple}
\begin{barticle}[author]
\bauthor{\bsnm{Harchaoui},~\bfnm{Za{\i}d}\binits{Z.}} \AND
  \bauthor{\bsnm{L{\'e}vy-Leduc},~\bfnm{C{\'e}line}\binits{C.}}
(\byear{2010}).
\btitle{Multiple change-point estimation with a total variation penalty}.
\bjournal{Journal of the American Statistical Association}
\bvolume{105}
\bpages{1480--1493}.
\end{barticle}
\endbibitem

\bibitem[\protect\citeauthoryear{Hubert and Arabie}{1985}]{hubert1985comparing}
\begin{barticle}[author]
\bauthor{\bsnm{Hubert},~\bfnm{L.}\binits{L.}} \AND
  \bauthor{\bsnm{Arabie},~\bfnm{P.}\binits{P.}}
(\byear{1985}).
\btitle{{Comparing partitions}}.
\bjournal{Journal of classification}
\bvolume{1}
\bpages{193--218}.
\end{barticle}
\endbibitem

\bibitem[\protect\citeauthoryear{Hébrail et~al.}{2010}]{Hebrail_2010}
\begin{barticle}[author]
\bauthor{\bsnm{Hébrail},~\bfnm{Georges}\binits{G.}},
  \bauthor{\bsnm{Hugueney},~\bfnm{Bernard}\binits{B.}},
  \bauthor{\bsnm{Lechevallier},~\bfnm{Yves}\binits{Y.}} \AND
  \bauthor{\bsnm{Rossi},~\bfnm{Fabrice}\binits{F.}}
(\byear{2010}).
\btitle{Exploratory analysis of functional data via clustering and optimal
  segmentation}.
\bjournal{Neurocomputing}
\bvolume{73}
\bpages{1125-1141}.
\bnote{Advances in Computational Intelligence and Learning}.
\end{barticle}
\endbibitem

\bibitem[\protect\citeauthoryear{Jacques and Preda}{2013}]{jjacques2013}
\begin{barticle}[author]
\bauthor{\bsnm{Jacques},~\bfnm{Julien}\binits{J.}} \AND
  \bauthor{\bsnm{Preda},~\bfnm{Cristian}\binits{C.}}
(\byear{2013}).
\btitle{Funclust: A curves clustering method using functional random variables
  density approximation}.
\bjournal{Neurocomputing}
\bvolume{112}
\bpages{164-171}.
\bnote{Advances in artificial neural networks, machine learning, and
  computational intelligence}.
\end{barticle}
\endbibitem

\bibitem[\protect\citeauthoryear{Jacques and Preda}{2014}]{jjacques2014}
\begin{barticle}[author]
\bauthor{\bsnm{Jacques},~\bfnm{Julien}\binits{J.}} \AND
  \bauthor{\bsnm{Preda},~\bfnm{Cristian}\binits{C.}}
(\byear{2014}).
\btitle{Functional Data Clustering: A Survey}.
\bjournal{Advances in Data Analysis and Classification}
\bvolume{8}
\bpages{231-255}.
\bdoi{10.1007/s11634-013-0158-y}
\end{barticle}
\endbibitem

\bibitem[\protect\citeauthoryear{James and Sugar}{2003}]{james2003_sparsefunc}
\begin{barticle}[author]
\bauthor{\bsnm{James},~\bfnm{Gareth~M}\binits{G.~M.}} \AND
  \bauthor{\bsnm{Sugar},~\bfnm{Catherine~A}\binits{C.~A.}}
(\byear{2003}).
\btitle{Clustering for Sparsely Sampled Functional Data}.
\bjournal{Journal of the American Statistical Association}
\bvolume{98}
\bpages{397-408}.
\bdoi{10.1198/016214503000189}
\end{barticle}
\endbibitem

\bibitem[\protect\citeauthoryear{Kay}{1993}]{kay1993fundamentals}
\begin{bbook}[author]
\bauthor{\bsnm{Kay},~\bfnm{Steven~M}\binits{S.~M.}}
(\byear{1993}).
\btitle{Fundamentals of statistical signal processing}.
\bpublisher{Prentice Hall PTR}.
\end{bbook}
\endbibitem

\bibitem[\protect\citeauthoryear{Kokoszka and Reimherr}{2017}]{kokoszka2017}
\begin{bbook}[author]
\bauthor{\bsnm{Kokoszka},~\bfnm{P.}\binits{P.}} \AND
  \bauthor{\bsnm{Reimherr},~\bfnm{M.}\binits{M.}}
(\byear{2017}).
\btitle{Introduction to Functional Data Analysis}.
\bseries{Chapman \& Hall / CRC numerical analysis and scientific computing}.
\bpublisher{CRC Press}.
\end{bbook}
\endbibitem

\bibitem[\protect\citeauthoryear{Lebarbier}{2005}]{lebarbier2005detecting}
\begin{barticle}[author]
\bauthor{\bsnm{Lebarbier},~\bfnm{{\'E}milie}\binits{{\'E}.}}
(\byear{2005}).
\btitle{Detecting multiple change-points in the mean of Gaussian process by
  model selection}.
\bjournal{Signal processing}
\bvolume{85}
\bpages{717--736}.
\end{barticle}
\endbibitem

\bibitem[\protect\citeauthoryear{Li, Qiu and Xu}{2022}]{Li2022}
\begin{barticle}[author]
\bauthor{\bsnm{Li},~\bfnm{Yehua}\binits{Y.}},
  \bauthor{\bsnm{Qiu},~\bfnm{Yumou}\binits{Y.}} \AND
  \bauthor{\bsnm{Xu},~\bfnm{Yuhang}\binits{Y.}}
(\byear{2022}).
\btitle{From multivariate to functional data analysis: Fundamentals, recent
  developments, and emerging areas}.
\bjournal{Journal of Multivariate Analysis}
\bvolume{188}
\bpages{104806}.
\bnote{50th Anniversary Jubilee Edition}.
\bdoi{https://doi.org/10.1016/j.jmva.2021.104806}
\end{barticle}
\endbibitem

\bibitem[\protect\citeauthoryear{Liu and Yang}{2009}]{Liu2009}
\begin{barticle}[author]
\bauthor{\bsnm{Liu},~\bfnm{Xueli}\binits{X.}} \AND
  \bauthor{\bsnm{Yang},~\bfnm{Mark C.~K.}\binits{M.~C.~K.}}
(\byear{2009}).
\btitle{Simultaneous curve registration and clustering for functional data}.
\bjournal{Computational Statistics \& Data Analysis}
\bvolume{53}
\bpages{1361-1376}.
\end{barticle}
\endbibitem

\bibitem[\protect\citeauthoryear{Maidstone et~al.}{2017}]{maidstone2017optimal}
\begin{barticle}[author]
\bauthor{\bsnm{Maidstone},~\bfnm{Robert}\binits{R.}},
  \bauthor{\bsnm{Hocking},~\bfnm{Toby}\binits{T.}},
  \bauthor{\bsnm{Rigaill},~\bfnm{Guillem}\binits{G.}} \AND
  \bauthor{\bsnm{Fearnhead},~\bfnm{Paul}\binits{P.}}
(\byear{2017}).
\btitle{On optimal multiple changepoint algorithms for large data}.
\bjournal{Statistics and computing}
\bvolume{27}
\bpages{519--533}.
\end{barticle}
\endbibitem

\bibitem[\protect\citeauthoryear{Mallat}{1989}]{Mallat2}
\begin{barticle}[author]
\bauthor{\bsnm{Mallat},~\bfnm{S.}\binits{S.}}
(\byear{1989}).
\btitle{A theory for multiresolution signal decomposition: the wavelet
  representation}.
\bjournal{IEEE Transactions on Pattern Analysis and Machine Intelligence}
\bvolume{11}
\bpages{674--693}.
\end{barticle}
\endbibitem

\bibitem[\protect\citeauthoryear{Mallat}{1999}]{Mallat}
\begin{bbook}[author]
\bauthor{\bsnm{Mallat},~\bfnm{S.}\binits{S.}}
(\byear{1999}).
\btitle{{A wavelet tour of signal processing.}}
\bpublisher{Academic Press}.
\end{bbook}
\endbibitem

\bibitem[\protect\citeauthoryear{Mariadassou and
  Tabouy}{2020}]{mariadassou2020consistency}
\begin{barticle}[author]
\bauthor{\bsnm{Mariadassou},~\bfnm{Mahendra}\binits{M.}} \AND
  \bauthor{\bsnm{Tabouy},~\bfnm{Timoth{\'e}e}\binits{T.}}
(\byear{2020}).
\btitle{Consistency and asymptotic normality of stochastic block models
  estimators from sampled data}.
\end{barticle}
\endbibitem

\bibitem[\protect\citeauthoryear{Maturo and Verde}{2023}]{Maturo2023}
\begin{barticle}[author]
\bauthor{\bsnm{Maturo},~\bfnm{F.}\binits{F.}} \AND
  \bauthor{\bsnm{Verde},~\bfnm{R.}\binits{R.}}
(\byear{2023}).
\btitle{Supervised classification of curves via a combined use of functional
  data analysis and tree-based methods}.
\bjournal{Comput Stat}
\bvolume{38}
\bpages{419-–459}.
\end{barticle}
\endbibitem

\bibitem[\protect\citeauthoryear{Mengersen, Robert and
  Titterington}{2011}]{mengersen2011mixtures}
\begin{barticle}[author]
\bauthor{\bsnm{Mengersen},~\bfnm{Kerrie~L}\binits{K.~L.}},
  \bauthor{\bsnm{Robert},~\bfnm{Christian}\binits{C.}} \AND
  \bauthor{\bsnm{Titterington},~\bfnm{Mike}\binits{M.}}
(\byear{2011}).
\btitle{Mixtures: estimation and applications}.
\end{barticle}
\endbibitem

\bibitem[\protect\citeauthoryear{Mersmann}{2021}]{olaf2021microbenchmark}
\begin{barticle}[author]
\bauthor{\bsnm{Mersmann},~\bfnm{Olaf}\binits{O.}}
(\byear{2021}).
\btitle{microbenchmark: Accurate Timing Functions}.
\bnote{R package version 1.4.9}.
\end{barticle}
\endbibitem

\bibitem[\protect\citeauthoryear{Misiti et~al.}{2004}]{Poggi}
\begin{bmanual}[author]
\bauthor{\bsnm{Misiti},~\bfnm{M.}\binits{M.}},
  \bauthor{\bsnm{Misiti},~\bfnm{Y.}\binits{Y.}},
  \bauthor{\bsnm{Oppenheim},~\bfnm{G.}\binits{G.}} \AND
  \bauthor{\bsnm{Poggi},~\bfnm{J-M.}\binits{J.-M.}}
(\byear{2004}).
\btitle{Matlab Wavelet Toolbox User's Guide. Version 3.}
\bpublisher{The Mathworks, Inc.},
\baddress{Natick, MA.}
\end{bmanual}
\endbibitem

\bibitem[\protect\citeauthoryear{Ramsay and Silverman}{2005}]{ramsay2005}
\begin{bbook}[author]
\bauthor{\bsnm{Ramsay},~\bfnm{J.~O.}\binits{J.~O.}} \AND
  \bauthor{\bsnm{Silverman},~\bfnm{B.~W.}\binits{B.~W.}}
(\byear{2005}).
\btitle{{Functional Data Analysis}}.
\bpublisher{Springer}.
\end{bbook}
\endbibitem

\bibitem[\protect\citeauthoryear{Rigaill}{2015}]{rigaill2015pruned}
\begin{barticle}[author]
\bauthor{\bsnm{Rigaill},~\bfnm{Guillem}\binits{G.}}
(\byear{2015}).
\btitle{A pruned dynamic programming algorithm to recover the best
  segmentations with 1 to $K_{\text{max}}$ change-points.}
\bjournal{Journal de la Soci{\'e}t{\'e} Fran{\c{c}}aise de Statistique}
\bvolume{156}
\bpages{180--205}.
\end{barticle}
\endbibitem

\bibitem[\protect\citeauthoryear{Robert}{2021}]{robert2021bikm1}
\begin{bmanual}[author]
\bauthor{\bsnm{Robert},~\bfnm{Valerie}\binits{V.}}
(\byear{2021}).
\btitle{bikm1: Co-Clustering Adjusted Rand Index and Bikm1 Procedure for
  Contingency and Binary Data-Sets}
\bnote{R package version 1.1.0}.
\end{bmanual}
\endbibitem

\bibitem[\protect\citeauthoryear{Robert, Vasseur and Brault}{2021}]{robert2021}
\begin{barticle}[author]
\bauthor{\bsnm{Robert},~\bfnm{Val{\'e}rie}\binits{V.}},
  \bauthor{\bsnm{Vasseur},~\bfnm{Yann}\binits{Y.}} \AND
  \bauthor{\bsnm{Brault},~\bfnm{Vincent}\binits{V.}}
(\byear{2021}).
\btitle{Comparing high-dimensional partitions with the co-clustering adjusted
  rand index}.
\bjournal{Journal of Classification}
\bvolume{38}
\bpages{158--186}.
\end{barticle}
\endbibitem

\bibitem[\protect\citeauthoryear{Samé and Govaert}{2012}]{Same_2012}
\begin{binproceedings}[author]
\bauthor{\bsnm{Samé},~\bfnm{Allou}\binits{A.}} \AND
  \bauthor{\bsnm{Govaert},~\bfnm{Gérard}\binits{G.}}
(\byear{2012}).
\btitle{Online Time Series Segmentation Using Temporal Mixture Models and
  Bayesian Model Selection}.
In \bbooktitle{2012 11th International Conference on Machine Learning and
  Applications}
\bvolume{1}
\bpages{602-605}.
\bdoi{10.1109/ICMLA.2012.111}
\end{binproceedings}
\endbibitem

\bibitem[\protect\citeauthoryear{Samé et~al.}{2011}]{same2011}
\begin{barticle}[author]
\bauthor{\bsnm{Samé},~\bfnm{Allou}\binits{A.}},
  \bauthor{\bsnm{Chamroukhi},~\bfnm{Faicel}\binits{F.}},
  \bauthor{\bsnm{Govaert},~\bfnm{Gérard}\binits{G.}} \AND
  \bauthor{\bsnm{Aknin},~\bfnm{Patrice}\binits{P.}}
(\byear{2011}).
\btitle{{Model-based clustering and segmentation of time series with changes in
  regime}}.
\bjournal{Advances in Data Analysis and Classification}
\bvolume{5}
\bpages{301-321}.
\end{barticle}
\endbibitem

\bibitem[\protect\citeauthoryear{Schwarz}{1978}]{Schwarz_1978}
\begin{barticle}[author]
\bauthor{\bsnm{Schwarz},~\bfnm{Gideon}\binits{G.}}
(\byear{1978}).
\btitle{Estimating the Dimension of a Model}.
\bjournal{The Annals of Statistics}
\bvolume{6}
\bpages{461--464}.
\end{barticle}
\endbibitem

\bibitem[\protect\citeauthoryear{Stephens}{2000}]{stephens2000dealing}
\begin{barticle}[author]
\bauthor{\bsnm{Stephens},~\bfnm{Matthew}\binits{M.}}
(\byear{2000}).
\btitle{Dealing with label switching in mixture models}.
\bjournal{Journal of the Royal Statistical Society: Series B (Statistical
  Methodology)}
\bvolume{62}
\bpages{795--809}.
\end{barticle}
\endbibitem

\bibitem[\protect\citeauthoryear{Wang, Chiou and M\"{u}ller}{2016}]{wang2016}
\begin{barticle}[author]
\bauthor{\bsnm{Wang},~\bfnm{Jane-Ling}\binits{J.-L.}},
  \bauthor{\bsnm{Chiou},~\bfnm{Jeng-Min}\binits{J.-M.}} \AND
  \bauthor{\bsnm{M\"{u}ller},~\bfnm{Hans-Georg}\binits{H.-G.}}
(\byear{2016}).
\btitle{Functional Data Analysis}.
\bjournal{Annual Review of Statistics and Its Application}
\bvolume{3}
\bpages{257-295}.
\bdoi{10.1146/annurev-statistics-041715-033624}
\end{barticle}
\endbibitem

\bibitem[\protect\citeauthoryear{Zhang and Siegmund}{2007}]{zhang2007modified}
\begin{barticle}[author]
\bauthor{\bsnm{Zhang},~\bfnm{Nancy~R}\binits{N.~R.}} \AND
  \bauthor{\bsnm{Siegmund},~\bfnm{David~O}\binits{D.~O.}}
(\byear{2007}).
\btitle{A modified Bayes information criterion with applications to the
  analysis of comparative genomic hybridization data}.
\bjournal{Biometrics}
\bvolume{63}
\bpages{22--32}.
\end{barticle}
\endbibitem

\end{thebibliography}

\end{document}